\documentclass[11pt]{amsart}

\usepackage[T1]{fontenc}
\usepackage[utf8]{inputenc}
\usepackage{amsmath,amssymb,amsfonts,amsthm,mathtools}
\usepackage{mathrsfs}
\usepackage{enumitem}
\usepackage{booktabs}
\usepackage{microtype}
\usepackage{xcolor}
\usepackage{array}
\usepackage{graphicx}
\usepackage{tikz}
\usepackage{tikz-cd}
\usetikzlibrary{arrows.meta,calc,positioning}
\usepackage{xurl}
\usepackage{hyperref}
\usepackage[margin=1.15in]{geometry}
\usepackage[section]{placeins}
\usepackage{float}

\hypersetup{
  colorlinks=true,
  linkcolor=blue!50!black,
  citecolor=blue!50!black,
  urlcolor=blue!50!black
}

\newtheorem{theorem}{Theorem}[section]
\newtheorem{proposition}[theorem]{Proposition}
\newtheorem{lemma}[theorem]{Lemma}
\newtheorem{corollary}[theorem]{Corollary}
\newtheorem{definition}[theorem]{Definition}

\newtheorem{principle}[theorem]{Principle}
\newtheorem{problem}[theorem]{Problem}
\theoremstyle{remark}
\newtheorem{remark}[theorem]{Remark}
\newtheorem{example}[theorem]{Example}

\newcommand{\cA}{\mathcal A}

\newcommand{\cC}{\mathcal C}
\newcommand{\cD}{\mathcal D}
\newcommand{\cE}{\mathcal E}
\newcommand{\cF}{\mathcal F}
\newcommand{\cH}{\mathcal H}

\newcommand{\cS}{\mathcal S}
\newcommand{\cV}{\mathcal V}
\newcommand{\bC}{\mathbb C}
\newcommand{\bP}{\mathbb P}
\newcommand{\bZ}{\mathbb Z}
\newcommand{\one}{\mathbf 1}
\newcommand{\Hom}{\operatorname{Hom}}
\newcommand{\End}{\operatorname{End}}

\newcommand{\id}{\operatorname{id}}
\newcommand{\Prob}{\operatorname{Pr}}
\newcommand{\Synd}{\operatorname{Synd}}
\newcommand{\Foot}{\operatorname{fp}}
\newcommand{\supp}{\operatorname{supp}}
\newcommand{\ev}{\operatorname{ev}}

\newcommand{\Pic}{\operatorname{Pic}}

\newcommand{\Jac}{\operatorname{Jac}}
\newcommand{\Spec}{\operatorname{Spec}}
\newcommand{\Irr}{\operatorname{Irr}}

\title[A diagrammatic field theory of quantum error correction]{A diagrammatic field theory of quantum error correction}
\author{Steven Rayan}
\address{Centre for Quantum Topology and Its Applications (quanTA) and Department of Mathematics and Statistics, University of Saskatchewan}
\email{rayan@math.usask.ca}
\date{\today}

\begin{document}

\begin{abstract}
We develop a field-theoretic framework for quantum error correction centred on fusion-space codes in unitary fusion categories.  An admissible cluster of tensor factors determines total-charge sectors and orthogonal \emph{footprint projectors}, which record locally visible field-theoretic data left by an error history.  The central distinction is between diagnostic footprint algebras and \emph{syndrome-admissible} commuting algebras: the latter may be measured without revealing logical information and resolve the chosen error representatives into measured sectors.  For such algebras, exact correctability is equivalent to a fibrewise Knill--Laflamme condition.  Under a contractible-vacuum locality hypothesis, closed neutral composites give a categorical sufficient criterion for scalar action on the code.

The Ising theory supplies complementary finite examples.  For four \(\sigma\)-punctures, pair-charge footprints become Pauli \(Z\)- and \(X\)-type diagnostics after recoupling, and the \(F\)- and braid data give an exact one-qubit Clifford shadow.  For six \(\sigma\)-punctures, a proper two-dimensional code admits a syndrome-admissible pair-charge measurement and exact recovery from an explicit Majorana bilinear error.  A second bilinear has the same measured footprint but differs by a logical bit flip, giving a concrete nontrivial footprint fibre and an actual decoding ambiguity.  We also formulate conformal-block likelihood data and compute geometry-dependent Ising four-point weights.

For growing families, we prove a conditional Peierls-type threshold theorem.  Bounded connected-region growth, local stochastic noise, local neutralizability of small residual components, and component-wise decoder balance imply \(\Pr_L(\mathrm{fail})\le C|\Omega_L|e^{-cL}\) below a nonzero constant error rate.  Its conclusion is conditional on these hypotheses and does not extend automatically to arbitrary TQFT/CFT codes.  We conclude with representation-theoretic and algebro-geometric directions involving tube and Hopf algebras, Yangian-type structures, Higgs bundles, spectral curves, Jacobians, and abelian varieties.\end{abstract}

\maketitle

\tableofcontents

\section{Introduction}

Quantum error correction (QEC) is usually introduced through subspaces, stabilizers, parity checks, and recovery maps.  This language is indispensable.  It underlies the Knill--Laflamme conditions \cite{KnillLaflamme1997}, the stabilizer formalism \cite{Gottesman1997}, CSS codes \cite{CalderbankShor1996,Steane1996}, threshold theorems, and the engineering of scalable fault tolerance.  Many geometrically defined codes, however, also carry field-theoretic structure.  Toric and surface codes reflect topological gauge theory, while string-net codes are lattice fixed points built from fusion-category data.  Anyonic computation, Turaev--Viro codes, condensation constructions, and holographic tensor networks provide further instances in which field theory enters directly into the organization of the code.

We make one part of this programme precise by treating a quantum code as a state space assigned by an extended topological or conformal field theory, or by a categorical model of one, to a spatial datum with boundaries, punctures, defects, and a chosen realization.  The field theory does not by itself implement a device.  It supplies ideal state spaces, sectors, local boundary data, defect moves, and amplitude weights.  A physical realization \(\Gamma\), a measurement model, and a recovery protocol remain separate operational data.

We call the locally visible boundary datum left by an error history its \emph{footprint}.  The history may be a local operator, a field insertion, a defect segment, a braid, or a string-net process.  Relative to a specified measurement model, a syndrome is a measured footprint.  A measured syndrome therefore records local field-theoretic data induced by the error.  We use ``footprint'' for the underlying boundary datum whose measurement, when syndrome-admissible, becomes QEC syndrome data; established syndrome terminology in stabilizer, homological, and anyonic codes remains unchanged.  This use is also unrelated to the common engineering use of \emph{footprint} for qubit, area, space, or hardware overhead.

The usual stabilizer question is
\[
\boxed{\begin{minipage}{0.86\linewidth}
\centering
\text{Which Pauli error is compatible with the measured stabilizer syndrome?}
\end{minipage}}
\]
The field-theoretic version asks instead
\[
\boxed{\begin{minipage}{0.86\linewidth}
\centering
Among all defect histories compatible with an observed footprint, which histories dominate the relevant field-theoretic weight?
\end{minipage}}
\]
Here ``dominate'' depends on the model: it may refer to posterior probability, amplitude norm, effective action, state-sum weight, or likelihood under a physical noise model supplemented by field-theoretic data.  In a topological theory the weight may depend only on the topological class of a defect network.  In a conformal theory it may also depend on cross-ratios, scaling dimensions, modular parameters, and conformal blocks.

The mathematical core is finite and semisimple.  In a unitary fusion category, an admissible cluster of tensor factors determines a total-charge decomposition of a fusion space, and the associated orthogonal projectors are the footprint projectors.  Compatible projectors generate commutative measurement algebras, whereas overlapping clusters may generate noncommuting diagnostic algebras.  A QEC syndrome algebra is more restrictive: it must be syndrome-admissible, so that the no-error measurement reveals no logical information and the chosen error representatives are resolved by measured sectors.  Under this hypothesis, exact correctability is equivalent to a fibrewise Knill--Laflamme condition.  The phrase ``locally visible boundary datum'' expresses the broad field-theoretic intuition.  The theorem-level definition used below is the finite semisimple categorical one in terms of total-charge sectors and their orthogonal projectors.

\subsection{Contributions and structure of the article}\label{subsec:contributions}

The scope should be set against several established strands of categorical and field-theoretic QEC.  Surface codes already show how topology can protect logical information: local errors create local excitations, whereas noncontractible strings can act logically.  Kitaev's anyonic models provide a field-theoretic setting for this picture \cite{Kitaev2003,Kitaev2006}, and Dennis--Kitaev--Landahl--Preskill relate topological memory to the statistical mechanics of error chains \cite{DennisKitaevLandahlPreskill2002}.  Levin--Wen models and Turaev--Viro codes replace the abelian toric-code input by general fusion-category data \cite{LevinWen2005,KoenigKuperbergReichardt2010,SchotteZhuBurgelmanVerstraete2022}.  Nonabelian error correction has also been studied directly through decoding and simulation of anyonic histories, including Fibonacci, Ising, and broader nonabelian settings \cite{BurtonBrellFlammia2017,WoottonHutter2016,DauphinaisPoulin2017}.  In a different operational role, measurement-only topological quantum computation uses topological-charge measurements deliberately as computational primitives \cite{BondersonFreedmanNayak2008}; this is a particularly clear precedent for the diagnostic/syndrome distinction made below.  The distinction also has close relatives in operator and subsystem QEC, where one separates preserved logical information from gauge or otherwise operationally irrelevant degrees of freedom \cite{KribsLaflammePoulin2005,BlumeKohoutNgPoulinViola2010}.  Our use of categorical multiplicity spaces should be read against that background.  Recent work has further established fault-tolerant anyonic computation under explicit active-correction schemes \cite{LyonsBrown2026}.  Anyon condensation and spacetime domain walls have been used to describe fault-tolerant gates and dynamically driven codes \cite{KesselringEtAl2024}, while holographic codes interpret tensor-network encoders through bulk/boundary reconstruction \cite{PastawskiYoshidaHarlowPreskill2015}.  Here we isolate a common local datum that appears across these settings.

The main contribution is organizational.  Total-charge decompositions and Knill--Laflamme equations conditioned on measured syndromes are established ingredients.  We organize the code around one field-theoretic datum: the footprint of an error.  This separates a latent error history from the local trace available to measurement.  Stabilizer syndromes, boundaries of error chains, and total-charge or fusion-channel measurements are established constructions; what is new here is their treatment as instances of a common intermediate datum, together with an explicit criterion distinguishing diagnostic measurements from syndrome-admissible ones.  The later Peierls hypotheses are intended as a portable checklist for labelled, non-chain models; the counting mechanism itself is familiar from surface-code arguments.  We then ask how exact recovery depends on the measured footprint sectors.  The main object is the syndrome-admissible footprint algebra.  For fusion-space codes in a unitary fusion category, we construct total-charge footprint projectors, distinguish diagnostic algebras from syndrome algebras, and prove a fibrewise Knill--Laflamme theorem.  In the contractible vacuum sector, the same condition is expressed by scalarity of closed neutral composites in \(\End_{\cC}(\one)\cong\bC\).

The same construction has diagrammatic and analytic forms.  We treat the ZX-calculus as the stabilizer shadow of a selected defect calculus, and conformal blocks as sources of geometry-dependent likelihood factors once a Hermitian conformal-block datum has been fixed.  The Ising calculations separate two roles that are easy to conflate.  The four-puncture qubit makes the diagnostic role explicit: its fusion-space encoding, recoupling matrices, braid matrices, complementary footprint observables, circuit shadow, and conformal-block likelihood weights are all computable.  A six-puncture subspace then supplies a genuine syndrome-admissible example in which a pair-charge measurement detects a specified error and a conditional recovery restores an arbitrary logical state; enlarging the error family produces two same-footprint representatives whose residual difference is logical.  Here ``diagrammatic'' has three related meanings.  At the stabilizer level, the ZX-calculus furnishes a graphical language for qubit processes, complementary observables, and surface-code lattice surgery \cite{CoeckeKissinger2017,deBeaudrapHorsman2020,ChancellorEtAl2016,BombinLitinskiNickersonPastawskiRoberts2023}.  At the topological level, string-net, Wilson-line, and defect-network diagrams encode the local and global structure of topological codes.  Finally, at the conformal level, the same diagrams acquire analytic weights through conformal blocks and correlation functions.  Writing TQFT and CFT for topological and conformal field theory, respectively, the schematic progression is
\begin{equation}\label{eq:containment}
\text{ZX diagrams}
\rightsquigarrow
\text{defect/string-net diagrams}
\rightsquigarrow
\text{extended TQFT/CFT amplitudes}.
\end{equation}
Equivalently, ZX-calculus is the stabilizer or Pauli shadow of the more general field-theoretic diagrammatics in those models where an appropriate qubit stabilizer sector has been selected. (For some prior work on the interface between TQFTs and quantum computing, at the foundational level of gate synthesis rather than error correction but also with a view to lifting the ZX-calculus, we point the reader to \cite{AzamRayan2024}.)

Beyond these finite-dimensional statements, we prove a conditional threshold theorem for families of footprint codes satisfying explicit local hypotheses.  The theorem has a conditional scope: if a growing family has bounded local geometry, local stochastic noise, neutralizability of small closed residual histories, and a decoder whose residual components contain a definite fraction of actual faults, then logical failure is exponentially suppressed below a nonzero noise strength.  This connects the categorical formalism to the architecture-level question of threshold behaviour.  We then work out two Ising models.  Four Ising \(\sigma\)-insertions with total vacuum charge define a two-dimensional logical space
\[
\cH_L=\Hom(\one,\sigma^{\otimes 4})\cong\bC^2.
\]
The nontrivial \(F\)-move is the Hadamard matrix, the pair-charge observables become complementary diagnostics, and the two four-point conformal blocks determine geometry-dependent relative weights for the competing fusion channels.  For six \(\sigma\)-insertions, the vacuum fusion space is four-dimensional.  Choosing a two-dimensional fixed-pair-charge subspace gives redundancy: the \((5,6)\)-charge is constant on the code and can therefore serve as a syndrome.  An explicit parity-transfer error moves the code isometrically into the orthogonal measured sector, so the footprint measurement and conditional recovery can be verified exactly.

The later sections turn to larger architectures and to representation-theoretic and geometric structures that may refine decoding.  A first geometric layer already appears on the Riemann sphere \(\bP^1\): the conformal-block example varies over the moduli space of marked spheres, with the planar anyon configuration appearing as one fibrewise picture.  Parallel transport is controlled by flat, generally meromorphic connections, notably the Knizhnik--Zamolodchikov connection, whose monodromy encodes braid and fusion data.  In semiclassical or algebro-geometric regimes, related structures are described by Higgs fields and spectral curves over \(\bP^1\).  This leads to a concrete question for decoding: how should insertion positions, collision divisors, monodromy, and spectral data enter the relative weights of compatible error histories?

The paper is organized as follows.  Section 2 recalls ordinary QEC and the homological prototype of footprint decoding.  Section 3 gives the categorical core: field-theoretic code data, categorical footprint projectors, the footprint projector algebra, and the fibrewise Knill--Laflamme theorem.  Sections 4--6 treat defect-network errors, the ZX shadow, and conformal-block likelihood weights.  Section 7 works through the four-puncture Ising diagnostic example and a six-puncture syndrome-admissible code with exact recovery.  Section 8 turns to scalable architectures and proves the Peierls-type threshold theorem.  Sections 9 and 10 develop representation-theoretic and algebro-geometric directions, and Section 11 collects concrete open problems.  Readers primarily interested in the formal QEC results may pass from the categorical setup and Theorem~\ref{thm:footprint-algebra-KL} in Section~3 directly to the Ising examples in Section~7 and the threshold criterion in Section~8.

\section{Classical quantum error correction and its topological enlargement}

We begin by recalling the elementary algebraic skeleton of quantum error correction.  A quantum code is a subspace \(\cH_L\subset \cH_P\) of a physical Hilbert space.  Let \(P\) denote the orthogonal projection onto \(\cH_L\).  A finite set of errors \(\cE=\{E_a\}\subset \End(\cH_P)\) is correctable precisely when the Knill--Laflamme condition holds:
\begin{equation}\label{eq:KL}
P E_a^\dagger E_b P = \lambda_{ab}P
\end{equation}
for some Hermitian matrix \((\lambda_{ab})\).  In words, the code corrects a family of errors if the code cannot internally distinguish the composite error \(E_a^\dagger E_b\) except through a scalar.  This condition is the algebraic heart of exact quantum error correction.  More general subsystem and operator-algebraic formulations place the same issue in a broader information-preserving framework \cite{KribsLaflammePoulin2005}; the present paper remains in the subspace setting except where gauge or multiplicity degrees of freedom are explicitly noted.

The topological version of \eqref{eq:KL} replaces the phrase ``acts as a scalar on the code space'' by a geometric mechanism.  Local operators supported in a contractible region act trivially, or create locally detectable excitations.  Logical operators are represented by operators with nontrivial topology: loops around handles, strings connecting suitable boundaries, or defect networks whose endpoints cannot be locally neutralized.  The toric code is the archetype.  On a closed surface \(\Sigma\), its logical operators are governed by the homology of \(\Sigma\), and local error strings are detected by their endpoints.  The stabilizer syndrome records the boundary of an error chain, while a logical error corresponds to a cycle with nontrivial homology class.

This picture can be expressed schematically as
\[
\begin{array}{ccl}
\text{physical qubits} &\rightsquigarrow& \text{edges or cells of a lattice}\\
\text{stabilizer checks} &\rightsquigarrow& \text{local flatness or charge constraints}\\
\text{syndrome} &\rightsquigarrow& \text{boundary of an error chain}\\
\text{logical operator} &\rightsquigarrow& \text{nontrivial cycle or Wilson line}
\end{array}
\]
The topological memory of Dennis--Kitaev--Landahl--Preskill can be read in precisely this way: error correction is an inference problem over chains with fixed boundary, modulo the homological distinction between correctable and logical chains \cite{DennisKitaevLandahlPreskill2002}.

Three layers will be kept separate throughout.  The \emph{algebraic} layer consists of the projection \(P\), the error operators \(E_a\), and the condition \(P E_a^\dagger E_bP=\lambda_{ab}P\).  The \emph{measurement} layer extracts classical information that distinguishes some error classes without revealing the encoded state.  The \emph{inference} layer chooses a recovery representative among errors compatible with the measured data.  In a stabilizer code these are the code projector, the syndrome map, and the decoder.  In a topological code they are the ground-state projector, the boundary or endpoint data of an error chain, and the choice of a compatible homology class.

A syndrome is a deliberately coarse local witness rather than a complete microscopic record of the error.  A more detailed record would usually carry unnecessary information and could disturb the encoded data.  In a surface code, an error chain \(c\) and another error chain \(c+c_0\), where \(c_0\) is a cycle, have the same boundary.  They therefore have the same local syndrome, but they may differ by a logical operator if \(c_0\) is homologically nontrivial.  The decoder consequently solves an equivalence problem: it selects a representative in a class of chains with a fixed local boundary.  The proposed footprint terminology is designed to preserve exactly this distinction in the field-theoretic setting.  A footprint records the local boundary datum left by an error history.

One may write the surface-code situation, at the coarsest level, as
\[
  \Foot(c)=\partial c,
  \qquad
  c\sim c' \quad\Longleftrightarrow\quad \partial c=\partial c'.
\]
The logical ambiguity is then measured not by the footprint but by the class of \(c-c'\) in homology.  This elementary formula is the prototype for the more general constructions below.  The boundary operator \(\partial\) will be replaced by a local field-theoretic boundary map, homology classes will be replaced by topological or categorical classes of defect networks, and stochastic chain weights will be replaced, or supplemented, by topological and conformal amplitudes.

\begin{proposition}[Homological prototype of footprint decoding]\label{prop:homological-prototype}
Consider a CSS surface-code model in which a family of Pauli \(X\)-type errors is represented by \(1\)-chains \(c\in C_1(\Gamma;\bZ_2)\), and in which the measured \(Z\)-type syndrome is the boundary \(\partial c\in C_0(\Gamma;\bZ_2)\).  Let \(P\) be the code projector.  If two chains \(c,c'\) have different boundaries, then they lie in different ideal syndrome sectors.  If they have the same boundary, then \(c+c'\) is a cycle.  The composite Pauli error \(X(c)X(c')=X(c+c')\) is correctable on the code precisely when the homology class of \(c+c'\) acts trivially on the encoded space, equivalently when the cycle is a stabilizer boundary in the chosen surface-code realization.
\end{proposition}

\begin{proof}
Let us write the CSS syndrome calculation explicitly.  The chain
\(c=\sum_e c_e e\) specifies the Pauli operator \(X(c)=\prod_e X_e^{c_e}\).  A vertex \(Z\)-check anticommutes with precisely those incident \(X_e\)'s which occur in the support of \(c\).  Hence the measured sign at a vertex is the parity of the incident coefficients of \(c\), which is exactly the coefficient of \(\partial c\) at that vertex.  Thus the ideal syndrome of \(X(c)\) is \(\partial c\).  If \(\partial c\neq\partial c'\), the corresponding error spaces are separated by orthogonal eigenspaces of the commuting check algebra, so they are different ideal syndrome sectors.

If \(\partial c=\partial c'\), then \(\partial(c+c')=0\), since the chains are over \(\bZ_2\).  The two errors therefore differ by a cycle.  The Knill--Laflamme composite is
\[
  P X(c)^\dagger X(c')P=P X(c+c')P,
\]
because Pauli \(X\)-operators are self-adjoint and the product is addition of chains mod \(2\).  A closed chain which is a stabilizer boundary acts trivially on the code space; in the usual surface-code realization it is a product of plaquette-type stabilizers.  A closed chain representing a nontrivial homology class acts instead as a logical operator.  Consequently the boundary \(\partial c\) is the local footprint, and the only residual ambiguity inside a footprint fibre is the homology class of the difference cycle.
\end{proof}

The same pattern persists beyond the abelian toric code.  In a Levin--Wen string-net model associated to a unitary fusion category \(\cC\), the local Hilbert space is spanned by labels on edges, vertex constraints enforce admissible fusion, and plaquette terms impose a topological vacuum constraint.  Excitations are no longer merely \(e\)- and \(m\)-type defects.  They are described by the Drinfeld center \(Z(\cC)\), and errors are better understood as string operators or local violations carrying nontrivial anyon charge.  The passage from the toric code to a general string-net code is therefore already a passage from binary stabilizer syndrome to categorical fusion syndrome.

The resulting relation will recur throughout the paper:
\[
\boxed{\begin{minipage}{0.86\linewidth}
\centering
For topological stabilizer codes admitting an anyonic description, stabilizer syndromes form an abelian shadow of fusion-sector data.
\end{minipage}}
\]
In the abelian topological case, the fusion data are often governed by a finite abelian group and the syndrome is a list of violated checks.  In the nonabelian case, the syndrome can include topological charge labels and fusion-channel data.  This statement is deliberately scoped to topological code families: a general stabilizer code need not possess a natural fusion interpretation.  The close relation between Abelian anyon theories and Pauli topological subsystem codes has been made precise in complementary directions \cite{EllisonEtAl2023}.  In a conformal enhancement, the same fusion-channel alternatives may carry geometry-dependent weights.

\section{Field-theoretic code data}

We begin with a broad code datum that can specialize to surface codes, string-net codes, conformal-block encodings, and defect-network models, including realizations obtained by discretizing a continuum field theory or by choosing a tensor network.  The background formalism is the standard theory of quantum groups, modular categories, and three-dimensional state-sum or surgery TQFTs; see, for example, \cite{ReshetikhinTuraev1991,Turaev1994,BakalovKirillov2001,Kassel1995,TuraevViro1992}.

\begin{definition}[Field-theoretic quantum code datum]\label{def:FTcode}
A \emph{field-theoretic quantum code datum} is a tuple
\begin{equation}\label{eq:code-datum}
\mathfrak Q=(\cF,\Sigma,\cD,\Gamma,\cE,\mu)
\end{equation}
with the following constituents:
\begin{enumerate}[label=(\roman*)]
\item \(\cF\) is a field theory input.  Depending on context, this may be an extended topological quantum field theory, a rational conformal field theory, a modular functor, a unitary modular tensor category, a unitary fusion category together with its string-net/Turaev--Viro theory, or a defect field theory.
\item \(\Sigma\) is the spatial datum.  It may be a compact oriented surface, punctured surface, manifold, stratified space, or boundary component to which \(\cF\) assigns a state space.
\item \(\cD\) is the decoration datum.  It records labels of punctures, boundary conditions, condensable algebras, domain walls, defect lines, marked sectors, or prescribed total charge constraints.
\item \(\Gamma\) is a physical realization datum.  It may be a lattice, cellulation, spine, tensor network, hardware graph, or other discretization through which the continuum or categorical state space is represented in a physical Hilbert space.
\item \(\cE\) is an error model.  It is a specified class of local field insertions, local operators, defect segments, Pauli errors, anyon pair-creation events, or spacetime defect networks.
\item \(\mu\) is a weighting or inference structure on errors.  In a stochastic Pauli model, \(\mu\) is a probability distribution.  In a topological model, \(\mu\) may be a topological path-integral weight.  In a conformal model, \(\mu\) may be built from conformal-block norms, correlation functions, scaling dimensions, or modular data.
\end{enumerate}
The associated logical Hilbert space is the field-theoretic state space
\begin{equation}\label{eq:logical}
\cH_L(\mathfrak Q)=Z_\cF(\Sigma,\cD),
\end{equation}
while the associated physical Hilbert space is a realization
\begin{equation}\label{eq:physical}
\cH_P(\mathfrak Q)=\cH_P(\Gamma).
\end{equation}
A choice of discretized realization supplies an encoding map
\begin{equation}\label{eq:encoding}
V_\Gamma: Z_\cF(\Sigma,\cD)\longrightarrow \cH_P(\Gamma),
\end{equation}
usually required to be an isometry or an approximate isometry onto a protected subspace.
\end{definition}

Definition \ref{def:FTcode} separates the abstract code space from its physical implementation.  A topological field theory may assign a finite-dimensional Hilbert space to a decorated surface without specifying a microscopic qubit Hamiltonian; a lattice realization, tensor network, or hardware graph then supplies the physical degrees of freedom.  Conversely, different physical realizations may represent the same field-theoretic code space, just as different surface-code lattices, boundaries, and surgery protocols can realize the same homological logic.

There are two conventions implicit in Definition \ref{def:FTcode}.  First, the object \(Z_\cF(\Sigma,\cD)\) is meant to be the \emph{ideal} logical space.  It is the state space seen by the continuum or categorical theory after imposing the topological, conformal, or defect constraints.  The physical Hilbert space \(\cH_P(\Gamma)\) is usually much larger.  It contains microscopic degrees of freedom, ultraviolet choices, leakage sectors, and excited sectors that are not part of the protected logical space.  Thus the encoding map \(V_\Gamma\) should be included explicitly among the data.  In an exactly solvable lattice model it may be the inclusion of the ground-state sector into the full lattice Hilbert space.  In a tensor-network realization it may be the isometry defined by contracting the network with open physical legs.  In a conformal-block realization it may be an experimentally engineered embedding of a fusion or block space into a device-level Hilbert space.

Second, the field-theoretic datum is allowed to be exact or approximate.  A topological fixed-point model may supply an exact commuting-projector code, whereas a physical platform may only approximate the ideal state space up to finite correlation length, finite temperature, device noise, or truncation of a continuum theory.  In the exact case, one can ask for strict Knill--Laflamme conditions on a prescribed family of errors.  In the approximate case, the same formalism should be read as a blueprint for approximate quantum error correction: local neutral insertions act nearly as scalars on the low-energy or protected sector, and nonlocal defect histories account for logical failure modes.

The decoration datum \(\cD\) is broad because it changes the local operator theory.  In a modular functor it records labelled punctures and total-charge constraints; in a boundary TQFT, boundary conditions such as condensable-algebra data; and in a defect field theory, domain walls and junctions.  These labels determine which charges may end on a boundary, which measurements are available, and which defect moves are treated as gauge equivalences.  The logical Hilbert space is therefore attached to the decorated pair \((\Sigma,\cD)\), not to \(\Sigma\) alone.

The inference structure \(\mu\) is part of the code specification.  A code and a syndrome map do not determine how one should choose among distinct errors with the same syndrome; a noise or likelihood model is still needed.  Here a topological theory may identify representatives of a topological class, whereas a conformal or geometric refinement may distinguish them by distances, cross-ratios, scaling behaviour, or moduli.  The tuple \(\mathfrak Q\) records both the protected state space and the additional structure used for decoding.

Definition \ref{def:FTcode} is broader than the hypotheses of the main theorems.  Three specializations will be used repeatedly.  In the \emph{semisimple categorical regime}, \(\cF\) is represented by a semisimple unitary fusion category, often a unitary modular tensor category or its modular functor; the logical spaces are fusion spaces; and footprints are local total-charge or fusion-channel sectors.  In the \emph{conformal-block regime}, these spaces vary in a bundle over a configuration or moduli space of marked curves.  In the \emph{state-sum regime}, the physical realization is a string-net or Turaev--Viro lattice model.  The categorical and conformal-block results below are always subject to the corresponding semisimplicity and unitarity assumptions.

\subsection{Categorical footprint sectors and projector algebras}

The theorem-level statements in this part of the paper are made in a finite semisimple unitary setting.  Thus, throughout this subsection, \(\cC\) is a unitary fusion category: a finite semisimple rigid \(C^*\)-tensor category with simple tensor unit, finitely many isomorphism classes of simple objects, and unitary associativity constraints.  Braiding or modularity is imposed only when explicitly needed later.  This convention makes the footprint construction a finite-dimensional Hilbert-space construction rather than a purely formal graphical idea.

\begin{definition}[Admissible cluster]\label{def:admissible-cluster}
Let \(X_1,\ldots,X_n\in\cC\) and fix a parenthesized tensor product, equivalently a planar fusion tree.  An \emph{admissible cluster} is a subexpression represented by an internal edge of this fusion tree.  Equivalently, it is a collection of tensor factors that has been grouped by specified associators.  If \(\cC\) is braided, more general ordered subsets may be treated as admissible after specifying the braiding and recoupling convention used to bring the chosen factors into a single tensor factor.
\end{definition}

The phrase ``cluster'' will always mean admissible cluster in the categorical theorem-level statements below.  This convention is important in a non-braided fusion category, where arbitrary non-contiguous subsets cannot be regrouped without extra planar or braiding data.

\begin{definition}[Categorical footprint sectors]\label{def:categorical-footprints}
Let \(X_1,\ldots,X_n\) be objects of a unitary fusion category \(\cC\).  Set
\[
  \cH(\vec X)=\Hom_{\cC}(\one,X_1\otimes\cdots\otimes X_n),
\]
with the inner product determined by the unitary categorical structure.  For an admissible cluster \(I\), write \(X_I\) for the corresponding subexpression and \(X_{I^c}\) for the complementary tensor factor determined by the chosen cut or recoupling convention.  The coarse categorical footprint of the cluster \(I\) is its total simple charge \(a\in\Irr(\cC)\).  A refined footprint may additionally include a measured multiplicity label, or more generally a projection in the multiplicity space \(\Hom_{\cC}(a,X_I)\), whenever that space has dimension greater than one.
\end{definition}

\begin{theorem}[Categorical footprint decomposition]\label{thm:categorical-footprint-decomp}
Let \(\cC\) be a unitary fusion category.  For every admissible bipartition \(I|I^c\) of the labelled tensor product \(X_1\otimes\cdots\otimes X_n\), there is a natural total-charge decomposition
\begin{equation}\label{eq:categorical-footprint-decomp}
  \Hom_{\cC}(\one,X_1\otimes\cdots\otimes X_n)
  \cong
  \bigoplus_{a\in\Irr(\cC)}
  \Hom_{\cC}(a,X_I)\otimes
  \Hom_{\cC}(a^*,X_{I^c}).
\end{equation}
The decomposition is canonical at the level of total-charge isotypic summands.  The displayed tensor-factor identification is unique only up to the unitary recouplings determined by the associator, together with basis choices inside multiplicity spaces.  The orthogonal projector onto the summand indexed by \(a\) is the coarse footprint projector \(P_{I,a}\).  It measures the total charge \(a\) of the admissible cluster \(I\), while forgetting multiplicity data inside \(\Hom_{\cC}(a,X_I)\) unless the measurement is refined.
\end{theorem}

\begin{proof}
The decomposition follows from semisimple resolution of the intermediate tensor factor.  Since \(\cC\) is finite semisimple, the object \(X_I\) admits a decomposition
\[
  X_I \cong \bigoplus_{a\in\Irr(\cC)} \Hom_{\cC}(a,X_I)\otimes a,
\]
where the vector space \(\Hom_{\cC}(a,X_I)\) records multiplicity.  Substituting this into
\(\Hom_{\cC}(\one,X_I\otimes X_{I^c})\) gives
\[
  \Hom_{\cC}(\one,X_I\otimes X_{I^c})
  \cong
  \bigoplus_a
  \Hom_{\cC}(a,X_I)\otimes
  \Hom_{\cC}(\one,a\otimes X_{I^c}).
\]
Rigid duality identifies the second factor with \(\Hom_{\cC}(a^*,X_{I^c})\).  Concretely, a vector in the summand labelled by \(a\) is obtained by choosing morphisms \(f:a\to X_I\) and \(g:a^*\to X_{I^c}\), and composing
\[
  \one \xrightarrow{\operatorname{coev}_a} a\otimes a^*
  \xrightarrow{f\otimes g} X_I\otimes X_{I^c}.
\]
Every vacuum vector is a finite sum of such composites.

The unitary structure supplies the graphical inner product.  Distinct simple labels give orthogonal isotypic summands because morphisms between nonisomorphic simple objects vanish and the category is \(C^*\)-semisimple.  The orthogonal projection onto the \(a\)-summand is therefore well defined.  Changing the fusion tree only changes the displayed tensor-factor coordinates by the unitary associator, i.e. by \(F\)-moves, and changing bases in the multiplicity spaces conjugates within the same summand.  Thus the coarse total-charge sector, and hence the projector \(P_{I,a}\), is intrinsic after the admissible cluster and its planar recoupling convention have been fixed.
\end{proof}

This theorem is the categorical prototype for the paper's use of the word footprint.  A local cluster measurement selects a total-charge summand already present in the fusion space.  The terminology emphasizes that this measured charge is local boundary data associated with a latent defect history.

\begin{definition}[Diagnostic and syndrome footprint algebras]\label{def:footprint-algebra}
For a fixed admissible cluster \(I\), define
\[
  \cA_I=\operatorname{span}_{\bC}\{P_{I,a}:a\in\Irr(\cC)\}\subseteq \End(\cH(\vec X)).
\]
This is a finite-dimensional commutative \(*\)-algebra with
\[
  P_{I,a}P_{I,b}=\delta_{ab}P_{I,a},
  \qquad
  \sum_{a\in\Irr(\cC)}P_{I,a}=\id_{\cH(\vec X)}.
\]
The generally noncommutative \(*\)-algebra generated by footprint projectors for several, possibly overlapping, clusters will be called the \emph{diagnostic footprint algebra} and denoted \(\cA_{\mathrm{diag}}\).  It records all selected footprint measurements as quantum observables, whether or not they can be measured simultaneously.

A \emph{footprint measurement algebra}, in the sense used for syndrome extraction, is a finite-dimensional commutative \(*\)-subalgebra
\[
  \cA_{\mathrm{fp}}=\operatorname{span}_{\bC}\{\Pi_s:s\in S\}
  \subseteq \End(\cH(\vec X))
\]
generated by a compatible family of footprint projectors, or by chosen coarse-grainings of them.  The projectors \(\Pi_s\) are orthogonal and sum to the identity on the measured Hilbert space.  Noncommuting footprint projectors may be diagnostically meaningful, but they do not jointly define a classical syndrome algebra without specifying an ordered or adaptive measurement instrument.
\end{definition}

\begin{definition}[Syndrome-admissible footprint algebra]\label{def:syndrome-admissible}
Let \(P\) be the projection onto a code subspace \(\cH_L\subseteq \cH(\vec X)\), let \(\cE=\{E_\alpha\}\) be an error family, and let
\[
  \cA_{\mathrm{fp}}=\operatorname{span}\{\Pi_s:s\in S\}
\]
be a commuting footprint measurement algebra.  We say that \(\cA_{\mathrm{fp}}\) is \emph{syndrome-admissible} for \((P,\cE)\) if the following two conditions hold.
\begin{enumerate}[label=(\roman*)]
\item The no-error measurement does not reveal logical information.  In the sharpest case, there is a distinguished sector \(s_0\in S\) such that
\[
  \Pi_{s_0}P=P,
  \qquad
  \Pi_sP=0\quad(s\neq s_0).
\]
More generally, the no-error outcome distribution of the measurement instrument should be independent of the encoded state.
\item The chosen error representatives are footprint-resolved: for each \(E_\alpha\in\cE\), there is a sector \(s(\alpha)\in S\) such that
\[
  \Pi_t E_\alpha P=\delta_{t,s(\alpha)}E_\alpha P.
\]
\end{enumerate}
\end{definition}

\begin{remark}[Resolving general errors into measured sectors]\label{rem:error-sector-resolution}
Condition~(ii) is a condition on the chosen representatives, not a claim that every physical error is sector-homogeneous.  For an arbitrary operator \(E\),
\[
  EP=\sum_{s\in S}\Pi_sEP.
\]
Thus one may replace a general error family by its resolved family \(\{\Pi_sE_\alpha:s\in S,\alpha\}\).  Exact correctability of the resolved family implies exact correctability of its linear span by the linearity of the Knill--Laflamme conditions.  This is the operator-level reason that sector-homogeneous representatives can be used without excluding coherent superpositions of measured footprints.
\end{remark}

Three pieces of data are sometimes conflated in stabilizer notation.  First, there is a 
\emph{sector decomposition}, encoded here by the projectors \(\Pi_s\).  Second, there is an 
\emph{instrument}, meaning a physical procedure whose Kraus operators realize or approximate this measurement.  Third, there is an 
\emph{interpretation} of the outcome as a footprint of some error history.  Definition~\ref{def:syndrome-admissible} is stated at the level of the first and third pieces: it specifies which sector information is allowed to become classical syndrome data.  A device-level construction must still implement the corresponding instrument without coupling to the encoded state in an uncontrolled way.

The sharp projective condition in (i) is stronger than mere independence of the outcome probabilities from the logical state.  It says that the code subspace is contained in one sector before any error occurs.  This is the direct analogue of the stabilizer convention that the codespace is the simultaneous \(+1\)-eigenspace of the check operators.  The more general instrument language is included because topological and conformal realizations may produce noisy or coarse measurements whose outcomes have state-independent probabilities but whose post-measurement states depend on the details of the detector.  Such measurements can still be useful, but then the detector itself becomes part of the error model.

A commuting footprint algebra can fail to be syndrome-admissible.  It may instead serve as a diagnostic, tomography, or gate-measurement algebra.  This distinction is essential in nonabelian fusion settings, where even natural local charge measurements may reveal logical information.

For a projective measurement \({\Pi_s}\), we say that it is \emph{selectively nondisturbing on the code} if, for every code state \(\psi\) and every outcome \(s\) of nonzero probability, the normalized post-measurement state represents the same ray as \(\psi\).

\begin{lemma}[No-error safety for projective footprint measurements]\label{lem:no-error-safe}
Let \(\cA_{\mathrm{fp}}=\operatorname{span}\{\Pi_s:s\in S\}\) be a projective footprint measurement algebra and let \(P\) be the code projection.  In the sharp projective case, the no-error measurement is selectively nondisturbing and reveals no logical information precisely when the code lies in a single measured sector, i.e. there is an \(s_0\) such that
\[
  \Pi_{s_0}P=P,\qquad \Pi_sP=0\quad(s\neq s_0).
\]
More generally, if a measurement instrument has several possible no-error outcomes with probabilities independent of the encoded state, the instrument is additional operational data; the bare projectors do not specify it.  Anyonic interferometry provides a canonical physical setting in which topological-charge measurements are implemented by a specified instrument rather than by an abstract projector alone \cite{BondersonShtengelSlingerland2008}.
\end{lemma}

\begin{proof}
If the displayed condition holds, then for any code vector \(\psi\in P\cH\) the Born probabilities of the sharp measurement are
\[
  \|\Pi_s\psi\|^2=
  \begin{cases}
  \|\psi\|^2,&s=s_0,\\
  0,&s\neq s_0.
  \end{cases}
\]
The outcome is deterministic and the selective post-measurement state is again \(\psi\).  Hence the measurement is nondisturbing on the code and reveals no logical information.

Conversely, assume that the projective measurement is selectively nondisturbing on the code and that its outcome distribution is independent of the encoded state.  The restriction \(P\Pi_sP\) is a positive operator on the code space.  State-independence says that the quadratic form \(\langle\psi,\Pi_s\psi\rangle\) has the same value \(p_s\) on every unit vector in \(P\cH\).  By the polarization identity this forces
\[
  P\Pi_sP=p_sP.
\]
Selective nondisturbance now adds more than state-independent probabilities: whenever outcome \(s\) occurs, every code vector must be preserved up to normalization.  Hence every code vector lies in an eigenspace of the projector \(\Pi_s\).  A projector has only the eigenvalues \(0\) and \(1\), so \(p_s\in\{0,1\}\).  Orthogonality and \(\sum_s\Pi_s=1\) imply that exactly one scalar is equal to \(1\).  Thus the code lies in a single measured sector.  The last assertion is a warning about general instruments: a non-projective or noisy detector may have state-independent outcome probabilities without being encoded by one sharp sector projector, and its Kraus operators must then be included as part of the operational model.
\end{proof}

\begin{proposition}[Relations among footprint projectors]\label{prop:footprint-projector-relations}
Let \(\cC\) be a unitary fusion category and let \(\cH(\vec X)\) be a fusion space.
\begin{enumerate}[label=(\roman*)]
\item For a fixed admissible cluster \(I\), the projectors \(P_{I,a}\) are pairwise orthogonal and sum to the identity.
\item If clusters \(I\) and \(J\) are represented by disjoint vertices of a common fusion tree, the corresponding footprint algebras commute in that tree basis.  Their joint eigenspaces are the fusion paths with the specified intermediate charges.
\item If \(I\subset J\) and both clusters occur as vertices of a common fusion tree, then the projectors for \(I\) and \(J\) commute, and their products refine the total-charge decomposition according to the relevant fusion multiplicities.
\item For overlapping clusters which are not simultaneously represented in a common fusion tree, the corresponding footprint projectors need not commute.  They generate a noncommutative diagnostic footprint algebra; a classical syndrome algebra requires a compatible commuting family.
\end{enumerate}
\end{proposition}

\begin{proof}
For (i), Theorem \ref{thm:categorical-footprint-decomp} writes the fusion space as an orthogonal direct sum over the possible total charges of the fixed admissible cluster.  The projectors onto distinct summands are therefore orthogonal idempotents, and their sum is the identity on the whole fusion space.

For (ii) and (iii), choose the common planar fusion tree.  A basis vector, or more invariantly a joint summand, is specified by simple labels on compatible internal edges together with possible multiplicity data at vertices.  Projectors associated to those internal edges act by asking whether the corresponding edge label has a prescribed value.  Such diagonal operators commute.  If \(I\subset J\), then the label of \(I\) and the label of \(J\) are two labels in the same nested fusion path; imposing both conditions simply refines the summand according to the allowed fusion multiplicities between them.

For (iv), no common diagonalization is available in general.  Passing from one overlapping cluster to another requires an associator move, or a product of such moves, and the relevant projectors are related by conjugation by the corresponding \(F\)-matrix.  A diagonal projector and its conjugate by a nontrivial \(F\)-matrix need not commute.  The four-puncture Ising calculation below is the smallest visible instance: in the \((12)3\) channel the \((12)\)-footprint is represented by \(Z\), while the \((23)\)-footprint is represented after recoupling by \(X\).
\end{proof}

\begin{proposition}[Fusion-tree footprint algebra]\label{prop:fusion-tree-footprint-algebra}
Fix a planar fusion tree \(T\) for \(X_1\otimes\cdots\otimes X_n\).  Let \(E(T)\) denote its internal edges, and let \(P_{e,a}\) be the total-charge footprint projector associated to the admissible cluster cut by the edge \(e\).  The algebra
\[
  \cA_T:=\operatorname{alg}\{P_{e,a}:e\in E(T),\ a\in\Irr(\cC)\}
\]
generated by these projectors is a finite-dimensional commutative \(*\)-algebra.  Its joint spectral projectors are indexed by admissible labellings \(\gamma:E(T)\to\Irr(\cC)\) satisfying the fusion rules at every vertex.  Equivalently,
\[
  \cH(\vec X)=\bigoplus_{\gamma\in\operatorname{Path}(T;\vec X)} \cH_\gamma,
\]
where \(\operatorname{Path}(T;\vec X)\) is the Bratteli set of fusion paths compatible with the external labels.  The minimal central projectors of \(\cA_T\) are the projections onto these joint path sectors.  If the relevant fusion multiplicities are all zero or one, then the summands \(\cH_\gamma\) are one-dimensional and \(\cA_T\) is a maximal diagonal algebra in this fusion-tree basis.
\end{proposition}

\begin{proof}
A fixed planar fusion tree is precisely a compatible collection of cuts.  Label its internal edges by simple objects.  For each edge \(e\), the projectors \(P_{e,a}\) ask for the value of one of these labels, so they are all diagonal in the same fusion-tree decomposition and commute.

Given a function \(\gamma:E(T)\to\Irr(\cC)\), form the product
\[
  Q_\gamma=\prod_{e\in E(T)}P_{e,\gamma(e)}.
\]
Because the factors commute, \(Q_\gamma\) is again a projection.  It is zero exactly when the labels assigned by \(\gamma\) violate a fusion rule at some vertex of the tree.  If all vertex constraints are satisfied, \(Q_\gamma\) projects onto the tensor product of the corresponding vertex multiplicity spaces.  The nonzero \(Q_\gamma\)'s are pairwise orthogonal and sum to the identity, giving the Bratteli, or fusion-path, decomposition
\[
  \cH(\vec X)=\bigoplus_{\gamma\in\operatorname{Path}(T;\vec X)}\cH_\gamma.
\]
The algebra generated by the edge projectors is therefore the algebra of functions on this finite set of joint sectors, with values constant on unresolved multiplicity spaces.  When all fusion multiplicities are zero or one, those unresolved factors are one-dimensional, so the joint projectors are rank one in the chosen fusion-tree basis and the algebra is maximal diagonal there.
\end{proof}

This proposition clarifies the relation with stabilizer syndrome algebras.  In an abelian stabilizer-type situation, the measured check algebra is commuting from the outset.  In a nonabelian fusion setting, a chosen fusion tree supplies a commuting classical footprint algebra, while changing to an overlapping cluster generally conjugates this algebra by nontrivial recoupling matrices and produces noncommuting diagnostics.

\begin{remark}[Syndromes versus diagnostics]\label{rem:commuting-noncommuting-footprints}
Unlike stabilizer checks in an ordinary stabilizer code, arbitrary categorical footprint projectors do not form a commuting syndrome algebra.  A QEC syndrome-extraction protocol must specify a commuting, syndrome-admissible footprint measurement algebra \(\cA_{\mathrm{fp}}\), or else an ordered/adaptive measurement instrument whose disturbance is part of the protocol.  Noncommuting footprint measurements are still meaningful: they may be alternative diagnostics, tomography measurements, logical operations, or gate primitives.  A measured syndrome therefore records a chosen compatible measurement protocol; incompatible local footprint projectors do not possess simultaneous measured values.
\end{remark}

\begin{theorem}[Footprint-algebra Knill--Laflamme theorem]\label{thm:footprint-algebra-KL}
Let \(P\) be the projection onto a finite-dimensional code subspace \(\cH_L\subseteq\cH(\vec X)\), and let \(\cA_{\mathrm{fp}}\subseteq\End(\cH(\vec X))\) be a syndrome-admissible footprint algebra for a finite error family \(\cE=\{E_\alpha\}\), with orthogonal sector projectors \(\{\Pi_s\}_{s\in S}\).  Then exact recovery conditioned on the measurement of \(\cA_{\mathrm{fp}}\) exists if and only if, for every measured sector \(s\) and all \(E_\alpha,E_\beta\) with \(s(\alpha)=s(\beta)=s\), there are scalars \(\lambda^{(s)}_{\alpha\beta}\) such that
\begin{equation}\label{eq:footprint-algebra-KL}
  P E_\alpha^\dagger E_\beta P
  =\lambda^{(s)}_{\alpha\beta}P.
\end{equation}
Equivalently, after conditioning on the measured footprint sector, every indistinguishable neutral composite ambiguity must act as a scalar on the code.
\end{theorem}

\begin{proof}
Let \(\Pi_s\) be the orthogonal sector projectors of the syndrome-admissible footprint algebra.  By definition, the no-error sector carries no logical information, and each chosen representative has a definite measured footprint:
\[
  \Pi_t E_\alpha P=\delta_{t,s(\alpha)}E_\alpha P.
\]
If \(s(\alpha)\neq s(\beta)\), then orthogonality of the measured sectors gives
\[
  P E_\alpha^\dagger E_\beta P
  =P E_\alpha^\dagger\Pi_{s(\alpha)}\Pi_{s(\beta)}E_\beta P=0.
\]
Thus cross-sector ambiguities have already been converted into classical syndrome information; a recovery need not preserve coherence between distinct measured outcomes.

Fix a sector \(s\).  On the subfamily \(\{E_\alpha:s(\alpha)=s\}\), the problem is exactly the usual finite-dimensional error-correction problem, but conditioned on the observed sector.  The ordinary Knill--Laflamme theorem says that a recovery acting after this conditioning exists precisely when
\[
  P E_\alpha^\dagger E_\beta P=\lambda^{(s)}_{\alpha\beta}P
\]
for all representatives in that sector.  If the equations hold in every sector, the corresponding conditional recovery maps assemble over the mutually orthogonal syndrome sectors to give a recovery for the whole syndrome-resolved family.  Conversely, any recovery conditioned on the footprint measurement restricts to an exact recovery for each fixed sector, so the same scalar equations are necessary in every fibre.
\end{proof}

\begin{remark}[Equivalence with ordinary Knill--Laflamme after sector resolution]\label{rem:fibrewise-KL-equivalence}
Under syndrome-admissibility, Theorem~\ref{thm:footprint-algebra-KL} is equivalent to the ordinary Knill--Laflamme criterion applied to the sector-resolved error family.  Indeed, if \(s(\alpha)\neq s(\beta)\), then the cross-sector compression vanishes automatically, so the corresponding Knill--Laflamme scalar is zero; within a fixed sector, the equations are exactly the usual scalar equations.  No stronger exact-correction criterion is claimed.  The additional structure is the measure-then-recover factorization: cross-sector coherence has become classical outcome data, while the remaining quantum ambiguity is isolated inside each measured fibre.
\end{remark}

The preceding proof is compressed to emphasize its relation with the ordinary Knill--Laflamme theorem.  We record the corresponding recovery picture explicitly, since it is the bridge between the algebraic theorem and the field-theoretic language of neutral composites.

\begin{proposition}[Conditional recovery inside a footprint fibre]\label{prop:conditional-recovery-fibre}
Fix a measured sector \(s\) and write \(\cE_s=\{E_\alpha:s(\alpha)=s\}\).  Suppose the fibrewise scalar equations
\[
  P E_\alpha^\dagger E_\beta P=\lambda^{(s)}_{\alpha\beta}P
  \qquad(E_\alpha,E_\beta\in\cE_s)
\]
hold.  Then, after the outcome \(s\) has been observed, the restricted noise channel on the code is correctable by a recovery depending only on \(s\).  More precisely, after diagonalizing the positive matrix \(\Lambda^{(s)}=(\lambda^{(s)}_{\alpha\beta})\), the corresponding linear combinations of error operators have mutually orthogonal images of the code and are corrected by partial isometries back to \(\cH_L\).
\end{proposition}

\begin{proof}
The matrix \(\Lambda^{(s)}\) is positive semidefinite because for any coefficients \(c_\alpha\) one has
\[
  P\Big(\sum_\alpha c_\alpha E_\alpha\Big)^\dagger
      \Big(\sum_\beta c_\beta E_\beta\Big)P
  =\Big(\sum_{\alpha,\beta}\overline{c_\alpha}c_\beta\lambda^{(s)}_{\alpha\beta}\Big)P,
\]
and the scalar multiplying \(P\) is nonnegative by positivity of the Hilbert-space inner product.  Choose a unitary matrix \(U\) diagonalizing \(\Lambda^{(s)}\), with eigenvalues \(d_j\ge 0\), and set
\[
  F_j=\sum_\alpha U_{j\alpha}E_\alpha .
\]
Then
\[
  P F_i^\dagger F_j P=d_j\delta_{ij}P.
\]
For \(d_j>0\), the operator \(d_j^{-1/2}F_jP\) is an isometry from \(\cH_L\) onto its image.  These images are mutually orthogonal.  A recovery for the sector \(s\) first projects onto the orthogonal sum of these images and then applies the adjoint partial isometry on each summand.  Operators with \(d_j=0\) vanish on the code and do not require correction.  This is the standard constructive proof of Knill--Laflamme, applied after the classical footprint sector has already been measured.
\end{proof}

The phrase ``inside a footprint fibre'' has a literal operator meaning.  The footprint measurement converts the cross-sector part of the error-correction problem into classical information, leaving a family of indistinguishable representatives within one sector.  Categorically, the matrix \(\Lambda^{(s)}\) is the Gram matrix obtained by closing pairs of compatible histories against the code.  In the contractible vacuum case, these closures are neutral diagrams and the matrix entries are their scalar evaluations.

The next proposition uses Proposition~\ref{prop:local-scalar}; its proof is independent of the results of the present section.

\begin{proposition}[Contractible-vacuum sufficient criterion]\label{prop:categorical-scalarity-correctability}
In the setting of Theorem \ref{thm:footprint-algebra-KL}, suppose that for every measured footprint sector \(s\) and every pair \(E_\alpha,E_\beta\) in that sector, the compressed composite \(P E_\alpha^\dagger E_\beta P\) is represented on the code by insertion of a defect-network composite contained in a contractible disk \(D\), disjoint from labelled punctures and unresolved multiplicity spaces, whose boundary carries total charge \(\one\).  Then
\[
  P E_\alpha^\dagger E_\beta P
  =\lambda^{(s)}_{\alpha\beta}P
\]
for some scalar \(\lambda^{(s)}_{\alpha\beta}\in\bC\).  Hence the error family satisfies the fibrewise Knill--Laflamme condition and is exactly correctable by a recovery conditioned on the measured footprint.
\end{proposition}

\begin{proof}
Fix \(s\) and a pair \(E_\alpha,E_\beta\) in the corresponding footprint fibre.  By hypothesis, their residual composite is represented on the code by a network contained in a contractible disk \(D\) with vacuum total charge on \(\partial D\), and \(D\) is disjoint from punctures and unresolved multiplicity data.  Proposition \ref{prop:local-scalar} applies to this local insertion.  Its evaluation lies in
\[
  \End_{\cC}(\one)\cong\bC,
\]
so gluing the disk back into the ambient surface multiplies every code vector by a scalar \(\lambda^{(s)}_{\alpha\beta}\).  Therefore
\[
  P E_\alpha^\dagger E_\beta P
  =\lambda^{(s)}_{\alpha\beta}P.
\]
Theorem \ref{thm:footprint-algebra-KL} then gives exact recovery conditioned on the measured sector.
\end{proof}

\begin{remark}[Why no unconditional converse is claimed]\label{rem:no-scalarity-converse}
The converse to Proposition \ref{prop:categorical-scalarity-correctability} is false without an additional faithfulness hypothesis.  A residual operator may be non-scalar in the ambient field-theoretic or physical operator algebra and nevertheless compress to a scalar on the chosen code:
\[
  X\notin\bC I
  \qquad\text{but}\qquad
  PXP=\lambda P.
\]
Thus failure to reduce categorically to an element of \(\End_{\cC}(\one)\) is a warning that extra boundary, multiplicity, or topological data remain; it is not by itself a proof of QEC failure.  A converse becomes available only after one assumes, for example, that the relevant residual operator algebra preserves the code and acts faithfully modulo scalars on the protected subspace.  Under such an assumption a non-scalar residual class cannot disappear under compression.
\end{remark}

The contractible-vacuum hypothesis in Proposition \ref{prop:categorical-scalarity-correctability} is essential for that sufficient argument.  If the composite acts near labelled punctures, nontrivial boundary sectors, or unresolved multiplicity spaces, the relevant endomorphism algebra may be larger than \(\End_{\cC}(\one)\).  One may then refine the footprint algebra to measure the additional sector, treat it as gauge or leakage data, include it deliberately in the protected logical subsystem, or verify the compressed Knill--Laflamme equations directly.  Remark \ref{rem:no-scalarity-converse} explains why a non-scalar ambient representative is not automatically an obstruction.

\begin{corollary}[Contractible-vacuum categorical footprint correctability]\label{cor:clean-categorical-footprint-correctability}
Let \(\cA_{\mathrm{fp}}\) be a syndrome-admissible footprint algebra for a fusion-space code and an error family represented by contractible local defect networks.  Suppose that, inside each measured footprint fibre, every indistinguishable composite closes to a contractible neutral vacuum diagram.  Then the error family is exactly correctable by a recovery conditioned on \(\cA_{\mathrm{fp}}\).  Within the stated hypotheses, the contractible-vacuum evaluation supplies the required scalar equations.  Outside those hypotheses, noncontractible Wilson lines and unresolved boundary or multiplicity operators are possible sources of non-scalar compressed action and must be checked separately.
\end{corollary}

\begin{proof}
By hypothesis, every indistinguishable composite inside a measured footprint fibre is represented by a contractible neutral vacuum diagram.  Proposition \ref{prop:categorical-scalarity-correctability} evaluates each such diagram as a scalar element of \(\End_{\cC}(\one)\).  Hence, for every pair of error representatives in the same measured fibre, the composite satisfies
\[
  P E_\alpha^\dagger E_\beta P=\lambda_{\alpha\beta}^{(s)}P.
\]
The fibrewise Knill--Laflamme equations of Theorem \ref{thm:footprint-algebra-KL} are therefore satisfied, and an exact recovery conditioned on the footprint measurement exists.

If a composite in a measured fibre is instead a noncontractible Wilson line, an operator on an unresolved multiplicity space, or a boundary-sector operator, then it is not forced by the vacuum evaluation to be scalar.  In that case the contractible-vacuum proof no longer applies.  Exact correction may still be possible after refining the syndrome algebra, changing the code subsystem, or treating the extra sector as gauge or leakage data, but it is not a consequence of contractible-vacuum scalarity alone.
\end{proof}

\begin{definition}[Footprint, field-theoretic syndrome, and decoder]\label{def:decoder}
Let \(\cE\) be the allowed class of error insertions or defect histories.  A \emph{footprint map} is a map
\begin{equation}\label{eq:footprint-map}
\Foot:\cE\longrightarrow \cS
\end{equation}
from errors to a space \(\cS\) of locally visible field-theoretic boundary data.  The footprint space may consist of stabilizer eigenvalue data, boundary data of chains, topological charge assignments, fusion-channel labels, defect endpoints, boundary-sector changes, or conformal-sector labels.  In an idealized measurement model, the syndrome map is the measurement of the footprint.  At that level one may write \(\Synd=\Foot\) as maps, while keeping their roles distinct: the footprint is the local field-theoretic datum induced by the error history, whereas the syndrome is the classical value exposed by the chosen measurement model.  Outside the idealized model the detector and its coarse-graining are additional data, so the two maps need not literally coincide.

Given an observed footprint, or syndrome, \(s\in\cS\), a \emph{field-theoretic decoder} is a rule assigning to \(s\) a recovery class \(\widehat E(s)\), usually modulo locally neutral networks and stabilizer-like equivalences, obtained from an optimization or inference problem of the form
\begin{equation}\label{eq:decoder}
\widehat E(s)\in
\operatorname*{arg\;max}_{E\in\cE:\,\Foot(E)=s}
\mu(E\mid s),
\end{equation}
where \(\mu(E\mid s)\) is the field-theoretic weight, probability, amplitude norm, or posterior assigned to the compatible error history \(E\).  In a microscopic implementation this recovery class must be represented by an actual physical recovery operation, for instance a Pauli-frame update, a code-deformation step, or a completely positive trace-preserving map.
\end{definition}

\begin{remark}[Representative MAP versus class maximum likelihood]\label{rem:representative-vs-class-ML}
Equation~\eqref{eq:decoder} is a most-likely-representative rule.  If the physical recovery depends only on an equivalence class \([E]\) of compatible histories and \(\mu\) is a posterior probability, class-level maximum-likelihood decoding instead uses
\[
  \widehat{[E]}(s)\in
  \operatorname*{arg\;max}_{[E]\subseteq\Foot^{-1}(s)}
  \sum_{E'\in[E]}\mu(E'\mid s).
\]
The distinction is consequential in degenerate codes, where a class of individually less likely representatives may carry more total posterior mass.  For amplitude-valued weights, aggregation over a class requires a specified physical inner product and interference or measurement rule; a probability sum should not be inserted by convention.
\end{remark}

In ordinary surface-code decoding, \(E\) is an error chain, its footprint is \(\Foot(E)=\partial E\), and \(\mu(E\mid s)\) is determined by the stochastic error model once the boundary syndrome \(s\) is observed.  In a TQFT decoder, \(E\) is a defect network in spacetime, and \(\mu\) depends on the topological amplitude of the corresponding decorated bordism.  In a conformal-block likelihood model, \(E\) is a family of insertions, the footprint records the compatible local channels or sectors, and \(\mu\) is built from the squared norm of conformal blocks or from a full correlation function.

\begin{principle}[Field-theoretic organization of operations]\label{principle:actor}
In a field-theoretic code, the field theory organizes the ideal structures underlying four operations:
\[
\begin{array}{rcl}
\text{encoding} &:& \parbox{0.63\linewidth}{state spaces assigned to decorated spatial data}\\[0.35em]
\text{extraction} &:& \parbox{0.63\linewidth}{measurements of error footprints: charges, channels, or boundary sectors}\\[0.35em]
\text{detection} &:& \parbox{0.63\linewidth}{local identification of nontrivial defect endpoints or violations}\\[0.35em]
\text{correction} &:& \parbox{0.63\linewidth}{inference over compatible defect histories weighted by field-theoretic amplitudes}
\end{array}
\]
\end{principle}

Principle~\ref{principle:actor} treats the field theory as organizing ideal state spaces and sector data rather than as implementing a device by itself.  The physical realization, measurement model, and recovery map remain part of the code datum.  In this formulation the protected space is attached to decorated spatial data, syndrome extraction exposes selected local sector information, and decoding compares compatible field histories.

\medskip
\noindent\textbf{Footprint fibres and exact correctability.}
It is worth making explicit how Definition \ref{def:decoder} reduces to the usual exact algebraic conditions when the field-theoretic model is realized by honest operators on a physical Hilbert space.  For each footprint value \(s\in\cS\), write
\[
  \cE_s=\Foot^{-1}(s)
\]
for the fibre of allowed errors with that footprint.  We impose the following explicit ideal footprint measurement hypothesis.  Assume that the measurement protocol determines orthogonal syndrome subspaces with projectors \(\Pi_s\), and that the chosen, syndrome-resolved error representatives satisfy
\[
  \Pi_t E_a P = \delta_{t,\Foot(E_a)} E_a P.
\]
Thus distinct footprint sectors are orthogonally distinguishable, while the logical state is protected from the information revealed by the measurement.  This is a hypothesis on the measurement model and on the chosen error representatives.  It is not automatic for an arbitrary named error set before one has resolved the syndrome sectors or diagonalized the corresponding Knill--Laflamme matrix.  Under this hypothesis one obtains the two requirements
\begin{align}
  P E_a^\dagger E_b P&=0,
  &&\Foot(E_a)\neq \Foot(E_b),
  \label{eq:foot-orthogonal}\\
  P E_a^\dagger E_b P&=\lambda_{ab}^{(s)}P,
  &&E_a,E_b\in\cE_s\text{ and }E_a^{-1}E_b\text{ locally neutral}.
  \label{eq:foot-KL}
\end{align}
with the obvious modification when the errors are not invertible and one replaces \(E_a^{-1}E_b\) by the composite \(E_a^\dagger E_b\).  Equation \eqref{eq:foot-orthogonal} says that different footprints are detectable without measuring the encoded state.  Equation \eqref{eq:foot-KL} says that the residual ambiguity inside a fixed footprint fibre is harmless exactly when the composite ambiguity acts as a scalar on the code.  Together these two formulas are the Knill--Laflamme condition reorganized after decomposing the error set by locally visible field-theoretic data.

In the surface-code case, \(\cE_s\) is the set of chains with boundary \(s\).  If \(c,c'\in\cE_s\), then \(c-c'\) is a cycle.  The cycle is harmless if it is a stabilizer boundary, and harmful if it represents a nontrivial homology class.  In a modular-category code, the same sentence becomes: two defect histories with the same footprint differ by a closed or globally neutral defect network, which is harmless precisely when it evaluates locally to a scalar and harmful when it acts as a nontrivial Wilson or fusion-channel operator.  This is the local-to-global ambiguity that the proposed decoder is designed to resolve.

A physical recovery map adds one further layer.  The field-theoretic decoder may output a class \([\widehat E(s)]\) of compatible histories rather than a unique microscopic operator.  A device-level recovery must then choose a representative physical operation \(R_s\).  Exact correction in the ordinary sense asks that
\[
  R_s E_a P = u_{a,s} P
  \qquad(E_a\in\cE_s)
\]
up to harmless scalars or unitary transformations on an auxiliary syndrome register.  The field-theoretic formalism therefore does not replace the operational recovery map.  It organizes the hypothesis space on which the recovery map is based.

\begin{remark}[Operator interpretation of the fibrewise theorem]\label{rem:operator-fibrewise-KL}
The ideal footprint-measurement hypothesis above is exactly the operator-level situation already covered by Theorem \ref{thm:footprint-algebra-KL}.  If \(\cE=\bigsqcup_s\cE_s\) and
\[
  \Pi_t E_aP=\delta_{t,\Foot(E_a)}E_aP,
\]
then orthogonality gives \(PE_a^\dagger E_bP=0\) for distinct measured fibres, while exact correction within a fixed fibre is equivalent to
\[
  PE_a^\dagger E_bP=\lambda^{(s)}_{ab}P.
\]
The constructive recovery is Proposition \ref{prop:conditional-recovery-fibre}: diagonalize the positive Gram matrix \(\Lambda^{(s)}\), obtain mutually orthogonal error images, and apply the inverse partial isometries after observing \(s\).  We do not repeat the theorem and proof in operator notation.  The operator formulation is the Hilbert-space shadow of the categorical statement, and the mathematical gain of the footprint language is the prior organization of the error family by measured field-theoretic sectors.
\end{remark}

\section{Errors as defect networks and syndromes as fusion data}

We next describe errors and syndromes in the field-theoretic picture.  Let \(\cC\) be a unitary modular tensor category with simple objects \(a,b,c,\ldots\), tensor unit \(\one\), fusion coefficients \(N_{ab}^c\), quantum dimensions \(d_a\), and duals \(a^*\).  The discussion below has analogues for fusion categories, rational conformal field theories, and extended TQFTs with defects.

A local error insertion may be represented by an object label \(a\), or by a small segment of defect line carrying label \(a\).  In a disk, topological charge conservation requires that a collection of labels \(a_1,\ldots,a_n\) have a possible total charge \(c\), encoded by
\begin{equation}\label{eq:fusion-space}
V_{a_1\cdots a_n}^c = \Hom(c,a_1\otimes\cdots\otimes a_n).
\end{equation}
A syndrome measurement in this setting asks for the total charge \(c\), or more finely for an intermediate fusion channel in a chosen fusion tree.

In footprint language, if \(E\) is an error insertion or defect network, then \(\Foot(E)\) is the local charge, endpoint, boundary-sector, or fusion-channel datum induced on the boundary of a small neighbourhood of the support of \(E\).  The footprint is boundary-like without necessarily being a literal boundary component of spacetime; it is the local datum through which the ambient field theory detects the insertion.

For a stabilizer code, the syndrome is a vector of signs.  For a modular category, the footprint is categorical and the syndrome is its measured value.  Pair creation of \(a\) and \(a^*\) in a contractible disk has total vacuum charge and may be locally invisible unless the pair is separated.  A single nontrivial charge cannot appear in isolation on a closed surface, but endpoints of an open error string can carry detectable charge.  Thus the categorical replacement for ``boundary of an error chain'' is ``the collection of nontrivial charges or defect endpoints left after local fusion.''

Choose a small regular neighborhood \(U\) of the support of an error network \(E\).  The complement sees the network through labelled data induced on \(\partial U\): total charge, endpoints of open strings, defect-junction labels, or boundary-condition changes.  In a fully extended theory, excision makes this picture literal: the error neighborhood and its complement are glued along \(\partial U\), and the measurement protocol accesses selected data in the boundary state assigned there.

This local record need not be complete.  Two networks inducing the same measured boundary sector may still differ by internal fusion multiplicity, braiding inside \(U\), or a global route taken before entering \(U\).  These are latent variables for the decoder.  In a semisimple theory the boundary state decomposes into simple sectors; with fusion multiplicities, one must decide whether multiplicity labels are measured syndrome data, unobserved environment data, leakage, or protected logical information.

There is a corresponding gluing picture.  If \(E_1\) and \(E_2\) are two error histories whose neighborhoods have compatible boundary labels, then composing the histories corresponds to gluing their footprint boundaries and summing over intermediate sectors.  The same gluing operation appears in TQFT state sums and in conformal-block factorization.  Footprint formation is therefore not an auxiliary operation imposed on the code, but a local boundary operation already present in the field theory.

The central protection mechanism is the following standard field-theoretic principle:

\begin{proposition}[Local neutrality acts scalarly]\label{prop:local-scalar}
Let \(Z\) be a semisimple unitary topological field theory associated to a modular tensor category \(\cC\).  Let \(\Sigma\) be a decorated surface and let \(D\subset \Sigma\) be a contractible disk disjoint from the marked boundary data.  Suppose a defect network \(N\subset D\) has total vacuum charge at \(\partial D\).  Then the local evaluation of \(N\) is an endomorphism of the tensor unit,
\[
  \operatorname{ev}(N)\in \End_{\cC}(\one)\cong\bC,
\]
and the operator induced by inserting \(N\) acts on the topological state space \(Z(\Sigma,\cD)\) by this scalar.
\end{proposition}

\begin{proof}
Choose a small circle \(\partial D\) enclosing the network \(N\).  Locality, or equivalently the gluing axiom of the field theory, says that insertion of \(N\) factors through the state space assigned to this artificial boundary.  The hypothesis that the total boundary charge is vacuum means that the local contribution of \(N\) is a morphism from the tensor unit to itself.  In the categorical notation this is
\[
  \operatorname{ev}(N)\in\End_{\cC}(\one).
\]
Since \(\cC\) is a unitary fusion category with simple tensor unit, \(\End_{\cC}(\one)\cong\bC\).  Thus the disk containing \(N\) contributes only a complex scalar.

Gluing the disk back into \(\Sigma\) pairs this scalar local morphism with the complementary surface.  There is no remaining boundary label or multiplicity index through which the local insertion could act nontrivially on the external state space.  Hence the induced operator on \(Z(\Sigma,\cD)\) is multiplication by that scalar.  In ribbon-graph language, this is the usual evaluation of a closed neutral diagram inside a ball or disk before the diagram is removed from the ambient surface.
\end{proof}

The scalar mechanism can be seen concretely in the fusion-space notation.  Suppose that the boundary of a small disk carries total charge \(c\).  The local state space is a direct sum of spaces of the form
\[
  V_{a_1\cdots a_n}^{c}=\Hom(c,a_1\otimes\cdots\otimes a_n).
\]
If \(c=\one\) and the disk is otherwise isolated from marked data, a closed neutral network defines an element of
\[
  \End(\one)\cong \bC
\]
in a simple semisimple category.  Hence it is a scalar.  If instead \(c\neq \one\), the same disk has a nontrivial footprint: the complement sees charge \(c\) on the artificial boundary.  If the channel \(c\) is vacuum but a multiplicity space occurs, say
\[
  \dim V_{a_1\cdots a_n}^{\one}=m>1,
\]
then a local insertion may act by an \(m\times m\) matrix on this multiplicity space.  An operational code must then make a decision.  It may measure the multiplicity label and include it in the footprint; it may treat the multiplicity space as leakage; or it may encode in it deliberately.  The scalar conclusion of Proposition \ref{prop:local-scalar} therefore applies when multiplicity data are absent, resolved by the footprint, or excluded from the encoded subsystem.  It does not remove the need to specify how unresolved multiplicities are treated.

This distinction is often invisible in abelian stabilizer examples because the relevant local sectors are one-dimensional.  In a nonabelian theory it becomes important.  A coarse footprint might record only the total charge \(c\), while a refined footprint records \((c,\alpha)\), where \(\alpha\) is an intermediate channel or multiplicity label.  If the measurement is coarse, correctability requires that all unmeasured \(\alpha\)-ambiguities satisfy a Knill--Laflamme scalar condition.  If the measurement is refined, the decoder may condition separately on each \((c,\alpha)\).  Thus the footprint is not a purely mathematical label; it is a specification of which part of the local field-theoretic boundary data is actually observed.

Semisimplicity enters through the one-dimensional vacuum boundary sector.  If a theory assigns a higher-dimensional local state space to the same apparent boundary label, a locally neutral insertion may act nontrivially there.  QEC can still accommodate this possibility, but only after the additional sector is measured, energetically suppressed, treated as leakage, or included in the physical noise model.  Scalarity is therefore a statement about the chosen protected sector and allowed local errors, not a property of arbitrary field theories.

The same proof may be read as a local-to-global argument.  A contractible neutral network can be enclosed by a small circle, evaluated there, and then forgotten by the rest of the surface.  A noncontractible neutral network cannot be enclosed in this way without cutting through nontrivial topology or through marked sectors.  This is exactly the difference between a stabilizer-like operation and a logical Wilson line.  The field-theoretic version of distance should therefore measure how difficult it is for an allowed error network to evade all such local enclosures and nevertheless return a neutral footprint.

Proposition \ref{prop:local-scalar} is the field-theoretic analogue of the Knill--Laflamme scalar condition.  In the algebraic theory, correctable composites satisfy \(P E_a^\dagger E_b P=\lambda_{ab}P\).  In the topological theory, locally neutral networks in contractible regions evaluate to scalars.  Logical operators must therefore arise from networks that cannot be reduced to local neutral evaluations: noncontractible loops, strings connecting appropriate boundaries, nontrivial domain walls, or changes of global fusion channel.

\begin{remark}
The proposition is stated in a semisimple modular setting.  If multiplicity spaces occur, then a locally neutral disk insertion is scalar only after fixing the appropriate simple vacuum sector; more generally it may act on a local multiplicity space.  The point is structural: a field-theoretic code must specify which local sectors are treated as gauge, syndrome, leakage, or logical degrees of freedom.  Nonsemisimple TQFTs, logarithmic CFTs, higher-categorical defect theories, and nonunitary theories may also exhibit nilpotent or nonsemisimple local sectors.  Those cases are mathematically rich, but for quantum error correction the unitary semisimple setting is the most immediate starting point.
\end{remark}

\begin{proposition}[Multiplicity-resolved footprint sectors]\label{prop:multiplicity-footprints}
Let \(\cC\) be a semisimple unitary fusion category, and fix labels \(a_1,\ldots,a_n\) on a contractible cluster.  Suppose the local state space decomposes as
\[
  a_1\otimes\cdots\otimes a_n
  \cong \bigoplus_{c\in\operatorname{Irr}(\cC)} V_{a_1\cdots a_n}^{c}\otimes c.
\]
A footprint measurement which records only the total charge \(c\) has projectors onto the summands indexed by \(c\).  It is an exact local syndrome measurement for a code only if every unobserved operator acting inside the multiplicity space \(V_{a_1\cdots a_n}^{c}\) is either scalar on the protected sector or else is included in the error model and corrected by a fibrewise Knill--Laflamme condition.  Equivalently, if multiplicity labels are not measured, they remain latent variables inside the footprint fibre.
\end{proposition}

\begin{proof}
The displayed decomposition is the semisimple decomposition of the tensor product into simple total-charge summands and multiplicity spaces.  A coarse charge measurement applies the projector onto the full isotypic summand with fixed label \(c\).  It therefore sees the simple object \(c\), but it does not distinguish vectors or operators inside the factor \(V_{a_1\cdots a_n}^{c}\).

Consequently, two local histories which differ only by an endomorphism of this multiplicity space have the same coarse footprint.  From the point of view of the decoder they lie in the same footprint fibre.  They are harmless exactly when the corresponding compressed composites act as scalars on the protected sector, i.e. when the fibrewise Knill--Laflamme equations hold after projection to the code.  If the equations fail, the unmeasured multiplicity variable carries information that is not being recorded by the syndrome.  The available remedies are precisely the alternatives stated in the proposition: refine the footprint measurement, suppress the multiplicity dynamically, treat it as leakage, or incorporate it deliberately as logical or gauge data.
\end{proof}

\begin{corollary}[Coarse-graining criterion]\label{cor:coarse-graining-criterion}
Let a fine footprint algebra have pairwise orthogonal projectors \(Q_r\), and let a coarse measured algebra be obtained by grouping them as
\[
  \Pi_s=\sum_{r\in R_s}Q_r .
\]
A coarse sector \(s\) is harmless for a code and error family exactly when all error representatives whose fine footprints lie in the same block \(R_s\) satisfy the Knill--Laflamme scalar equations after compression to the code.  Thus coarse-graining is allowed precisely when the fine information that has been forgotten is not logical information and does not distinguish uncorrectable representatives.
\end{corollary}

\begin{proof}
The coarse measurement identifies all fine sectors inside the same block \(R_s\).  Hence representatives with fine labels \(r,r'\in R_s\) are indistinguishable to the measured syndrome.  They must therefore be treated as lying in one footprint fibre.  Applying Theorem~\ref{thm:footprint-algebra-KL} to the coarse algebra gives exactly the stated scalar equations for all such pairs.  Conversely, if these equations hold in each coarse block, then the conditional recovery of Proposition~\ref{prop:conditional-recovery-fibre} corrects each block separately.  No further condition is imposed by the fine labels, since they are not measured.
\end{proof}

This corollary is a useful design rule.  A large diagnostic algebra may contain far more information than one can safely or practically measure.  The measured syndrome algebra should be a coarse quotient of that diagnostic algebra only along directions that are invisible to the logical subsystem after correction.  Stabilizer degeneracy is the abelian example: many Pauli errors with different microscopic supports share the same syndrome because their differences are stabilizers.  In a fusion category, the analogous degeneracies may come from locally neutral bubbles, associator-equivalent networks, or unmeasured multiplicity variables that act trivially on the protected sector.

A field-theoretic code therefore inherits a natural hierarchy of errors:
\begin{center}
\renewcommand{\arraystretch}{1.15}
\begin{tabular}{>{\raggedright\arraybackslash}p{0.43\linewidth}c>{\raggedright\arraybackslash}p{0.43\linewidth}}
locally neutral contractible networks &:& act as scalars or stabilizers\\
networks with detectable endpoints &:& produce observable footprints or syndromes\\
globally nontrivial neutral networks &:& act as logical operators\\
networks changing boundary or defect sectors &:& implement gates, leakage, or code deformation
\end{tabular}
\end{center}
The footprint only records the locally visible data.  The syndrome is the measured footprint.  Decoding requires choosing a compatible global class.  This is the topological origin of degeneracy in decoding: many networks can share the same footprint, but differ by nontrivial topology.

\section{ZX-calculus as the stabilizer shadow}

The ZX-calculus is a graphical language for qubit processes built from two complementary families of spiders, usually drawn green and red, corresponding to \(Z\)- and \(X\)-type classical structures.  It is complete for stabilizer quantum mechanics and extends to broader fragments of quantum computation \cite{Backens2014,CoeckeKissinger2017}.  For quantum error correction, the ZX-calculus is especially suggestive because surface-code lattice surgery has a direct ZX interpretation.  De Beaudrap and Horsman identify rough and smooth merges and splits with red and green spiders satisfying dagger special Frobenius algebra relations \cite{deBeaudrapHorsman2020}.  Chancellor--Kissinger--Zohren--Roffe--Horsman use ZX-based graphical structures to design and verify stabilizer error-correcting codes via coherent parity checks \cite{ChancellorEtAl2016}.  More recently, Bombin--Litinski--Nickerson--Pastawski--Roberts use ZX instrument networks to unify circuit-based, measurement-based, fusion-based, and Floquet models of stabilizer fault tolerance \cite{BombinLitinskiNickersonPastawskiRoberts2023}.

In the present article, we use ZX-calculus as a stabilizer or Pauli shadow of selected field-theoretic sectors.  The comparison is model-dependent and applies only when an appropriate lift has been specified.  The schematic progression already displayed in \eqref{eq:containment} is
\begin{equation}\label{eq:diagrammatic-hierarchy}
\boxed{
\text{ZX diagrams}
\rightsquigarrow
\text{defect/string-net diagrams}
\rightsquigarrow
\text{extended TQFT/CFT amplitudes}.
}
\end{equation}

The four-\(\sigma\) Ising sector gives one precise instance of this slogan.
\begin{proposition}[Exact one-qubit Clifford shadow in the Ising sector]\label{prop:ising-clifford-shadow}
On \(\Hom(\one,\sigma^{\otimes 4})\), choose the \((12)|(34)\) fusion basis.  Then, up to an overall phase,
\[
  Z_f^{(12)}=Z,
  \qquad
  F=H,
  \qquad
  B_1=S=\begin{pmatrix}1&0\\0&i\end{pmatrix},
\]
where \(F\) is the nontrivial Ising recoupling matrix and \(B_1\) is the braid of the first pair.  Consequently
\[
  FZ_f^{(12)}F^{-1}=X,
\]
and the projective image generated by \(F\) and \(B_1\) is the one-qubit Clifford group.
\end{proposition}

\begin{proof}
The Ising data computed explicitly in Section~\ref{sec:ising-example} give
\[
  F=\frac{1}{\sqrt2}\begin{pmatrix}1&1\\1&-1\end{pmatrix}=H,
  \qquad
  B_1=e^{-i\pi/8}\begin{pmatrix}1&0\\0&i\end{pmatrix}=e^{-i\pi/8}S,
\]
and the \((1,2)\)-pair-charge observable is diagonal with eigenvalues \(\pm1\), hence equals \(Z\) in this basis.  Since \(HZH=X\) and \(H,S\) generate the one-qubit Clifford group projectively, the assertions follow.
\end{proof}

This proposition is the precise stabilizer content used later.  The remaining comparisons in this section are motivational and indicate structures that a model-specific lift would have to realize.

At the ZX level, we have qubits, Pauli measurements, stabilizer propagation, Clifford transformations, and lattice-surgery corrections. Comparatively, at the defect/string-net level, we have anyon lines, Wilson operators, fusion vertices, domain walls, condensable boundaries, and topological charge measurements.  Finally, at the TQFT/CFT level, one assigns amplitudes, conformal blocks, or path-integral weights to the diagrams.

To indicate what a stronger model-specific statement would look like, Figure \ref{fig:shadowtriangle} displays a schematic triangle.  Once a lift \(\iota\) has actually been constructed for a selected stabilizer fragment, commutativity is the condition one would ask it to satisfy.  No such universal lift is asserted here.

\begin{figure}[htbp]
\centering
\[
\begin{tikzcd}[column sep=huge, row sep=large]
\mathsf{ZX}_{\mathrm{stab}}
\arrow[r, "\iota"]
\arrow[dr, swap, "\ev_{\mathrm{stab}}"]
&
\mathsf{Def}_{\cF}
\arrow[d, "\ev_{\cF}"]
\\
&
\mathsf{Hilb}
\end{tikzcd}
\]
\caption{The compatibility triangle desired of a model-dependent stabilizer lift.  If a lift \(\iota\) is specified for an appropriate ZX fragment, one asks that direct stabilizer evaluation agree with evaluation after passage to the defect/string-net calculus.  The figure records the target compatibility condition for a specified lift \(\iota\).}
\label{fig:shadowtriangle}
\end{figure}
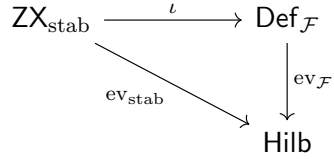

A second useful picture is supplied by the spiders themselves.  Figure \ref{fig:zxspiders} records the minimal local grammar that underlies many stabilizer and lattice-surgery manipulations: green and red spiders model complementary \(Z\)- and \(X\)-type merge/split structures, and the spider-fusion rule is the simplest diagrammatic shadow of topological composition.  In the present context, however, the spiders should not remain merely stabilizer-theoretic icons.  They should be read as the qubit-level shadows of genuinely field-theoretic operations.  Figure \ref{fig:liftspider} therefore gives a schematic lift from spider calculus to pair-of-pants and defect-junction pictures.

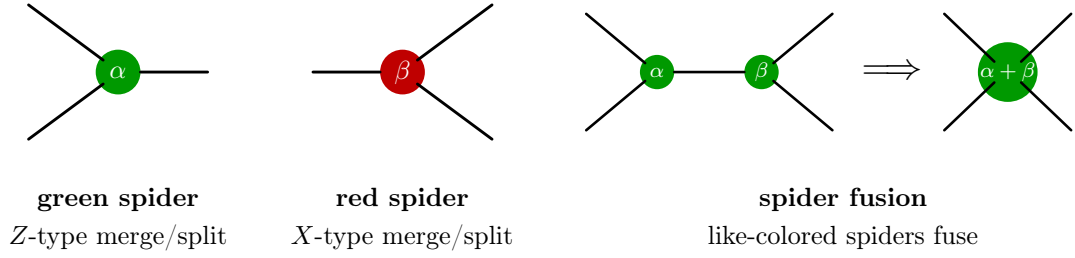
\begin{figure}[t]
\centering
\setlength{\tabcolsep}{0.8em}
\renewcommand{\arraystretch}{1.1}
\resizebox{0.98\textwidth}{!}{%
\begin{tabular}{ccc}
\begin{tikzpicture}[baseline=(current bounding box.center), line cap=round, line join=round, every node/.style={font=\small}]
  \path[use as bounding box] (-1.6,-1.45) rectangle (1.6,1.25);
  \fill[green!60!black] (0,0) circle (0.30);
  \draw[line width=1.1pt] (-1.20,0.90) -- (-0.20,0.16);
  \draw[line width=1.1pt] (-1.20,-0.90) -- (-0.20,-0.16);
  \draw[line width=1.1pt] (1.20,0.00) -- (0.30,0.00);
  \node at (0,0) {\color{white}\(\alpha\)};
\end{tikzpicture}
&
\begin{tikzpicture}[baseline=(current bounding box.center), line cap=round, line join=round, every node/.style={font=\small}]
  \path[use as bounding box] (-1.6,-1.45) rectangle (1.6,1.25);
  \fill[red!75!black] (0,0) circle (0.30);
  \draw[line width=1.1pt] (-1.20,0.00) -- (-0.30,0.00);
  \draw[line width=1.1pt] (1.20,0.90) -- (0.20,0.16);
  \draw[line width=1.1pt] (1.20,-0.90) -- (0.20,-0.16);
  \node at (0,0) {\color{white}\(\beta\)};
\end{tikzpicture}
&
\begin{tikzpicture}[baseline=(current bounding box.center), line cap=round, line join=round, every node/.style={font=\small}]
  \path[use as bounding box] (-1.9,-1.45) rectangle (5.5,1.25);
  \fill[green!60!black] (-0.70,0) circle (0.23);
  \fill[green!60!black] (0.70,0) circle (0.23);
  \draw[line width=1.0pt] (-1.65,0.78) -- (-0.86,0.12);
  \draw[line width=1.0pt] (-1.65,-0.78) -- (-0.86,-0.12);
  \draw[line width=1.0pt] (-0.47,0) -- (0.47,0);
  \draw[line width=1.0pt] (1.65,0.78) -- (0.86,0.12);
  \draw[line width=1.0pt] (1.65,-0.78) -- (0.86,-0.12);
  \node at (-0.70,0) {\color{white}\scriptsize \(\alpha\)};
  \node at (0.70,0) {\color{white}\scriptsize \(\beta\)};
  \node at (2.45,0) {\Large\(\Longrightarrow\)};
  \fill[green!60!black] (4.00,0) circle (0.4);
  \draw[line width=1.0pt] (3.15,0.78) -- (3.82,0.14);
  \draw[line width=1.0pt] (3.15,-0.78) -- (3.82,-0.14);
  \draw[line width=1.0pt] (4.85,0.78) -- (4.18,0.14);
  \draw[line width=1.0pt] (4.85,-0.78) -- (4.18,-0.14);
  \node at (4.00,0) {\color{white}\scriptsize \(\alpha+\beta\)};
\end{tikzpicture}
\\
\small\textbf{green spider} & \small\textbf{red spider} & \small\textbf{spider fusion} \\
\small\(Z\)-type merge/split & \small\(X\)-type merge/split & \small like-colored spiders fuse \\
\end{tabular}%
}
\caption{Basic colored spider diagrams in the ZX-calculus.  Left: a green spider encoding a \(Z\)-type merge/split structure.  Middle: a red spider encoding the complementary \(X\)-type structure.  Right: the spider-fusion rewrite for like-colored spiders, shown in the green case.  In the present article these spiders serve as stabilizer shadows of more general topological and conformal operations.}
\label{fig:zxspiders}
\end{figure}

\begin{figure}[t]
\centering
\begin{minipage}[t]{0.47\textwidth}
\centering
\textbf{ZX shadow \(\rightsquigarrow\) pair-of-pants cobordism}\par\medskip
\begin{tikzpicture}[x=1cm,y=1cm,line cap=round,line join=round,every node/.style={font=\scriptsize}]
  \path[use as bounding box] (-3.3,-1.9) rectangle (3.0,1.9);
  \fill[green!60!black] (-2.2,0) circle (0.23);
  \draw[line width=1.0pt] (-3.05,0.78) -- (-2.38,0.12);
  \draw[line width=1.0pt] (-3.05,-0.78) -- (-2.38,-0.12);
  \draw[line width=1.0pt] (-1.25,0) -- (-1.98,0);
  \node at (-2.2,0) {\color{white}\(\alpha\)};
  \node[below] at (-2.2,-1.10) {green spider};
  \node at (-0.15,0) {\Large\(\rightsquigarrow\)};
  \draw[line width=1.0pt] (1.05,1.05) .. controls (0.92,0.63) and (0.95,0.25) .. (1.06,-0.02);
  \draw[line width=1.0pt] (1.95,1.05) .. controls (2.08,0.63) and (2.05,0.25) .. (1.94,-0.02);
  \draw[line width=1.0pt] (1.06,-0.02) .. controls (1.22,-0.68) and (1.43,-1.00) .. (1.50,-1.15);
  \draw[line width=1.0pt] (1.94,-0.02) .. controls (1.78,-0.68) and (1.57,-1.00) .. (1.50,-1.15);
  \draw[line width=1.0pt] (1.05,1.05) arc[start angle=180,end angle=360,radius=0.30];
  \draw[line width=1.0pt] (1.95,1.05) arc[start angle=180,end angle=360,radius=0.30];
  \draw[line width=1.0pt] (1.20,-1.15) arc[start angle=180,end angle=360,radius=0.30];
  \node[below,align=center] at (1.50,-1.48) {pair-of-pants\\fusion vertex};
\end{tikzpicture}

\vspace{0.35cm}
\smallskip
\small Green spiders are interpreted as stabilizer shadows of fusion or merge operations in the ambient field theory.
\end{minipage}
\hfill
\begin{minipage}[t]{0.47\textwidth}
\centering
\textbf{ZX shadow \(\rightsquigarrow\) defect junction / condensable boundary}\par\medskip
\begin{tikzpicture}[x=1cm,y=1cm,line cap=round,line join=round,every node/.style={font=\scriptsize}]
  \path[use as bounding box] (-3.3,-1.9) rectangle (3.8,1.9);
  \fill[red!75!black] (-2.2,0) circle (0.23);
  \draw[line width=1.0pt] (-3.05,0) -- (-2.42,0);
  \draw[line width=1.0pt] (-1.32,0.78) -- (-2.04,0.12);
  \draw[line width=1.0pt] (-1.32,-0.78) -- (-2.04,-0.12);
  \node at (-2.2,0) {\color{white}\(\beta\)};
  \node[below] at (-2.2,-1.10) {red spider};
  \node at (-0.25,0) {\Large\(\rightsquigarrow\)};
  \draw[line width=1.0pt] (1.45,0) -- (2.45,0);
  \draw[line width=1.0pt] (1.45,0) -- (0.85,0.88);
  \draw[line width=1.0pt] (1.45,0) -- (0.85,-0.88);
  \draw[line width=1.0pt,dashed] (2.45,0) -- (3.10,0);
  \fill (1.45,0) circle (1.8pt);
  \node[above left] at (0.82,0.88) {sector \(a\)};
  \node[below left] at (0.82,-0.88) {sector \(b\)};
  \node[above] at (2.45,0.08) {sector \(c\)};
  \node[below,align=center] at (1.92,-1.32) {defect junction\\condensable boundary};
\end{tikzpicture}

\smallskip
\small Red spiders are interpreted as shadows of complementary defect-junction, measurement, or wall operations.
\end{minipage}
\caption{A schematic field-theoretic lift of ZX spiders.  Left: a green spider is interpreted as the stabilizer shadow of a pair-of-pants cobordism or fusion vertex.  Right: a red spider is interpreted as the shadow of a complementary defect-junction or condensable-boundary operation.  The lift is not unique; it is model-dependent and encodes the idea that familiar ZX generators arise from more geometric and categorical operations in the ambient field theory.}
\label{fig:liftspider}
\end{figure}

Figures \ref{fig:shadowtriangle}--\ref{fig:liftspider} summarize three aspects of the comparison with ZX.  Figure \ref{fig:shadowtriangle} records the compatibility condition that a concrete lift would have to satisfy.  Figure \ref{fig:zxspiders} records the local spider grammar.  Figure \ref{fig:liftspider} shows how the same local moves may be represented by pair-of-pants operations, fusion vertices, domain-wall junctions, or condensable-boundary manipulations.

The conceptual lift is summarized in Table \ref{tab:lift}.

\begin{table}[htbp]
\centering
\renewcommand{\arraystretch}{1.25}
\begin{tabular}{>{\raggedright\arraybackslash}p{0.23\linewidth}>{\raggedright\arraybackslash}p{0.30\linewidth}>{\raggedright\arraybackslash}p{0.35\linewidth}}
\toprule
\textbf{ZX-calculus} & \textbf{Topological code} & \textbf{Field-theoretic code} \\
\midrule
wire & qubit or logical patch & anyon line, boundary sector, or object of a tensor category \\
spider & merge/split/check operation & pair-of-pants cobordism, fusion vertex, or defect junction \\
color & \(X/Z\)-type complementary structure & dual boundary type, condensable algebra, or complementary charge sector \\
rewrite rule & fault-tolerant equivalence or circuit identity & topological invariance under defect-network deformation \\
phase & Pauli/Clifford phase data & topological spin, charge, conformal weight, or defect holonomy \\
measurement branch & syndrome outcome & measured footprint: fusion channel, topological charge, or field-insertion sector \\
correction & Pauli-frame update & defect-history selection modulo topological equivalence \\
\bottomrule
\end{tabular}
\caption{The lifting of ZX-calculus features to their corresponding structures and operations in topological and field-theoretic QEC.}
\label{tab:lift}
\end{table}

 Here, Green and red spiders are algebraic shadows of pair-of-pants operations.  ZX rewrite rules are stabilizer-level shadows of topological invariance.  Pauli phases are shadows of richer charge, spin, and conformal data, while measurement branches are shadows of fusion-channel decompositions.  However, the table is motivational and is not an assertion that every ZX diagram embeds fully faithfully into every field theory. A precise functorial lift will depend on a number of choices: the phase, its admissible boundaries, the condensable algebra data, and the fragment of ZX under consideration.

The ZX-calculus plays two roles here.  First, it is an established graphical language for stabilizer fault tolerance, so any field-theoretic extension should recover its reasoning in the appropriate abelian qubit sector.  Second, it sets a practical standard for diagrammatic calculi: equivalences should be local, compositional, and checkable from the diagram.  A nonabelian or conformal calculus should provide analogous control for defect networks and conformal-block decompositions.

\begin{principle}[ZX shadow principle]\label{principle:ZXshadow}
A field-theoretic quantum code should admit, whenever an appropriate qubit stabilizer sector is selected, a shadow diagrammatics compatible with ZX reasoning for the corresponding stabilizer processes.  Conversely, ZX diagrams should be regarded as a qubit-stabilizer collapse of a richer, model-dependent calculus of defects, boundaries, and field-theoretic amplitudes.
\end{principle}

The word ``collapse'' does not mean that the ZX-calculus is replaced.  It means that a selected stabilizer sector forgets some of the labels and amplitudes present in the ambient field-theoretic calculus.  In a purely stabilizer setting, diagrammatic equivalence is often enough: ZX diagrams related by the rewrite rules denote the same process.  In the field-theoretic setting, inequivalent defect histories with the same observed footprint may also carry different amplitudes.  The passage is therefore
\begin{equation}\label{eq:zx-to-cft}
\begin{gathered}
\text{ZX rewrite equivalence}
\quad\leadsto\quad
\text{topological defect equivalence}\\
\leadsto\quad
\text{conformal-block-weighted inference}.
\end{gathered}
\end{equation}
Conformal field theory enters at this step by assigning geometry-dependent weights to histories that remain distinct after the topological quotient.

\section{Conformal enhancement and field-theoretic decoder weights}

A topological field theory is insensitive to local geometry, and this insensitivity underlies topological protection: local deformations of a defect network do not change the encoded information.  Decoding, however, may benefit from geometric information.  When the noise process has geometric or energetic structure, the relative likelihood of competing histories may depend on distances, cross-ratios, local curvature, or boundary geometry.  Conformal field theory offers an intermediate regime, retaining strong locality and functoriality while permitting analytic dependence on conformal geometry.

In rational conformal field theory, correlation functions decompose into conformal blocks.  For primary insertions \(\phi_{a_i}(z_i)\), one has schematically
\begin{equation}\label{eq:blocks}
\langle \phi_{a_1}(z_1)\cdots \phi_{a_n}(z_n)\rangle
=\sum_{\alpha,\beta} C_{\alpha\beta}\,\mathcal F_\alpha(z_1,\ldots,z_n)\,
\overline{\mathcal F_\beta(z_1,\ldots,z_n)},
\end{equation}
where \(\alpha\) and \(\beta\) range over internal fusion channels or conformal blocks, and the matrix \(C_{\alpha\beta}\) depends on the chosen full CFT and Hermitian pairing.  In diagonal unitary examples, or after choosing a basis diagonalizing the relevant Hermitian form, this reduces to a sum of squared block norms.  In a chiral theory, one works directly with the vector space of conformal blocks.  In either case, the block label is a natural footprint-compatible hidden fusion-history variable.

This motivates the following definition.

\begin{definition}[Conformal-block likelihood datum]\label{def:CBdecoder}
A \emph{conformal-block likelihood datum} consists of a conformal-block bundle \(\cV\to B\) over a configuration or moduli space of marked curves, a projectively flat connection \(\nabla\), a Hermitian form \(h_z\) on each fibre, a collection of footprint projectors \(P_s(z)\), and a prior or detector likelihood model for the relevant histories.  In a local trivialization, with a chosen fusion-tree basis labelled by channels \(\alpha\), this datum gives normalized weights
\begin{equation}\label{eq:CBposterior}
\Prob(\alpha\mid z,s)
=\frac{w_\alpha(z;s)}{\sum_\beta w_\beta(z;s)},
\end{equation}
where \(z=(z_1,\ldots,z_n)\), \(s\) is the observed footprint, and the weights are built from conformal-block norms, full correlation-function contributions, or a specified physical noise model coupled to the conformal data.  In the simplest block-norm model,
\begin{equation}\label{eq:blocknorm}
 w_\alpha(z;s)=\|\mathcal F_\alpha(z)\|^2_{h_z}\,\chi_s(\alpha),
\end{equation}
with \(\chi_s(\alpha)=1\) when the channel \(\alpha\) is compatible with footprint \(s\), and \(0\) otherwise.
\end{definition}

The Hermitian structure is part of the likelihood datum.  In a unitary modular functor it is supplied by the unitary structure, while in a chiral CFT presentation it depends on normalization conventions and on the relation between chiral blocks and full correlators.  The coordinate-free version of the definition is useful, especially because conformal-block spaces come with many natural bases.  Let \(\cV_z\) be the conformal-block space at the configuration \(z=(z_1,\ldots,z_n)\), equipped with a Hermitian form, and let \(P_s\) be the orthogonal projector onto the subspace compatible with the measured footprint \(s\).  If \(v(z)\in\cV_z\) is the block vector determined by the external insertions and the chosen state, then the intrinsic footprint weight is
\begin{equation}\label{eq:intrinsic-block-weight}
  w_s(z)=\langle v(z),P_s v(z)\rangle.
\end{equation}
When a fusion tree diagonalizes the footprint, \(P_s\) becomes a sum of coordinate projectors and \eqref{eq:intrinsic-block-weight} reduces to a sum of squared block components.

There are physical settings in which this Hermitian interpretation is more than a formal ansatz.  For Moore--Read Pfaffian quasihole states, Bonderson--Gurarie--Nayak prove, under the plasma-screening assumptions of their analysis and for well-separated quasiholes, that the conformal-block wavefunctions are orthogonal with equal constant norms in the physical inner product \cite{BondersonGurarieNayak2011}.  In such a realization, once a prepared state is expanded in that orthonormal channel basis, squared channel amplitudes have the ordinary Born interpretation.  This does not make every pointwise formula of the form \eqref{eq:blocknorm} a universal physical law: the preparation amplitudes, detector model, regime of separation, and physical inner product still have to be specified.  It does show that block-resolved probabilistic weights occur in an explicit Ising-type quantum Hall setting.

\begin{proposition}[Basis covariance of block weights]\label{prop:block-covariance}
Let \(U:\cV_z\to\cV_z\) be a unitary change of conformal-block basis, such as a unitary \(F\)-move in a unitary modular functor.  If the block vector and footprint projector transform simultaneously by
\[
  v(z)\longmapsto Uv(z),
  \qquad
  P_s\longmapsto UP_sU^{-1},
\]
then the weight \(w_s(z)=\langle v(z),P_s v(z)\rangle\) is unchanged.
\end{proposition}

\begin{proof}
The statement is the invariance of a matrix coefficient under simultaneous conjugation.  Indeed,
\[
  \langle Uv,UP_sU^{-1}Uv\rangle
  =\langle Uv,UP_s v\rangle
  =\langle v,P_s v\rangle,
\]
where the second equality uses unitarity of \(U\).  Thus the number \(w_s(z)\) depends on the subspace selected by the footprint projector and on the vector \(v(z)\), not on the coordinates used to describe them.  If one changes only the block components and not the projector, one has changed the question being asked.  This is why recoupling must be accompanied by the corresponding transformation of the measured footprint sector.
\end{proof}

\begin{proposition}[Footprint factorization in sewing coordinates]\label{prop:footprint-factorization}
Let \(\cV_{\vec X}\to M_{0,n}\) be a genus-zero conformal-block bundle associated to a unitary modular functor with labels \(\vec X=(X_1,\ldots,X_n)\), equipped with its projectively flat connection and Hermitian structure.  Let \(D_{I|I^c}\subset \overline M_{0,n}\) be the boundary divisor corresponding to a stable partition \(I|I^c\), with \(|I|\ge 2\) and \(|I^c|\ge 2\).  In a formal sewing coordinate \(q\) transverse to this divisor, the modular-functor sewing axiom gives a completed factorization
\begin{equation}\label{eq:modular-factorization-footprint}
  \widehat{\cV}_{\vec X}
  \cong
  \bigoplus_{a\in\Irr(\cC)}
  \cV_{I,a}\,\widehat\otimes\,\cV_{I^c,a^*},
\end{equation}
up to the usual projective factors and conformal-weight powers of the sewing parameter in analytic CFT normalizations.  The summand indexed by \(a\) is the conformal-block footprint sector in which the cluster \(I\) carries total charge \(a\) across the node.
\end{proposition}

\begin{proof}
The statement is the sewing axiom of a modular functor near the boundary stratum determined by the stable partition \(I|I^c\).  Degeneration along this partition produces two stable components meeting at a node.  Gluing them requires a sum over a simple label \(a\) on one branch and the dual label \(a^*\) on the other, yielding the direct sum in \eqref{eq:modular-factorization-footprint}.

The completion in the sewing coordinate \(q\) indicates that the statement is local near the boundary divisor in \(\overline M_{0,n}\).  In analytic conformal-field-theoretic normalizations, one also sees powers of \(q\) determined by conformal weights, together with projective factors coming from the chosen normalization of the connection and Hermitian structure.  These factors do not change the sector decomposition.  The label \(a\) transmitted through the neck is exactly the total charge of the cluster \(I\), so the summand indexed by \(a\) is the conformal-block version of the footprint sector.  Thus the sewing axiom is the analytic counterpart of the categorical total-charge decomposition in Theorem \ref{thm:categorical-footprint-decomp}.
\end{proof}

\begin{proposition}[Sewing asymptotics as likelihood separation]\label{prop:sewing-likelihood-separation}
Assume, in the setting of Proposition~\ref{prop:footprint-factorization}, that a block-norm likelihood model is expressed in a sewing coordinate \(q\) and that the contribution of the channel \(a\) has leading form
\[
  w_a(q)=C_a |q|^{2\Delta_a}(1+O(|q|))
  \qquad(C_a>0),
\]
where \(\Delta_a\) is the exponent determined by the conformal weight of the intermediate label and by the normalization of the two components.  In the standard sewing/OPE normalization used later, this exponent is the corresponding weight difference \(h_a-h_I-h_{I^c}\), up to any common convention-dependent shift already absorbed into the prefactor.  Then for two compatible channels \(a\) and \(b\), the log-likelihood ratio satisfies
\[
  \log\frac{w_a(q)}{w_b(q)}
  =\log\frac{C_a}{C_b}+2(\Delta_a-\Delta_b)\log |q|+O(|q|)
  \qquad(q\to 0).
\]
Consequently, near a degeneration divisor, the leading conformal exponents determine which footprint sector is favoured unless the leading exponents coincide.
\end{proposition}

\begin{proof}
Taking logarithms of the assumed asymptotic expression gives
\[
  \log w_a(q)=\log C_a+2\Delta_a\log |q|+\log(1+O(|q|)).
\]
Since \(\log(1+O(|q|))=O(|q|)\) as \(q\to 0\), subtracting the corresponding formula for \(b\) gives the displayed expression.  The final assertion follows because \(\log |q|\to -\infty\) near the boundary.  If \(\Delta_a\ne\Delta_b\), the logarithmic term dominates the bounded constant term.  If the exponents agree, then the leading powers do not separate the channels and the constants or higher-order terms in the conformal blocks become relevant.
\end{proof}

The calculation shows exactly what the conformal enhancement adds.  A topological footprint records the channel label \(a\), but does not rank compatible histories sharing the same measured local data.  Conformal blocks attach analytic weights to those alternatives.  Near a boundary divisor, the sewing parameter \(q\) supplies a geometric scale and channels with different exponents separate at order \(|\log |q||\).  Away from the boundary, one must use the full block functions.

The normalization in \eqref{eq:CBposterior} is a finite-dimensional Bayesian update.  Let \(A\) be the set of channels compatible with the external insertions and let \(A_s\subset A\) be the subset compatible with the measured footprint.  In the sharp block-norm model,
\[
  \Prob(\alpha\mid z,s)=0\quad (\alpha\notin A_s),
  \qquad
  \Prob(\alpha\mid z,s)=\frac{\|\mathcal F_\alpha(z)\|^2}{\sum_{\beta\in A_s}\|\mathcal F_\beta(z)\|^2}
  \quad (\alpha\in A_s).
\]
If the detector is imperfect, the sharp subset \(A_s\) is replaced by likelihoods \(L(s\mid\alpha)\), as in ordinary statistical decoding.  Conformal blocks do not remove the need for a noise model.  They supply a geometry-dependent family of channel weights that can be combined with one.

A second elementary consistency check concerns changes of fusion tree.  If \(F\) is the recoupling matrix from one channel basis to another, then a vector of block amplitudes transforms by
\[
  \mathcal F^{(t)}_\beta(z)=\sum_\alpha F_{\beta\alpha}\,\mathcal F^{(s)}_\alpha(z).
\]
A footprint associated to the \(t\)-channel should therefore be compared with amplitudes in the \(t\)-basis, not with un-transformed \(s\)-channel amplitudes.  This is the conformal-block analogue of changing measurement basis before interpreting a stabilizer outcome.  It also prevents a common conceptual mistake: the footprint is local, but its coordinate description depends on a global choice of fusion tree.

The notation in \eqref{eq:CBposterior} suppresses two choices that are important in applications.  First, a conformal-block space is usually a vector bundle over a configuration or moduli space of marked curves rather than a single fixed vector space.  To compare weights at different points \((z_i)\), one should specify a trivialization, a Hermitian structure, or a projectively flat connection.  In unitary rational CFTs there are natural inner products and modular-functor structures, but the exact normalization conventions matter.  For the present purpose, only ratios of compatible block weights are used, and common prefactors cancel in the elementary Ising example below.

Second, the indicator \(\chi_s(\alpha)\) represents the compatibility of a hidden channel \(\alpha\) with the observed footprint \(s\).  In a noiseless idealization this is a sharp constraint.  In a realistic setting it should be replaced by a likelihood \(L(s\mid\alpha)\), accounting for measurement error, imperfect localization, and detector noise.  Thus a more general decoder would have
\[
  w_\alpha(z;s)=\|\mathcal F_\alpha(z)\|^2 L(s\mid\alpha)\pi(\alpha),
\]
where \(\pi(\alpha)\) is a prior coming from the physical noise model.  Definition \ref{def:CBdecoder} is the essential field-theoretic core obtained by setting \(L\) to a sharp compatibility function and absorbing simple priors into the block weights.

The word ``posterior'' in \eqref{eq:CBposterior} is conditional on the operational model.  A conformal block is not, by itself, an observed probability distribution.  It contributes to a decoding probability only after amplitudes, measurement events, and physical noise have been related.  Rational CFT nevertheless supplies two pieces of data unavailable in a purely topological decoder: a finite set of channel variables and analytic dependence on geometry.

The definition keeps only the data needed here.  A realistic decoder would combine conformal-block weights with a physical noise model and account for measurement imperfections, leakage, and finite-temperature processes.  The CFT contribution is a structured family of priors or likelihood factors on fusion histories.  In the topological limit, the dependence on \((z_i)\) disappears or becomes locally constant.  In the conformal regime, histories with the same topological footprint may still receive different geometric weights.

\begin{remark}[Conformal geometry as side information]
A conformal-block likelihood model should not be interpreted as replacing the standard statistical mechanics of decoding.  It supplies additional field-theoretic side information.  In a physical anyonic medium, the positions of quasiparticles, boundary defects, or measurement events are not irrelevant to the likelihood of a fusion history.  The conformal-block formalism is a mathematically natural way to package this geometry dependence.
\end{remark}

One may also describe the decoder in spacetime.  Let \(M\) be a spacetime cobordism from an initial code surface to a final code surface, decorated by an error network \(N\).  The field theory assigns an amplitude
\begin{equation}\label{eq:path-integral-decoder}
Z_\cF(M,N):Z_\cF(\Sigma_{\mathrm{in}},\cD_{\mathrm{in}})
\longrightarrow
Z_\cF(\Sigma_{\mathrm{out}},\cD_{\mathrm{out}}).
\end{equation}
Given an observed footprint \(s\), one sums or optimizes over networks \(N\) with \(\Foot(N)=s\).  In a semisimple TQFT this becomes a finite state-sum or ribbon-graph evaluation problem.  In a rational CFT it becomes an inference problem over conformal blocks.  In either case, the decoder weight is computed from a path-integral or state-sum object conditioned on observations; a complete decoder also includes the measurement model and the recovery decision rule.

The conformal enhancement can be stated directly.  Topological data determine which fusion histories are allowed; conformal blocks vary those histories analytically with the marked points; and the decoder combines the resulting weights with the observed footprint and a physical noise model.  Once a recovery decision rule is supplied, the same classical measurement record is interpreted in a hypothesis space that now varies over conformal geometry.

The topological limit gives a consistency check.  If geometric dependence is ignored, or only locally constant modular-functor transport is retained, histories in the same topological class receive the same field-theoretic weight.  The likelihood model then reduces to inference over fusion trees and topological charges.  When puncture or defect positions are known and physically relevant, conformal geometry can further distinguish histories with the same topological footprint.  Conformal decoding is therefore an enhancement of topological decoding, not a replacement for it.

\section{Worked example: the Ising conformal-block qubit}\label{sec:ising-example}

We turn to explicit Ising calculations.  The first is the minimal four-puncture conformal-block qubit.  It should be read as a local model of fusion-space encoding, recoupling, diagnostic footprint measurement, and geometry-sensitive likelihoods rather than as a high-distance code.  We then enlarge to six \(\sigma\)-punctures and choose a proper two-dimensional code subspace inside the four-dimensional vacuum fusion space.  That additional redundancy permits an actual syndrome-admissible pair-charge measurement and an exact conditional recovery for a specified local bilinear error.  Neither finite example makes a scalable threshold claim.  Scalability is addressed separately in Section~\ref{sec:scalable}, where the Peierls-type criterion is formulated for growing families.  We first analyze the four-puncture logical space, recoupling and braid operations, and complementary diagnostics; then we give the six-puncture correction example; finally we return to the four-point conformal blocks and their geometry-sensitive weights.  The Ising modular category has simple objects
\begin{equation}\label{eq:ising-objects}
\one,\qquad \sigma,\qquad \psi,
\end{equation}
with fusion rules
\begin{equation}\label{eq:ising-fusion}
\sigma\otimes\sigma=\one\oplus\psi,
\qquad
\sigma\otimes\psi=\sigma,
\qquad
\psi\otimes\psi=\one.
\end{equation}
The object \(\sigma\) is the nonabelian anyon.  The object \(\psi\) is a fermion.  The quantum dimensions are
\begin{equation}\label{eq:ising-dims}
d_\one=1,
\qquad
 d_\psi=1,
\qquad
 d_\sigma=\sqrt 2.
\end{equation}
This category appears in the Moore--Read/Ising anyon context and in Kitaev's honeycomb model in the appropriate nonabelian phase \cite{MooreRead1991,ReadMoore1992,Kitaev2006}.

The fusion rules already show the essential nonclassical feature.  Two \(\sigma\)-anyons do not have a single deterministic product: their pair can carry either vacuum charge or fermion charge.  A footprint measurement of a pair is therefore a measurement of a fusion channel.  The code space below is built precisely from this ambiguity.  It is small enough to compute explicitly, but it is already nonabelian in the sense relevant to the paper: the hidden variable is not a bit flip at a site but an internal fusion channel of a field-theoretic state space.

For orientation, one can count the dimension directly.  Since
\[
  \sigma\otimes\sigma=\one\oplus\psi,
\]
the tensor product of four \(\sigma\)'s contains a vacuum component in two ways:
\[
  ((\sigma\sigma)_\one(\sigma\sigma)_\one)_\one,
  \qquad
  ((\sigma\sigma)_\psi(\sigma\sigma)_\psi)_\one.
\]
The mixed possibilities \((\one,\psi)\) and \((\psi,\one)\) do not have total vacuum charge.  Thus the vacuum fusion space is two-dimensional.  This elementary count is the categorical analogue of saying that a code has one logical qubit.

Consider four \(\sigma\)-punctures on the sphere or disk with total vacuum charge.  The associated logical space is
\begin{equation}\label{eq:ising-logical}
\cH_L=\Hom(\one,\sigma^{\otimes 4}).
\end{equation}
Here the protected space is the entire four-anyon vacuum fusion space, so the code projection is \(P_4=I_{\cH_L}\).

\begin{corollary}[No nontrivial sharp projective syndrome on the full four-\(\sigma\) space]\label{cor:four-sigma-full-space-trivial-syndrome}
Let \(\cA_{\mathrm{fp}}=\operatorname{span}\{\Pi_s\}\) be a sharp projective footprint measurement algebra on \(\cH_L=\Hom(\one,\sigma^{\otimes4})\).  If it is syndrome-admissible in the sharp projective sense of Definition~\ref{def:syndrome-admissible}(i) with respect to the full-space code projection \(P_4=I_{\cH_L}\), then every \(\Pi_s\) is either \(0\) or \(I_{\cH_L}\).  Hence any nontrivial sharp projective footprint measurement on this full protected space is diagnostic rather than syndrome-admissible.
\end{corollary}

\begin{proof}
Condition~(i) of Definition~\ref{def:syndrome-admissible}, equivalently Lemma~\ref{lem:no-error-safe}, gives \(\Pi_sP_4=\epsilon_sP_4\) with \(\epsilon_s\in\{0,1\}\).  Since \(P_4=I_{\cH_L}\), one has \(\Pi_s=\epsilon_sI_{\cH_L}\).  Orthogonality and \(\sum_s\Pi_s=I_{\cH_L}\) leave exactly one nonzero outcome projector.
\end{proof}

Thus redundancy is necessary, not merely helpful, for obtaining a nontrivial sharp projective syndrome from this type of footprint measurement.  This is why the six-\(\sigma\) example below passes to a proper subspace of the ambient vacuum fusion space.

The notation in the next display is standard in the anyon literature, but it is worth spelling out why the displayed fusion trees are literally vectors in \(\cH_L\).  Since the Ising category is semisimple and all relevant fusion multiplicities are either \(0\) or \(1\), we may choose normalized nonzero splitting morphisms
\begin{equation}\label{eq:ising-splittings}
  i_a\in \Hom(a,\sigma\otimes\sigma),
  \qquad a\in\{\one,\psi\},
\end{equation}
identifying the two simple summands \(a\subset \sigma\otimes\sigma\).  Since \(\one^*\cong \one\) and \(\psi^*\cong \psi\), choose also the coevaluation morphism
\begin{equation}\label{eq:ising-coev}
  \mathrm{coev}_a\in \Hom(\one,a\otimes a).
\end{equation}
Up to the associator, which we suppress in the usual graphical convention, the composite
\begin{equation}\label{eq:ising-basis-morphism}
  e_a
  :=(i_a\otimes i_a)\circ \mathrm{coev}_a
  \in
  \Hom\bigl(\one,(\sigma\otimes\sigma)\otimes(\sigma\otimes\sigma)\bigr)
  \cong \Hom(\one,\sigma^{\otimes4})
\end{equation}
is therefore an honest morphism from the tensor unit into four \(\sigma\)-labels.  The two cases \(a=\one\) and \(a=\psi\) are precisely the two fusion trees drawn below.  Different normalized choices of the one-dimensional splitting and coevaluation morphisms change these vectors only by phases.  The associated channel projectors and footprint measurements are therefore independent of these harmless basis conventions.

This also gives a direct dimension count.  Expanding the two pairs first, one obtains
\begin{align}\label{eq:ising-hom-decomp}
\Hom\bigl(\one,(\sigma\otimes\sigma)\otimes(\sigma\otimes\sigma)\bigr)
&\cong
\bigoplus_{a,b\in\{\one,\psi\}}
\Hom(\one,a\otimes b)
\otimes \Hom(a,\sigma\otimes\sigma)
\otimes \Hom(b,\sigma\otimes\sigma).
\end{align}
The last two factors are one-dimensional for \(a,b\in\{\one,\psi\}\), and \(\Hom(\one,a\otimes b)\) is one-dimensional exactly when \(b\cong a^*\).  Since both \(\one\) and \(\psi\) are self-dual, only \((a,b)=(\one,\one)\) and \((a,b)=(\psi,\psi)\) contribute.  Hence \(\cH_L\) is two-dimensional, with basis \(e_\one,e_\psi\).

With this morphism-level interpretation understood, we write
\begin{align}
|0_L\rangle
&:=e_\one
=\bigl|((\sigma\sigma)_\one(\sigma\sigma)_\one)_\one\bigr\rangle,
\label{eq:ising-zero}\\
|1_L\rangle
&:=e_\psi
=\bigl|((\sigma\sigma)_\psi(\sigma\sigma)_\psi)_\one\bigr\rangle.
\label{eq:ising-one}
\end{align}
Thus the logical qubit is a fusion-channel qubit.  The two logical basis states differ by whether the first pair splits through, or equivalently is measured to fuse to, the vacuum channel \(\one\) or to the fermion channel \(\psi\), with the second pair constrained to match so that the total charge is vacuum.  This basis is illustrated in Figure \ref{fig:isingqubit}.

\begin{figure}[t]
\centering
\begin{minipage}{0.47\linewidth}
\centering
\begin{tikzpicture}[scale=0.92, every node/.style={font=\small}]
\foreach \x/\lab in {0/1,1.6/2,3.2/3,4.8/4}{
  \fill (\x,0) circle (2.2pt);
  \node[below] at (\x,-0.05) {$z_{\lab}$};
  \node[above] at (\x,0.20) {$\sigma$};
}
\draw[line width=0.8pt] (0,0.15) .. controls (0.8,1.0) and (0.8,1.0) .. (1.6,0.15);
\draw[line width=0.8pt] (3.2,0.15) .. controls (4.0,1.0) and (4.0,1.0) .. (4.8,0.15);
\draw[line width=0.8pt] (0.8,1.35) -- (2.4,2.15);
\draw[line width=0.8pt] (4.0,1.35) -- (2.4,2.15);
\node at (0.8,1.22) {$\one$};
\node at (4.0,1.22) {$\one$};
\node at (2.4,2.40) {$\one$};
\node at (2.4,2.82) {$|0_L\rangle$};
\end{tikzpicture}
\end{minipage}
\hfill
\begin{minipage}{0.47\linewidth}
\centering
\begin{tikzpicture}[scale=0.92, every node/.style={font=\small}]
\foreach \x/\lab in {0/1,1.6/2,3.2/3,4.8/4}{
  \fill (\x,0) circle (2.2pt);
  \node[below] at (\x,-0.05) {$z_{\lab}$};
  \node[above] at (\x,0.20) {$\sigma$};
}
\draw[line width=0.8pt] (0,0.15) .. controls (0.8,1.0) and (0.8,1.0) .. (1.6,0.15);
\draw[line width=0.8pt] (3.2,0.15) .. controls (4.0,1.0) and (4.0,1.0) .. (4.8,0.15);
\draw[line width=0.8pt] (0.8,1.35) -- (2.4,2.15);
\draw[line width=0.8pt] (4.0,1.35) -- (2.4,2.15);
\node at (0.8,1.22) {$\psi$};
\node at (4.0,1.22) {$\psi$};
\node at (2.4,2.40) {$\one$};
\node at (2.4,2.82) {$|1_L\rangle$};
\end{tikzpicture}
\end{minipage}
\caption{The four-\(\sigma\) Ising conformal-block qubit.  Four \(\sigma\)-punctures with total vacuum charge support a two-dimensional fusion space.  In the displayed fusion basis, \(|0_L\rangle\) and \(|1_L\rangle\) are distinguished by whether the pairs \((1,2)\) and \((3,4)\) fuse through the vacuum channel \(\one\) or the fermion channel \(\psi\), while the total charge remains \(\one\).}
\label{fig:isingqubit}
\end{figure}
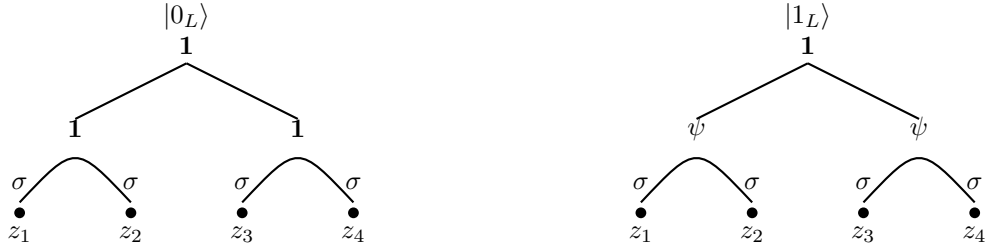

The nontrivial associativity move on three \(\sigma\)'s is
\begin{equation}\label{eq:Fising}
F_{\sigma\sigma\sigma}^{\sigma}
=\frac{1}{\sqrt 2}
\begin{pmatrix}
1&1\\
1&-1
\end{pmatrix},
\end{equation}
where the basis is indexed by the intermediate channels \(\one\) and \(\psi\).  Thus a change of fusion tree acts as a Hadamard gate on the fusion-channel qubit.  This is already a small but striking manifestation of the ZX shadow principle: in the Ising theory, a purely topological recoupling move acts like a familiar Clifford operation.

The relevant braiding eigenvalues for two \(\sigma\)'s are, up to conventional phase choices,
\begin{equation}\label{eq:Rising}
R_{\sigma\sigma}^{\one}=e^{-\pi i/8},
\qquad
R_{\sigma\sigma}^{\psi}=e^{3\pi i/8}.
\end{equation}
Braiding and fusion therefore generate a Clifford-level representation, but not a universal gate set by braiding alone.  Universality is not needed for the present example.  What matters here is that encoding, footprint measurement, and correction can all be expressed in the same field-theoretic language.

Let us make the braid action explicit.  In the basis \((|0_L\rangle,|1_L\rangle)\), the braid exchanging the first two \(\sigma\)'s is diagonal:
\begin{equation}\label{eq:ising-braid1}
  B_1=\begin{pmatrix}
  R_{\sigma\sigma}^{\one}&0\\
  0&R_{\sigma\sigma}^{\psi}
  \end{pmatrix}
  = e^{-\pi i/8}
  \begin{pmatrix}
  1&0\\0&i
  \end{pmatrix}.
\end{equation}
Up to the global phase \(e^{-\pi i/8}\), this is the phase gate.  The braid exchanging the middle two \(\sigma\)'s is obtained by recoupling, applying the same diagonal braid, and recoupling back:
\begin{equation}\label{eq:ising-braid2}
  B_2 = F_{\sigma\sigma\sigma}^{\sigma} B_1
  F_{\sigma\sigma\sigma}^{\sigma}
  = e^{-\pi i/8}\frac{1}{2}
  \begin{pmatrix}
  1+i&1-i\\
  1-i&1+i
  \end{pmatrix}.
\end{equation}
Equations \eqref{eq:ising-braid1} and \eqref{eq:ising-braid2} show, in matrix form, how topological recoupling and braiding become logical Clifford operations on the fusion qubit.  The same matrices also govern how a local footprint measurement in one pairing is represented in another pairing.

The matrices \(F\) and \(R\) also illustrate the distinction between global and relative phase.  The common phase in a braid representation has no effect on a projective logical state.  The relative phase between the \(\one\)- and \(\psi\)-channels, however, is physically meaningful on the encoded qubit.  This is the same distinction that appears in elementary circuit language between global phase and a logical phase gate, but here it arises from topological spin and braiding.  In the present paper the emphasis is not on using braids to perform a universal computation.  It is on the fact that the same field-theoretic data governing braiding and recoupling also governs footprint measurement and decoding.

We may therefore summarize the four-puncture code by the dictionary
\[
\begin{array}{rcl}
\text{logical basis} &:& \text{intermediate fusion channel } \one \text{ or } \psi\\
\text{basis change} &:& F\text{-move}\\
\text{phase information} &:& R\text{-symbols and topological spins}\\
\text{syndrome/footprint} &:& \text{measured local fusion channel}\\
\text{decoder weight} &:& \text{conformal-block norm or related likelihood}
\end{array}
\]
This dictionary has a limited purpose.  It gives a fully explicit local model of the framework; it does not endow that model with macroscopic distance.  The next task is therefore to pass from logical structure to observable structure.  We do this in two steps: first by identifying the local fusion footprints themselves, and then by exhibiting a stabilizer-shadow circuit that extracts them in ordinary ancilla language.

\subsection{Diagnostic pair-charge footprints in the four-\texorpdfstring{\(\sigma\)}{sigma} qubit}

In the four-\(\sigma\) code, a local measurement of the total charge of a pair of nearby \(\sigma\)'s is a fusion-channel measurement.  If the measured pair is \((1,2)\), the possible outcomes are \(\one\) and \(\psi\), corresponding to the logical basis \eqref{eq:ising-zero} and \eqref{eq:ising-one}.  If the measured pair is \((2,3)\), then the measurement is diagonal in a different fusion basis related to the first by the \(F\)-matrix \eqref{eq:Fising}.  Thus different pair-charge footprint measurements are related by topological recoupling.  In the logical basis \((|0_L\rangle,|1_L\rangle)\), it is convenient to write
\begin{equation}\label{eq:Zf12}
Z_f^{(12)} := \Pi_{\one}^{(12)}-\Pi_{\psi}^{(12)}
= \begin{pmatrix}1&0\\0&-1\end{pmatrix},
\end{equation}
where \(\Pi_{\one}^{(12)}\) and \(\Pi_{\psi}^{(12)}\) are the projectors onto the two fusion channels of the pair \((1,2)\).  Measuring the footprint of the pair \((2,3)\) gives the recoupled observable
\begin{equation}\label{eq:Zf23}
Z_f^{(23)} := F\,Z_f^{(12)}\,F^{-1}
= \begin{pmatrix}0&1\\1&0\end{pmatrix}.
\end{equation}
Thus the two simplest local footprints in the Ising qubit are represented by the Pauli-type pair \((Z,X)\) on the logical fusion qubit.  This is a precise sense in which the stabilizer shadow of the field-theoretic syndrome calculus becomes visible in ordinary circuit language.

The projectors themselves make the noncommuting nature of the two pair measurements transparent.  In the \((12)\)-fusion basis one has
\[
  \Pi_{\one}^{(12)}=\frac{1}{2}(I+Z),\qquad
  \Pi_{\psi}^{(12)}=\frac{1}{2}(I-Z).
\]
After recoupling, the \((23)\)-projectors are
\[
  \Pi_{\one}^{(23)}=F\Pi_{\one}^{(12)}F^{-1}=\frac{1}{2}(I+X),\qquad
  \Pi_{\psi}^{(23)}=F\Pi_{\psi}^{(12)}F^{-1}=\frac{1}{2}(I-X).
\]
Thus \(\Pi_{\one}^{(12)}\) and \(\Pi_{\one}^{(23)}\) do not commute.  Their commutator is
\[
  [\Pi_{\one}^{(12)},\Pi_{\one}^{(23)}]
  =\frac{1}{4}[Z,X]
  =\frac{i}{2}Y,
\]
up to the usual Pauli convention for \(Y\).  This is a two-dimensional matrix calculation, but it carries the conceptual content of the example: overlapping local fusion footprints are not simultaneously classical data.  They become Pauli observables only after choosing a fusion basis and taking the stabilizer shadow.

In the Ising example, this shadow can be drawn explicitly.  Figure~\ref{fig:ising-circuit} gives an ancilla-based circuit that extracts two complementary local footprints and displays the role of the \(F\)-move as the bridge between their measurement bases.

\begin{figure}[t]
\centering
\begin{tikzpicture}[x=0.95cm,y=0.92cm, line cap=round, line join=round, every node/.style={font=\small}]
  \path[use as bounding box] (-0.4,-1.45) rectangle (13.1,3.25);
  % wire labels
  \node[left, fill=white, inner sep=1.2pt] at (0,2) {$|0\rangle_{a_{12}}$};
  \node[left, fill=white, inner sep=1.2pt] at (0,1) {$|\psi_L\rangle$};
  \node[left, fill=white, inner sep=1.2pt] at (0,0) {$|0\rangle_{a_{23}}$};

  % wires
  \draw[thick] (0.2,2) -- (6.2,2);
  \draw[thick] (0.2,1) -- (12.8,1);
  \draw[thick] (0.2,0) -- (12.8,0);

  % Titles above blocks
  \node[font=\normalsize, align=center, fill=white, inner sep=1.5pt] at (3.2,2.92) {measure footprint\\of pair $(1,2)$};
  \node[font=\normalsize, align=center, fill=white, inner sep=1.5pt] at (9.8,2.92) {measure footprint\\of pair $(2,3)$};

  % first ancilla H gate
  \draw[fill=white] (1.1,1.72) rectangle (1.7,2.28);
  \node at (1.4,2) {$H$};
  % first controlled gate
  \fill (2.7,2) circle (2.2pt);
  \draw[thick] (2.7,2) -- (2.7,1.32);
  \draw[fill=white, rounded corners=1.5pt] (2.15,0.72) rectangle (3.25,1.28);
  \node at (2.70,1.00) {$Z_f^{(12)}$};
  % second H and measurement
  \draw[fill=white] (3.9,1.72) rectangle (4.5,2.28);
  \node at (4.2,2) {$H$};
  \draw[fill=white, rounded corners=1.5pt] (5.3,1.70) rectangle (6.1,2.30);
  \node at (5.7,2.0) {$M$};
  \node[above, fill=white, inner sep=1pt] at (5.7,2.32) {$m_{12}$};

  % bridge on data wire: braid then F and F^-1 around second measurement
  \draw[fill=white, rounded corners=1.5pt] (6.55,0.72) rectangle (7.35,1.28);
  \node at (6.95,1.0) {$B_2$};
  \draw[fill=white] (7.8,0.72) rectangle (8.4,1.28);
  \node at (8.1,1.0) {$F$};

  % second ancilla H gate
  \draw[fill=white] (7.9,-0.28) rectangle (8.5,0.28);
  \node at (8.2,0) {$H$};
  % second controlled gate
  \fill (9.5,0) circle (2.2pt);
  \draw[thick] (9.5,0) -- (9.5,0.72);
  \draw[fill=white, rounded corners=1.5pt] (8.95,0.72) rectangle (10.05,1.28);
  \node at (9.5,1.0) {$Z_f^{(12)}$};
  % second H and measurement
  \draw[fill=white] (10.7,-0.28) rectangle (11.3,0.28);
  \node at (11.0,0) {$H$};
  \draw[fill=white, rounded corners=1.5pt] (11.75,-0.30) rectangle (12.55,0.30);
  \node at (12.15,0.0) {$M$};
  \node[above, fill=white, inner sep=1pt] at (12.15,0.32) {$m_{23}$};
  % F^-1
  \draw[fill=white] (10.7,0.72) rectangle (11.5,1.28);
  \node at (11.1,1.0) {$F^{-1}$};

  % explanatory labels beneath circuit
  \node[font=\small, align=center, fill=white, inner sep=2pt, text width=3.7cm] at (2.7,-0.80)
    {$Z_f^{(12)}=\Pi_{\one}^{(12)}-\Pi_{\psi}^{(12)}$};
  \node[font=\small, align=center, fill=white, inner sep=2pt, text width=3.9cm] at (9.5,-0.80)
    {$Z_f^{(23)}=F\,Z_f^{(12)}\,F^{-1}$};
\end{tikzpicture}
\caption{A stabilizer-shadow circuit for complementary footprint measurements in the four-\(\sigma\) Ising code.  The upper ancilla extracts the local fusion footprint of the pair \((1,2)\) by phase-kickback from the logical observable \(Z_f^{(12)}\).  As drawn, an optional braid gate \(B_2\) is inserted before the second readout.  The logical line is then recoupled by the \(F\)-move, the same controlled observable is measured again, and the basis is restored.  With \(B_2\) omitted, the lower ancilla measures \(Z_f^{(23)}\) on the original input state.  With \(B_2\) included, it measures the same recoupled footprint on the transported state.  The circuit makes explicit that distinct local footprints are related by field-theoretic recoupling rather than by an arbitrary choice of coordinates.}
\label{fig:ising-circuit}
\end{figure}
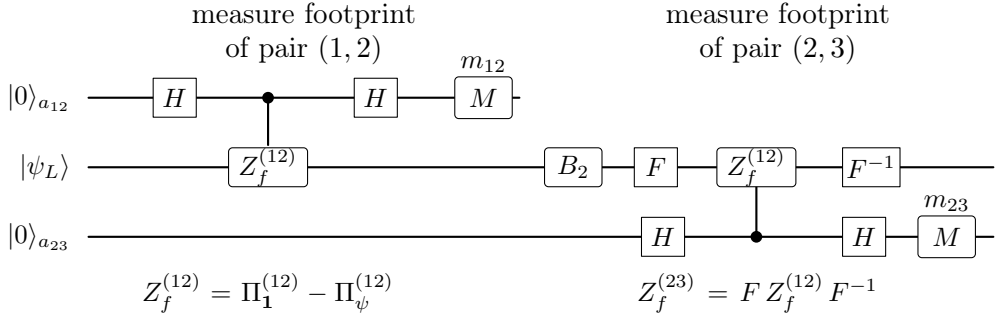

For an input state \(|\psi_L\rangle=a|0_L\rangle+b|1_L\rangle\), the first ancilla block performs the familiar phase-kickback measurement of \(Z_f^{(12)}\):
\[
(a|0_L\rangle+b|1_L\rangle)|0\rangle
\mapsto
(a|0_L\rangle+b|1_L\rangle)\frac{|0\rangle+|1\rangle}{\sqrt2}
\mapsto
\frac{a|0_L\rangle|0\rangle+a|0_L\rangle|1\rangle+b|1_L\rangle|0\rangle-b|1_L\rangle|1\rangle}{\sqrt2}
\]
followed by a final Hadamard on the ancilla, which yields
\begin{equation}\label{eq:ising-kickback}
(a|0_L\rangle+b|1_L\rangle)|0\rangle
\longmapsto
a|0_L\rangle|0\rangle+b|1_L\rangle|1\rangle.
\end{equation}
Thus a computational-basis measurement of the ancilla records precisely the \((1,2)\)-fusion footprint.  The second block is formally the same calculation, but conjugated by \(F\), so it measures \(Z_f^{(23)}\).  Since \(F\) is the Hadamard matrix in the Ising case, this second footprint measurement is an \(X\)-type measurement on the logical fusion qubit.  The additional braid gate \(B_2\) is included to emphasize that field-theoretic transport and field-theoretic measurement naturally coexist in a single circuit shadow.  If the sole goal is to display the measurement of \(Z_f^{(23)}\) on the original state, one simply deletes the \(B_2\) box or sets it equal to the identity.

The two observables are complementary in the ordinary qubit sense:
\begin{equation}\label{eq:ising-anticommutation}
  Z_f^{(12)}Z_f^{(23)}=-Z_f^{(23)}Z_f^{(12)}.
\end{equation}
This anticommutation is not an arbitrary Pauli convention.  It is the circuit shadow of the fact that the two footprint measurements correspond to two different fusion trees.  Consequently, a projective measurement of the \((1,2)\)-footprint generally disturbs the statistics of the \((2,3)\)-footprint, just as a \(Z\)-measurement disturbs a later \(X\)-measurement on an ordinary qubit.  For the state \(|\psi_L\rangle=a|0_L\rangle+b|1_L\rangle\), the \((1,2)\)-measurement has probabilities
\[
  p_{12}(\one)=|a|^2,
  \qquad
  p_{12}(\psi)=|b|^2.
\]
The \((2,3)\)-measurement, by contrast, is diagonal in the recoupled basis.  Since \(F=F^{-1}\), one obtains
\[
  p_{23}(\one)=\left|\frac{a+b}{\sqrt2}\right|^2,
  \qquad
  p_{23}(\psi)=\left|\frac{a-b}{\sqrt2}\right|^2.
\]
This small calculation is useful operationally.  It shows that footprint measurements are not merely labels attached to the same hidden error.  They are genuine measurements of different local field-theoretic boundary decompositions, with the usual disturbance expected of noncommuting observables.  Thus \(Z_f^{(12)}\) and \(Z_f^{(23)}\) should not be interpreted as simultaneous stabilizer checks.  They generate the full noncommutative diagnostic algebra \(M_2(\bC)\) on the fusion qubit.  In the four-puncture local model they are complementary footprint diagnostics or measurement primitives.  This is consistent with measurement-only topological quantum computation, where topological-charge measurements are deliberately used to enact computational transformations rather than to report a harmless syndrome, and with interferometric realizations of topological-charge measurement \cite{BondersonFreedmanNayak2008,BondersonShtengelSlingerland2008}.  A genuine syndrome-extraction protocol for an unknown encoded state requires a syndrome-admissible commuting footprint algebra whose no-error outcomes do not distinguish logical states.  The next subsection gives such an example.

Local error insertions may be described schematically as follows.
\begin{enumerate}[label=(\roman*)]
\item A local cluster of two \(\sigma\)-defects admits fusion channels \(\one\) and \(\psi\).  A strictly neutral pair-creation event from the vacuum selects the total vacuum channel unless accompanied by compensating charge or embedded in a larger defect history.
\item A \(\psi\)-line encircling or passing between punctures can change relative phases between fusion channels.
\item A defect history that changes the total charge of a local cluster is detectable by a charge measurement.
\item A defect history that is globally neutral but topologically nontrivial may act as a logical operator.
\end{enumerate}
A measured pair charge is therefore not merely a binary stabilizer sign; it is a categorical footprint in the fusion theory.  In the four-puncture qubit the observables just displayed are diagnostics, not safe syndromes for an unknown encoded state.  The six-puncture construction below shows how the same type of pair-charge datum becomes a syndrome once redundancy makes the no-error outcome independent of the logical state.

\subsection{A syndrome-admissible six-\texorpdfstring{\(\sigma\)}{sigma} code with exact recovery}\label{subsec:six-sigma-recovery}

The four-puncture calculation deliberately exhibits diagnostics that are not safe syndromes.  A slightly larger Ising fusion space already gives the complementary phenomenon: a nontrivial pair-charge measurement can be constant on the code in the no-error sector, detect a specified error, and support exact conditional recovery.

Consider six \(\sigma\)-anyons with total vacuum charge,
\begin{equation}\label{eq:ising-six-space}
  \cH_6:=\Hom(\one,\sigma^{\otimes 6}).
\end{equation}
Pair the anyons as \((1,2)\), \((3,4)\), and \((5,6)\).  Write
\[
  |z_1,z_2,z_3\rangle,
  \qquad z_j\in\{+1,-1\},
\]
for the fusion-path basis in which \(z_j=+1\) means that the \(j\)-th pair fuses to \(\one\), while \(z_j=-1\) means that it fuses to \(\psi\).  Since \(\psi\otimes\psi\cong\one\), the total-vacuum constraint is
\begin{equation}\label{eq:ising-six-parity-constraint}
  z_1z_2z_3=+1.
\end{equation}
Hence
\begin{equation}\label{eq:ising-six-basis}
  \cH_6
  =\operatorname{span}\bigl\{
  |+,+,+\rangle,
  |+,-,-\rangle,
  |-,+,-\rangle,
  |-,-,+\rangle
  \bigr\},
\end{equation}
so \(\dim\cH_6=4\).

Choose the two-dimensional code
\begin{equation}\label{eq:ising-six-code}
  \cC_6
  :=\operatorname{span}\{|0_L\rangle,|1_L\rangle\}
  =\operatorname{span}\{|+,+,+\rangle,|-,-,+\rangle\}.
\end{equation}
Thus the last pair \((5,6)\) fuses to vacuum on every code state, while the logical qubit is carried by the remaining correlated pair labels.  Let
\begin{equation}\label{eq:ising-six-syndrome-observable}
  Z_3:=\Pi_{\one}^{(56)}-\Pi_{\psi}^{(56)},
  \qquad
  \Pi_\pm:=\frac{1}{2}(I\pm Z_3).
\end{equation}
With the convention above, \(\Pi_+\) is the vacuum-charge projector for the pair \((5,6)\) and \(\Pi_-\) is the \(\psi\)-charge projector.  If \(P_6\) denotes the projection onto \(\cC_6\), then
\begin{equation}\label{eq:ising-six-noerror-safe}
  \Pi_+P_6=P_6,
  \qquad
  \Pi_-P_6=0.
\end{equation}
The pair-charge measurement therefore reveals no logical information in the no-error case.  It is syndrome-admissible for any error family whose representatives are resolved by these two sectors.

Figure~\ref{fig:six-sigma-syndrome} summarizes the geometry used below.  The measured pair is \((5,6)\); the bilinear \(T=i\gamma_4\gamma_5\) is local on the linear arrangement, whereas \(E'=i\gamma_1\gamma_6\) connects the two ends and produces the same measured footprint.

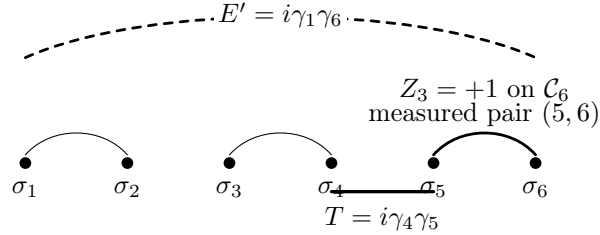
\begin{figure}[H]
\centering
\begin{tikzpicture}[x=1.35cm,y=0.9cm, every node/.style={font=\small}, line cap=round, line join=round]
  \foreach \x/\j in {0/1,1/2,2/3,3/4,4/5,5/6}{
    \fill (\x,0) circle (2.2pt);
    \node[below=3pt] at (\x,0) {$\sigma_{\j}$};
  }
  % Pairing arcs
  \draw (0,0.12) .. controls (0.25,0.55) and (0.75,0.55) .. (1,0.12);
  \draw (2,0.12) .. controls (2.25,0.55) and (2.75,0.55) .. (3,0.12);
  \draw[line width=1.05pt] (4,0.12) .. controls (4.25,0.55) and (4.75,0.55) .. (5,0.12);
  \node at (4.5,0.72) {measured pair $(5,6)$};
  \node at (4.5,1.08) {$Z_3=+1$ on $\cC_6$};
  % T support
  \draw[line width=1.25pt] (3,-0.42) -- (4,-0.42);
  \node[below=2pt] at (3.5,-0.42) {$T=i\gamma_4\gamma_5$};
  % E' long-range support
  \draw[dashed,line width=1.05pt] (0,1.55) .. controls (1.2,2.35) and (3.8,2.35) .. (5,1.55);
  \node[fill=white, inner sep=1.5pt] at (2.5,2.18) {$E'=i\gamma_1\gamma_6$};
\end{tikzpicture}
\caption{The six-\(\sigma\) syndrome example.  The code \(\cC_6\) lies in the fixed \(Z_3=+1\) sector of the pair \((5,6)\).  Both \(T\) and \(E'\) flip that measured pair charge and hence produce the same nontrivial footprint, but their residual difference acts logically on the code.  In a linear arrangement, \(T\) is adjacent while \(E'\) is long-range; on a cyclic arrangement the endpoints \(1\) and \(6\) may also be adjacent.}
\label{fig:six-sigma-syndrome}
\end{figure}

To give an explicit nontrivial family, use the standard Majorana realization of Ising fusion spaces \cite{NayakSimonSternFreedmanDasSarma2008}.  Let \(\gamma_1,\ldots,\gamma_6\) be Majorana operators satisfying
\begin{equation}\label{eq:majorana-clifford}
  \gamma_j^\dagger=\gamma_j,
  \qquad
  \{\gamma_j,\gamma_k\}=2\delta_{jk}I,
\end{equation}
and choose signs so that the pair-parity observables
\begin{equation}\label{eq:majorana-pair-parities}
  Z_1=i\gamma_1\gamma_2,
  \qquad
  Z_2=i\gamma_3\gamma_4,
  \qquad
  Z_3=i\gamma_5\gamma_6
\end{equation}
are precisely the \(\pm1\) fusion-channel observables above.  The total-vacuum sector is the \(+1\) eigenspace of \(Z_1Z_2Z_3\), up to the fixed overall sign convention already absorbed into the definition of the \(Z_j\)'s.  Thus, within \(\cH_6\),
\[
  Z_3=Z_1Z_2.
\]
The code \(\cC_6\) is therefore the \(+1\)-eigenspace of the single check \(Z_3\) inside the fixed-parity sector.  One may choose logical Pauli operators
\begin{equation}\label{eq:six-sigma-logical-paulis}
  \overline Z=Z_1,
  \qquad
  \overline X=i\gamma_2\gamma_3,
\end{equation}
for which \(\overline X\) flips the two logical basis states up to a common phase and commutes with the check.  In this sense \(\cC_6\) is also a small Majorana fermion stabilizer code, expressed here in fusion-space coordinates; see \cite{BravyiLeemhuisTerhal2010} for the general Majorana-code framework.  This translation is useful because it makes clear which part of the example is specifically categorical and which part is familiar stabilizer structure.

Now consider the adjacent bilinear
\begin{equation}\label{eq:ising-six-error-T}
  T:=i\gamma_4\gamma_5.
\end{equation}
It is Hermitian and unitary:
\[
  T^\dagger=T,
  \qquad
  T^2=I.
\]
Because it contains an even number of Majorana operators, it preserves the fixed total-parity sector.  Its commutation relations with the pair observables are
\begin{equation}\label{eq:ising-six-T-commutation}
  [T,Z_1]=0,
  \qquad
  \{T,Z_2\}=0,
  \qquad
  \{T,Z_3\}=0.
\end{equation}
Consequently, up to basis phases,
\begin{align}\label{eq:ising-six-T-action}
  T|+,+,+\rangle&\propto |+,-,-\rangle,\\
  T|-,-,+\rangle&\propto |-,+,-\rangle.
\end{align}
The logical label \(z_1\) is preserved, whereas the measured footprint \(z_3\) flips.  In particular,
\begin{equation}\label{eq:ising-six-sector-resolution}
  \Pi_+TP_6=0,
  \qquad
  \Pi_-TP_6=TP_6.
\end{equation}
For an arbitrary encoded state
\[
  |\phi_L\rangle=a|0_L\rangle+b|1_L\rangle,
\]
the no-error history gives the deterministic outcome \(+\), while the error \(T\) gives the deterministic outcome \(-\).  The coefficients \(a,b\) are unchanged by the syndrome readout because both logical basis states lie in the same measured sector before the error and in the same orthogonal sector after it.

\begin{proposition}[Exact six-\(\sigma\) footprint correction]\label{prop:six-sigma-footprint-correction}
For the code \(\cC_6\subset\cH_6\) in \eqref{eq:ising-six-code}, the commuting algebra generated by the pair-charge projectors \(\{\Pi_+,\Pi_-\}\) is syndrome-admissible for the error family
\[
  \cE_6=\{I,T\},
  \qquad T=i\gamma_4\gamma_5.
\]
The family is exactly correctable.  After measuring the \((5,6)\)-pair charge, apply the identity on outcome \(+\) and apply \(T\) on outcome \(-\).
\end{proposition}

\begin{proof}
Equation \eqref{eq:ising-six-noerror-safe} proves no-error safety.  Equations \eqref{eq:ising-six-sector-resolution} show that the two error representatives have definite and orthogonal measured footprints.  The Knill--Laflamme matrix is immediate:
\begin{align*}
  P_6 I^\dagger I P_6 &= P_6,\\
  P_6 T^\dagger T P_6 &= P_6,\\
  P_6 I^\dagger T P_6 &= P_6TP_6=0,\\
  P_6 T^\dagger I P_6 &= P_6TP_6=0.
\end{align*}
The cross terms vanish because \(T\cC_6\subset\Pi_-\cH_6\) while \(\cC_6\subset\Pi_+\cH_6\).  Thus Theorem \ref{thm:footprint-algebra-KL} applies.  More directly, if the outcome is \(+\), no correction is needed; if the outcome is \(-\), then \(T^\dagger=T\) and \(T^2=I\), so applying \(T\) returns every vector \(T|\phi_L\rangle\) to \(|\phi_L\rangle\).
\end{proof}

The same code also exhibits the nontrivial ambiguity that motivates decoding \emph{within} a footprint fibre.  Define
\begin{equation}\label{eq:ising-six-error-Eprime}
  E':=i\gamma_1\gamma_6.
\end{equation}
Then
\begin{equation}\label{eq:ising-six-Eprime-commutation}
  \{E',Z_1\}=0,
  \qquad
  [E',Z_2]=0,
  \qquad
  \{E',Z_3\}=0.
\end{equation}
Thus \(E'\) has the same measured \((5,6)\)-footprint as \(T\):
\[
  \Pi_-E'P_6=E'P_6,
  \qquad
  \Pi_+E'P_6=0.
\]
The two representatives are nevertheless not harmlessly equivalent.

\begin{proposition}[A nontrivial six-\(\sigma\) footprint fibre]\label{prop:six-sigma-ambiguous-fibre}
For the enlarged error family \(\{I,T,E'\}\), the \((5,6)\)-pair-charge measurement remains no-error safe and resolves \(I\) from the two nontrivial errors, but it does not make the family exactly correctable.  The errors \(T\) and \(E'\) lie in the same measured sector and
\[
  P_6T^\dagger E'P_6
\]
acts as a logical bit flip up to the phase convention of the logical basis.  In particular, it is not a scalar multiple of \(P_6\).
\end{proposition}

\begin{proof}
The common footprint statement follows from \eqref{eq:ising-six-T-commutation} and \eqref{eq:ising-six-Eprime-commutation}.  The composite \(W=T^\dagger E'=TE'\) preserves \(Z_3\) because both factors anticommute with it.  It anticommutes with \(Z_1\) and with \(Z_2\), so on the code basis it exchanges the two joint eigenvalue patterns
\[
  (+,+,+)\longleftrightarrow(-,-,+).
\]
Hence
\[
  W|0_L\rangle=e^{i\theta_0}|1_L\rangle,
  \qquad
  W|1_L\rangle=e^{i\theta_1}|0_L\rangle
\]
for phases \(\theta_0,\theta_1\).  Since \(T\) and \(E'\) are commuting Hermitian unitaries with disjoint Majorana supports, \(W\) is a Hermitian involution; after rephasing the logical basis, its compression is \(\overline X\).  Thus the within-fibre Knill--Laflamme entry \(P_6T^\dagger E'P_6\) is non-scalar, and Theorem~\ref{thm:footprint-algebra-KL} forbids exact correction of the enlarged family from this syndrome alone.
\end{proof}

This proposition turns the footprint fibre into an actual inference problem.  On a linear arrangement, \(T=i\gamma_4\gamma_5\) is adjacent while \(E'=i\gamma_1\gamma_6\) is long-range, so a locality-based prior can favour \(T\) after the outcome \(-\).  On a cyclic arrangement where modes \(6\) and \(1\) are also adjacent, that simple geometric preference disappears; the residual difference \(T^\dagger E'\) remains logical.  The example is finite, but it is the same logical structure encountered in topological decoding: one measured footprint, several compatible representatives, and a harmful residual class between them.

The Majorana bilinears also give a compact taxonomy.  Bilinears with exactly one endpoint in \(\{5,6\}\) anticommute with \(Z_3\) and are detected by the pair-charge check.  Bilinears supported away from \(\{5,6\}\) commute with that check and may act logically, as \(i\gamma_2\gamma_3=\overline X\) does.  The within-pair parities \(Z_j\) are diagonal in the chosen fusion basis.  This makes explicit both the utility and the limitation of the single measured footprint.

This example is intentionally finite and does not supply a growing distance.  It does, however, realize the central operational sequence of the paper without analogy:
\[
  \text{encoded state}
  \longrightarrow
  \text{specified error}
  \longrightarrow
  \text{measured footprint}
  \longrightarrow
  \text{conditional recovery}.
\]
The distinction from the four-puncture example is exact.  There the natural pair-charge observables act as complementary logical diagnostics.  Here the code is a proper subspace of the six-anyon fusion space, the \((5,6)\)-charge is fixed on the code, and the same type of field-theoretic measurement becomes a genuine syndrome.  This is the simplest role of redundancy in the present fusion-space language.

\subsection{Conformal blocks and a geometry-sensitive likelihood model}

Now place four \(\sigma\)-insertions at complex positions \(z_1,z_2,z_3,z_4\), and let
\begin{equation}\label{eq:crossratio}
x=\frac{(z_1-z_2)(z_3-z_4)}{(z_1-z_3)(z_2-z_4)}
\end{equation}
be the cross-ratio.  By a conformal transformation, the four points may be moved to \(0,x,1,\infty\).  The chiral Ising four-point blocks for the spin field \(\sigma\), in a standard choice of branch and up to a common convention-dependent normalization irrelevant for the normalized ratios, may be written as
\begin{align}
\mathcal F_\one(x)
&=
\frac{1}{\sqrt 2}\,[x(1-x)]^{-1/8}
\sqrt{1+\sqrt{1-x}},
\label{eq:block-one}\\
\mathcal F_\psi(x)
&=
\frac{1}{\sqrt 2}\,[x(1-x)]^{-1/8}
\sqrt{1-\sqrt{1-x}}.
\label{eq:block-psi}
\end{align}
The common prefactor encodes the singular scaling dictated by the external fields, while the square-root factor distinguishes the two possible intermediate fusion channels in \(\sigma\times\sigma\).  In the block-norm model of Definition \ref{def:CBdecoder}, for real \(0<x<1\), the common prefactor cancels from the normalized ratio and one obtains weights
\begin{align}
\Prob(\one\mid x)
&=\frac{1+\sqrt{1-x}}{2},
\label{eq:prob-one}\\
\Prob(\psi\mid x)
&=\frac{1-\sqrt{1-x}}{2}.
\label{eq:prob-psi}
\end{align}
These probabilities should be interpreted as a minimal conformal-block likelihood model, not as a universal physical noise law.  They express the fact that conformal geometry biases the relative likelihood of competing fusion channels.

For example, at the symmetric value \(x=1/2\), the simple block-norm model gives
\[
  \Prob(\one\mid x=1/2)
  =\frac{1+1/\sqrt2}{2}
  \approx 0.8536,
  \qquad
  \Prob(\psi\mid x=1/2)
  =\frac{1-1/\sqrt2}{2}
  \approx 0.1464.
\]
The likelihood ratio is
\[
  \frac{\Prob(\one\mid x)}{\Prob(\psi\mid x)}
  =\frac{1+\sqrt{1-x}}{1-\sqrt{1-x}}.
\]
This ratio diverges as \(x\to0\), where the first two punctures collide and the vacuum channel dominates, and tends to \(1\) as \(x\to1\) in this channel basis.  Thus even this smallest calculation exhibits the intended phenomenon: topology tells us which channels exist, while conformal geometry gives a quantitative preference among them.

Equivalently, the log-likelihood ratio is
\begin{equation}\label{eq:ising-loglikelihood}
  \Lambda(x)=\log\frac{\Prob(\one\mid x)}{\Prob(\psi\mid x)}
  =\log\bigl(1+\sqrt{1-x}\bigr)-\log\bigl(1-\sqrt{1-x}\bigr).
\end{equation}
For \(0<x<1\), \(\Lambda(x)>0\), so the maximum-likelihood channel in this particular \(s\)-channel normalization is always \(\one\), with confidence decreasing as \(x\) approaches \(1\).  This does not make the decoder trivial.  For this chosen pairing and this simple prior, the geometry supplies a graded confidence score.  In a many-puncture problem the same local likelihoods would be combined with competing pairings, physical noise rates, and global total-charge constraints.

Let us spell out what has and has not been normalized in \eqref{eq:prob-one}--\eqref{eq:prob-psi}.  The formulas use the two chiral blocks in a standard branch for \(0<x<1\).  Both blocks contain the same singular prefactor \([x(1-x)]^{-1/8}\), reflecting the external spin fields and the chosen coordinate normalization.  Since a decoder compares channels compatible with the same external insertions and the same observed footprint, this common prefactor cancels from the normalized ratio.  What remains is the relative channel dependence
\[
  1+\sqrt{1-x}
  \quad\text{versus}\quad
  1-\sqrt{1-x}.
\]
It is important to clarify that the probabilities in \eqref{eq:prob-one}--\eqref{eq:prob-psi} are not meant to be absolute four-point probabilities --- rather, they are \emph{normalized relative weights between the two internal channels in the simplified block-norm model}.

If one changes fusion tree, the same two-dimensional space is described by a different pair of channels, related by the \(F\)-matrix.  A decoder written in one channel basis must therefore be transformed before comparison with a footprint measurement naturally associated to another pairing.  This is a useful consistency check: a footprint is local to the region being measured, but the representation of that footprint in a chosen global basis depends on the fusion tree.  The field-theoretic formalism keeps track of this dependence through recoupling rather than through an arbitrary change of coordinates.

The limiting regimes are instructive.  As \(x\to 0\), the insertions at \(0\) and \(x\) approach one another.  The vacuum block dominates:
\[
\Prob(\one\mid x)\to 1,
\qquad
\Prob(\psi\mid x)\to 0.
\]
Thus a nearby pair of \(\sigma\)-insertions is overwhelmingly assigned to the vacuum channel in this simplest block-norm model.  As \(x\to 1\), the two displayed \(s\)-channel weights tend to equality. Equivalently, the geometrically natural description is moving toward a different pairing and hence a recoupled fusion tree.  The decoder's preference depends on geometry as well as topology.

\begin{table}[htbp]
\centering
\renewcommand{\arraystretch}{1.25}
\begin{tabular}{>{\raggedright\arraybackslash}p{0.18\linewidth}>{\raggedright\arraybackslash}p{0.27\linewidth}>{\raggedright\arraybackslash}p{0.43\linewidth}}
\toprule
\textbf{Regime} & \textbf{Dominant conformal block} & \textbf{Decoder interpretation}\\
\midrule
\(x\to 0\) & \(\one\)-channel & Nearby \(\sigma,\sigma\) pair fuses mostly to vacuum.\\
\(0<x<1\) & geometry-dependent mixture & Competing fusion histories carry different weights.\\
\(x\to 1\) & recoupled-channel competition & Alternate pairing becomes geometrically natural.\\
\bottomrule
\end{tabular}
\caption{Geometry-sensitive interpretation of the Ising conformal-block likelihood weights.}
\label{tab:ising-decoder}
\end{table}

Taken together, the Ising examples capture the main thesis in miniature.  The four-puncture system shows that natural charge measurements may be logical diagnostics.  The six-puncture code shows an actual syndrome and exact conditional recovery, while Proposition~\ref{prop:six-sigma-ambiguous-fibre} shows that the same syndrome can leave a harmful residual ambiguity.  Conformal-block weights then provide one possible source of soft information for ranking compatible histories, once a physical likelihood model has been specified.

The worked examples remain within the Ising theory and hence within its Majorana/Clifford shadow.  They establish the diagnostic, syndrome-admissible, and same-footprint ambiguity mechanisms explicitly, but they do not test the framework in a genuinely non-Clifford nonabelian theory or in the presence of nontrivial fusion multiplicities.  Those extensions require separate examples rather than an inference from the Ising calculations.

\subsection{Relation to the Knill--Laflamme condition}

Let \(P\) be the projection onto \(\Hom(\one,\sigma^{\otimes 4})\) inside a physical realization.  A local neutral insertion in a disk disjoint from the punctures acts by a scalar on the conformal-block space, by the same locality reasoning as Proposition \ref{prop:local-scalar}.  Thus the field-theoretic mechanism reproduces the Knill--Laflamme form
\[
P E_a^\dagger E_b P=\lambda_{ab}P
\]
for errors whose composite is locally neutral and contractible.  In contrast, a \(\psi\)-line separating punctures can distinguish the \(\one\) and \(\psi\) fusion channels and hence act nontrivially on the logical qubit.  This is a logical operator rather than a correctable local error.

One may phrase this as follows.

\begin{proposition}[Ising local protection]\label{prop:ising-local}
In the four-\(\sigma\) Ising conformal-block code with fixed total vacuum charge, any error composite represented by a locally neutral contractible defect network disjoint from the punctures acts as a scalar on \(\Hom(\one,\sigma^{\otimes 4})\).  Non-scalar logical action requires a defect network whose topology separates punctures or changes the global fusion-channel decomposition.
\end{proposition}

\begin{proof}
The first assertion is Proposition \ref{prop:local-scalar} applied to the Ising modular category.  A locally neutral contractible network disjoint from the four punctures can be enclosed by a disk whose boundary has total vacuum charge.  Its evaluation is therefore an element of \(\End(\one)\cong\bC\), and insertion of the disk into the four-puncture surface multiplies the whole fusion space by that scalar.

For the second assertion, recall that
\[
  \Hom(\one,\sigma^{\otimes 4})
  \cong
  \Hom(\one,\sigma\otimes\sigma)\otimes
  \Hom(\one,\sigma\otimes\sigma)
  \oplus
  \Hom(\psi,\sigma\otimes\sigma)\otimes
  \Hom(\psi,\sigma\otimes\sigma)
\]
in the \((12)|(34)\) channel.  The two summands are the two logical basis sectors.  An operator that distinguishes them, or changes one into the other after recoupling, must interact with the global fusion-channel decomposition.  Diagrammatically this requires a line or defect move which separates punctures, winds around them, or implements a nontrivial recoupling or braiding operation.  Such an operator is no longer a contractible neutral insertion in a puncture-free disk.  It is a logical operation rather than a correctable local error in the clean topological model.
\end{proof}

This proposition also clarifies the limitation of the example.  The four-puncture qubit is protected against local neutral insertions in the ideal topological sense, but it does not by itself provide a distance parameter growing with system size.  A scalable code would require a family of surfaces, punctures, defects, or lattices for which the minimal nontrivial logical network becomes increasingly costly relative to the physical noise model.  The example should therefore be viewed as the local building block of a larger architecture.  Its purpose is to demonstrate the simultaneous presence of all ingredients: field-theoretic state space, footprints, recoupling, conformal weights, and a Knill--Laflamme-type scalar condition.

Once the punctures are embedded in a larger surface or connected to additional defects, the same measured fusion channel can be compatible with several histories: braids, local insertions, pair-creation events, or Wilson-line segments.  A decoder must then compare those histories using the topology, conformal geometry, and physical noise data available in the model.

\section{Extensions and routes toward scalable architectures}\label{sec:scalable}

The finite Ising models in Section~7 are local tests by design.  The four-puncture qubit displays noncommuting diagnostic algebras and conformal-block weights; the six-puncture subspace displays a syndrome-admissible pair-charge measurement and exact recovery for the family \(\{I,i\gamma_4\gamma_5\}\), and then a same-footprint pair of errors whose residual action is logical.  Neither example has a growing distance.  A scalable architecture requires a family of physical realizations whose size grows, together with a measurement schedule, a noise model, a recovery protocol, and a distance or threshold statement.  Existing nonabelian-decoding work provides important benchmarks for this task, from active and fault-tolerant anyon correction to numerical Fibonacci decoding \cite{WoottonHutter2016,DauphinaisPoulin2017,BurtonBrellFlammia2017}; recent work also gives a direct fault-tolerant anyonic-computation scheme under local noise assumptions \cite{LyonsBrown2026}.

This section turns from the local formalism to growing code families.  The categorical core established the meaning of a footprint and the conditions under which measured footprint sectors support exact QEC.  Here we ask how those ingredients can be assembled in topological, conformal, lattice, measurement-based, and hyperbolic architectures.  We also prove a conditional threshold theorem.  It does not assert that every TQFT or CFT code has a threshold; it gives explicit local hypotheses under which a family of footprint codes has exponentially small logical failure below a nonzero noise strength.

\subsection{What must scale?}\label{subsec:what-must-scale}

A family of field-theoretic codes should consist of data
\begin{equation}\label{eq:family-code-data}
  \mathfrak Q_N=(\cF_N,\Sigma_N,\cD_N,\Gamma_N,\cE_N,\mu_N),
  \qquad N=1,2,\ldots,
\end{equation}
where the index may count punctures, lattice cells, tensor-network tensors, genus, hyperbolic volume, or some other geometric size parameter.  For each value of~\(N\), the realization \(\Gamma_N\) produces a physical Hilbert space, a code projector \(P_N\), a set of error representatives \(\cE_N\), and a likelihood model \(\mu_N\).  The field theory supplies the ideal state spaces, topological sectors, defect moves, and amplitude weights, but the engineering question is whether the combined data produce an increasingly robust family of encodings.

One may isolate three quantities which are natural from the footprint point of view.  First, there is a \emph{logical network size}: the least cost of a defect network which preserves all measured local footprints but acts non-scalarly on the encoded space.  In a surface code this is the usual homological distance.  In a fusion-category code it should be replaced by the least cost of a nontrivial anyonic or defect-network representative inside a footprint fibre.  Second, there is a \emph{footprint resolution}: the ability of the chosen measurement algebra to separate error representatives without revealing logical information.  Third, there is a \emph{decoder gap}: the degree to which the likelihood model favors correctable histories over logically nontrivial histories with the same measured footprint.

This leads to the following definition.

\begin{definition}[Footprint distance and footprint fibre]\label{def:footprint-distance}
Fix a field-theoretic code datum \(\mathfrak Q\), a syndrome-admissible footprint algebra \(\cA\), and a cost function \(w\) on error histories or defect networks.  For a measured sector \(s\in\Spec(\cA)\), the \emph{footprint fibre} over \(s\) is the collection
\[
  \cF_s=\{E\in\cE:\ E \text{ has measured footprint }s\}.
\]
The \emph{footprint distance} of \((\mathfrak Q,\cA,w)\) is the minimum value of \(w(E^\dagger F)\), among pairs \(E,F\in\cE\) lying in a common footprint fibre, such that the composite \(E^\dagger F\) acts non-scalarly on the code space.
\end{definition}

The definition leaves ``cost'' realization-dependent.  It may be Hamming weight on qubits, string length, number of violated local terms, spacetime volume of a measurement history, conformal action, negative log-likelihood, or a weighted combination of these.  The relevant distance is therefore not simply the support size of a local operator.  It is the least cost of a logically nontrivial ambiguity that remains after the selected footprints have been measured.

The same distinction appears in ordinary stabilizer decoding.  A short error chain and a long error chain can have the same boundary; only their difference cycle determines whether a logical operator has occurred.  Definition~\ref{def:footprint-distance} transfers this idea to defect networks: two histories with the same measured footprint may differ by a closed Wilson line, a braid, a condensed-wall loop, a nontrivial element of a tube algebra, or a monodromy operation on conformal blocks.  The scalable challenge is to make such harmful differences costly and statistically suppressed.

\begin{principle}[Scalable footprint design]\label{prin:scalable-footprint-design}
A field-theoretic architecture should be designed so that local footprint measurements are syndrome-admissible, local neutral composites remain scalar on the code, and the least-cost non-scalar composite inside a footprint fibre grows with the size of the realization.

The preceding sections also suggest a more precise mathematical task for architecture design.  It is not enough to exhibit a large family of state spaces or a large family of local projectors.  One must prove a compatibility statement of the following form: the chosen local projectors assemble into commuting measurement rounds; the no-error sector is independent of the encoded state; every allowed low-weight error has a definite measured footprint; and every same-footprint ambiguity either evaluates as a local scalar, belongs to a gauge subsystem, or has cost above the intended distance.  In other words, the categorical construction must be accompanied by a proof that the selected measurements are syndrome-admissible and that the harmful fibres are pushed to high cost.

This point is especially important in nonabelian theories.  Nonabelian fusion gives more diagnostic data, but more diagnostic data is not automatically better syndrome data.  Measuring an intermediate channel may distinguish logical basis states, as in the four-\(\sigma\) Ising qubit, and is therefore a logical measurement rather than an error syndrome for an unknown encoded state.  Scalable constructions should therefore separate three uses of local charge measurements: stabilizer-like checks that are syndrome-admissible, gauge measurements whose outcomes may be tracked but not corrected immediately, and logical measurements or gates that deliberately reveal or transform encoded information.

\end{principle}

Principle~\ref{prin:scalable-footprint-design} parallels the usual distance principle for topological stabilizer codes.  The difference is that the objects inside a fibre need not be binary chains.  They may be labelled string nets, defect worldsheets, annular tube-algebra sectors, conformal-block monodromies, or geometric degenerations in a moduli space.

\subsection{Local stochastic noise and residual histories}\label{subsec:local-stochastic-footprints}

The preceding design principle admits the following precise form.  The result is modeled on the Peierls counting mechanism behind topological-memory threshold arguments, but it is stated in the language of footprint decoding.  It should not be confused with a universal fault-tolerance theorem for arbitrary field theories.  A field theory may supply state spaces and defect labels without supplying a growing distance, a local decoder, or a physically meaningful noise model.  The theorem below says that, once these additional pieces satisfy explicit local hypotheses, the footprint formalism gives a nonzero threshold for quantum memory.  In this respect it sits between the exact Knill--Laflamme theorem of Section~3 and the more architecture-specific threshold theorems for surface-code and cluster-state models \,\cite{AharonovBenOr2008,DennisKitaevLandahlPreskill2002,RaussendorfHarringtonGoyal2007}.  The result concerns quantum memory, not passive self-correction, and is compatible with no-go phenomena for low-dimensional self-correcting stabilizer memories \,\cite{BravyiTerhal2009}.

Let \(L\) be a linear size parameter.  A \emph{spacetime footprint-decoding datum of size \(L\)} consists of the following data:
\begin{enumerate}[label=(\roman*)]
\item a finite set \(\Omega_L\) of elementary fault locations, usually the cells of a bounded-degree spacetime complex or a set of labelled local defect moves;
\item a set \(\mathsf{Hist}_L\) of allowed labelled error histories, each with a support \(\supp(E)\subseteq\Omega_L\) and size \(|E|:=|\supp(E)|\);
\item a measured-footprint set \(\mathsf S_L\) and a footprint map
\[
        \partial_{\mathrm{fp}}:\mathsf{Hist}_L\longrightarrow \mathsf S_L;
\]
\item a decoder
\[
        D_L:\mathsf S_L\longrightarrow \mathsf{Hist}_L,
\]
where \(D_L(s)\) is a chosen recovery history compatible with the measured footprint \(s\);
\item a residual-composition operation assigning to each actual history \(E\) a closed residual history
\[
        \mathcal R_L(E)=E\star D_L(\partial_{\mathrm{fp}}E),
        \qquad
        \partial_{\mathrm{fp}}\mathcal R_L(E)=0;
\]
\item a logical-action map from closed residual histories to operators on the encoded space, considered up to scalar.
\end{enumerate}
Here \(\partial_{\mathrm{fp}}=0\) denotes the no-footprint sector.  The notation \(E\star D_L(\partial_{\mathrm{fp}}E)\) is schematic because the relevant composition depends on the realization.  In a Pauli stabilizer code it is multiplication of Pauli errors modulo phase.  In a chain-complex model it is addition of chains.  In an anyonic or string-net model it is concatenation of defect histories followed by local fusion and isotopy moves.  Only the following coarse features of this composition are used in the threshold estimate: the residual history has trivial measured footprint, decomposes into connected residual components, and has a well-defined logical action.  We use the standard memory-threshold convention that every noisy syndrome-extraction round under consideration is included in the spacetime fault complex, followed by an ideal final classical decoding step and ideal application of the selected recovery.  A theorem with noisy recovery circuitry would require those additional locations to be included in \(\Omega_L\).

The set \(\Omega_L\) is equipped with an adjacency relation.  A subset \(K\subseteq\Omega_L\) is connected if it is connected in this adjacency graph.  When histories carry labels, the word ``connected history'' includes both a connected support and the compatible local labels on that support.  All label multiplicities are absorbed into the counting constant below.

\begin{lemma}[A crude connected-region bound]\label{lem:connected-region-count}
Let \(G\) be a graph of maximum degree \(\Delta\ge 2\).  For a fixed vertex \(x\), the number of connected \(m\)-vertex subsets containing \(x\) is at most
\[
  \Delta^{2(m-1)}.
\]
If each support of size \(m\) admits at most \(M^m\) compatible local labellings, the number of labelled connected regions is at most \(M^m\Delta^{2(m-1)}\).
\end{lemma}

\begin{proof}
Fix an ordering of the incident edges at every vertex.  Every connected \(m\)-vertex set \(K\) containing \(x\) has a canonical rooted spanning tree, obtained for example by breadth-first search with these fixed tie-breaking orders.  The depth-first traversal of this tree is a walk of length \(2(m-1)\) starting at \(x\), and its visited vertex set is exactly \(K\).  Thus the map from \(K\) to its canonical traversal is injective.  There are at most \(\Delta^{2(m-1)}\) walks of that length from \(x\).  Multiplying by the assumed label count gives the labelled version.
\end{proof}

Thus bounded-degree spacetime geometry with uniformly bounded local label growth implies a condition of the form (P1) below, albeit with deliberately crude constants.

\begin{definition}[Local stochastic noise]\label{def:local-stochastic-noise}
Following the standard locally decaying stochastic model used in fault-tolerance theory \cite{Gottesman2014ConstantOverhead}, a probability distribution on histories \(E\in\mathsf{Hist}_L\) is \emph{local stochastic with strength \(p\)} if, for every finite set \(F\subseteq\Omega_L\) of elementary fault locations,
\[
        \Prob(F\subseteq\supp(E))\leq p^{|F|}.
\]
Equivalently, the probability that any prescribed collection of \(|F|\) elementary faults all occurs is at most \(p^{|F|}\), without any independence assumption between distinct faults.
\end{definition}

This is the usual adversarially correlated local stochastic model.  It is more general than independent Bernoulli noise, since it permits correlations, but it still excludes arbitrary adversarial noise by requiring prescribed sets of faults to be exponentially suppressed.

\begin{definition}[Peierls footprint family]\label{def:peierls-footprint-family}
A family of spacetime footprint-decoding data \(\{\mathfrak Q_L\}_{L\geq 1}\) is called a \emph{Peierls footprint family} if there are constants
\[
        A>0,
        \qquad B>0,
        \qquad \delta>0,
        \qquad \kappa>0,
\]
independent of \(L\), with the following properties.
\begin{enumerate}[label=(P\arabic*)]
\item \textbf{Bounded connected-region growth.}  For every \(m\geq 1\) and every \(x\in\Omega_L\), the number of labelled connected residual regions of size \(m\) containing \(x\) is at most \(A B^m\).

\item \textbf{Local neutralizability below scale \(L\).}  Every closed residual history all of whose connected residual components have size strictly smaller than \(\delta L\) acts as a scalar on the encoded space.  Equivalently, every non-scalar closed residual history contains a connected residual component of size at least \(\delta L\).  In punctured or bounded realizations, this hypothesis includes the requirement that a component treated by contractible-vacuum scalarity lie in a neighbourhood free of protected punctures, nontrivial boundary sectors, and unresolved multiplicity data; components meeting such structures must be controlled by a separate local neutralization or gauge argument.

\item \textbf{Componentwise decoder balance.}  If \(K\) is a connected residual component of \(\mathcal R_L(E)=E\star D_L(\partial_{\mathrm{fp}}E)\), then the actual error has at least a fixed fraction of faults inside \(K\):
\[
        |\supp(E)\cap K|\geq \kappa |K|.
\]

\end{enumerate}

\end{definition}

The second condition is the distance condition in field-theoretic form.  It requires every closed residual history built from small connected pieces to be harmless: each component can be locally neutralized, evaluated as a scalar, or absorbed into a gauge subsystem.  Proposition~\ref{prop:categorical-scalarity-correctability} supplies one sufficient mechanism only for puncture-free contractible vacuum components; (P2) is intentionally broader and must account separately for boundaries, punctures, and other protected defects.  The third condition is the decoding condition.  It prevents the decoder from creating, at no cost, a long residual component whose support is mostly recovery rather than actual noise.  For ordinary minimum-weight decoding of binary chain errors, this balance is automatic with \(\kappa=1/2\), as recorded below.

\begin{theorem}[Peierls threshold for footprint decoding]\label{thm:peierls-footprint-threshold}
Let \(\{\mathfrak Q_L\}_{L\geq 1}\) be a Peierls footprint family with constants \(A,B,\delta,\kappa\).  Suppose that, for each \(L\), the actual error history is drawn from a local stochastic noise model of strength \(0\leq p\leq 1\).  Set
\[
        q(p)=2B p^\kappa .
\]
If \(q(p)<1\), then the logical failure probability of footprint decoding satisfies
\begin{equation}\label{eq:peierls-threshold-bound}
        \Prob_L(\mathrm{fail})
        \leq
        \frac{A|\Omega_L|}{1-q(p)}\, q(p)^{\delta L}.
\end{equation}
In particular, there is a nonzero threshold
\[
        p_0=\min\{1,(2B)^{-1/\kappa}\}>0
\]
such that, for every \(p<p_0\), there are constants \(C,c>0\), independent of \(L\), for which
\[
        \Prob_L(\mathrm{fail})\leq C|\Omega_L|e^{-cL}.
\]
If \(|\Omega_L|\) grows at most polynomially in \(L\), then the logical failure probability tends to zero exponentially in \(L\).
\end{theorem}

\begin{proof}
Let \(E\) be the actual error history and let
\[
        R=D_L(\partial_{\mathrm{fp}}E)
\]
be the recovery chosen by the decoder.  The residual history
\[
        \mathcal R_L(E)=E\star R
\]
has trivial measured footprint.  A logical failure occurs exactly when this closed residual history acts nontrivially on the encoded space, up to scalar.

By local neutralizability below scale \(L\), a non-scalar closed residual history must contain a connected residual component \(K\) with
\[
        |K|=m\geq \delta L .
\]
By componentwise decoder balance, the actual error has at least \(\kappa m\) elementary faults in \(K\):
\[
        |\supp(E)\cap K|\geq \kappa m.
\]
Thus a logical failure implies the existence of a connected residual region \(K\subseteq\Omega_L\) of some size \(m\geq \delta L\) containing at least \(\kappa m\) actual faults.

Fix such a connected region \(K\) of size \(m\).  The event that at least \(\kappa m\) actual faults occur inside \(K\) is contained in the union, over all subsets \(F\subseteq K\) with \(|F|=\lceil\kappa m\rceil\), of the events \(F\subseteq\supp(E)\).  By local stochasticity,
\[
        \Prob(F\subseteq\supp(E))\leq p^{|F|}\leq p^{\kappa m}.
\]
There are at most \(2^m\) such subsets.  Therefore
\[
        \Prob\bigl(|\supp(E)\cap K|\geq \kappa m\bigr)
        \leq 2^m p^{\kappa m}.
\]

It remains to sum over possible connected regions.  By bounded connected-region growth, the number of labelled connected residual regions of size \(m\) containing a fixed location is at most \(AB^m\).  Summing first over a marked location in \(\Omega_L\) gives the harmless overcount
\[
        \#\{K:\ K\text{ connected},\ |K|=m\}\leq A|\Omega_L|B^m.
\]
Hence
\[
\begin{aligned}
        \Prob_L(\mathrm{fail})
        &\leq
        \sum_{m\geq \delta L}
        A|\Omega_L|B^m 2^m p^{\kappa m}  \\
        &=
        A|\Omega_L|
        \sum_{m\geq \delta L} (2Bp^\kappa)^m .
\end{aligned}
\]
If \(q(p)=2Bp^\kappa<1\), the last sum is bounded by
\[
        \frac{A|\Omega_L|}{1-q(p)}q(p)^{\delta L},
\]
which proves \eqref{eq:peierls-threshold-bound}.  For fixed \(p<p_0\), put
\[
        c=-\delta\log q(p)>0,
        \qquad
        C=\frac{A}{1-q(p)}.
\]
Then \(q(p)^{\delta L}=e^{-cL}\), giving the stated exponential bound.
\end{proof}

\begin{remark}[What the theorem does and does not prove]\label{rem:threshold-scope}
Theorem~\ref{thm:peierls-footprint-threshold} proves a threshold for a family once the geometric and decoding hypotheses have been verified.  It does not assert that an arbitrary fusion category, modular category, conformal-block theory, or string-net Hamiltonian automatically gives such a family.  In particular, the theorem contains the hard architecture questions as hypotheses: one must still construct syndrome-admissible measurement rounds, prove a growing neutralization scale, and analyze a decoder.  The value of the theorem is that these requirements are expressed in the same footprint language used throughout the paper.
\end{remark}

\begin{proposition}[Minimum-weight balance in the chain case]\label{prop:minimum-weight-balance}
Suppose that the histories form a binary chain model: errors and recoveries are chains with support in a bounded-degree cell complex, the measured footprint is the boundary, and the residual is the mod-two sum \(E+R\).  If \(D_L\) chooses a minimum-weight chain with the measured boundary, then the componentwise balance condition holds with \(\kappa=1/2\).
\end{proposition}

\begin{proof}
Let \(K\) be a connected component of the residual cycle \(E+R\).  Since \(E+R\) has zero boundary, the restrictions \(E|_K\) and \(R|_K\) have the same boundary.  Both restrictions are supported inside \(K\), so replacing the recovery on \(K\) leaves every other residual component unchanged.  If \(|R\cap K|>|E\cap K|\), then replacing \(R|_K\) by \(E|_K\) gives another recovery chain with the same measured boundary and strictly smaller total weight.  This contradicts the minimum-weight property of \(R\).  Hence \(|R\cap K|\leq |E\cap K|\).  Since the residual component is contained in the union of the error and recovery supports on \(K\),
\[
        |K|\leq |E\cap K|+|R\cap K|\leq 2|E\cap K|.
\]
Therefore \(|E\cap K|\geq |K|/2\), which is the balance condition with \(\kappa=1/2\).
\end{proof}

\begin{example}[Recovery of the surface-code Peierls mechanism]\label{ex:surface-code-peierls}
For a surface-code memory, \(\Omega_L\) may be taken to be the set of spacetime fault locations or, in the simplest phenomenological model, the set of lattice edges on which error chains live.  The measured footprint is the boundary of an error chain, the decoder chooses a minimum-weight chain with the observed boundary, and a closed residual chain is harmful exactly when it contains a homologically nontrivial component.  If the code has distance proportional to \(L\), then every homologically nontrivial residual component has size at least \(\delta L\).  Proposition~\ref{prop:minimum-weight-balance} supplies \(\kappa=1/2\), while bounded-degree lattice geometry supplies constants \(A\) and \(B\).  Theorem~\ref{thm:peierls-footprint-threshold} then recovers the standard qualitative conclusion: below a nonzero error rate, the logical failure probability is exponentially small in the linear size.  The constants are crude because the argument isolates the field-theoretic mechanism without attempting numerical threshold optimization.
\end{example}

\begin{remark}[Nonabelian and conformal versions]\label{rem:nonabelian-threshold}
In a nonabelian anyonic or conformal-block architecture, the same theorem applies only after the three Peierls hypotheses are checked in the labelled defect calculus.  The bounded-region constant \(B\) must absorb the growth of admissible labels and local fusion multiplicities.  The neutralizability hypothesis must assert scalarity, or harmless gauge action, for all small closed labelled residual components.  The balance hypothesis must be proved for the chosen decoder; fusion rules alone do not imply it.  The abstract formulation isolates the mathematical work required to turn the local footprint formalism into a scalable architecture.
\end{remark}

\subsection{Many-punctured conformal-block codes}\label{subsec:many-punctures}

The most direct scaling of the Ising example uses \(2m\) Ising \(\sigma\)-punctures with total vacuum charge.  The fusion space grows as \(2^{m-1}\), as the following elementary count shows.  This growth is not a distance statement.  It gives a family of fusion spaces on which one may impose redundancy, local footprint measurements, and recovery rules.

\begin{proposition}[Ising fusion-space count]\label{prop:ising-count}
Let \(N_n(a)\) denote the number of fusion paths of \(n\) Ising \(\sigma\)-objects with total charge \(a\).  Then
\[
  N_{2m}(\one)=N_{2m}(\psi)=2^{m-1},
  \qquad
  N_{2m+1}(\sigma)=2^m,
\]
and all other total charges have multiplicity zero.
\end{proposition}

\begin{proof}
For \(n=1\), the only possible total charge is \(\sigma\), so \(N_1(\sigma)=1\).  Fusion by another \(\sigma\) uses the Ising rules
\[
  \one\otimes\sigma=\sigma,
  \qquad
  \psi\otimes\sigma=\sigma,
  \qquad
  \sigma\otimes\sigma=\one\oplus\psi.
\]
Hence, after an odd number of \(\sigma\)-labels, only total charge \(\sigma\) is possible, while after an even number only \(\one\) and \(\psi\) are possible.  The multiplicities satisfy
\[
  N_{2m+1}(\sigma)=N_{2m}(\one)+N_{2m}(\psi),
  \qquad
  N_{2m+2}(\one)=N_{2m+1}(\sigma)=N_{2m+2}(\psi).
\]
Starting from \(N_2(\one)=N_2(\psi)=1\), this recursion doubles the number of paths every two tensor factors.  Induction gives
\[
  N_{2m}(\one)=N_{2m}(\psi)=2^{m-1},
  \qquad
  N_{2m+1}(\sigma)=2^m.
\]
The remaining multiplicities are zero by the parity statement just proved.  In particular, fixing total vacuum charge for \(2m\) punctures gives \(2^{m-1}\) fusion paths, which is \(\dim\Hom(\one,\sigma^{\otimes 2m})\).
\end{proof}

For many punctures, the geometry of marked spheres is part of the model rather than auxiliary data.  The natural base is
\begin{equation}\label{eq:M0n-base}
  \left((\bP^1)^n\setminus \bigcup_{i<j}\{z_i=z_j\}\right)/\operatorname{PGL}_2,
\end{equation}
or, after compactification, a version of \(\overline M_{0,n}\).  Collision divisors record operator-product expansions.  Boundary strata record fusion degenerations.  Paths in the configuration space record braid histories.  Thus conformal geometry supplies likelihood data for compatible histories, while the categorical footprint algebra supplies the measured local sectors.

The practical question is how to choose a commuting family of local total-charge measurements.  One simple pattern is a fusion tree with adjacent pair or block measurements.  Another is a set of overlapping clusters measured in different rounds.  In one round the measured algebra must be commutative, but different rounds may involve noncommuting diagnostic algebras.  The four-anyon Ising example exhibits this pattern: the \((12)\)-measurement is naturally \(Z\)-type, while the \((23)\)-measurement becomes \(X\)-type after recoupling.  In a large puncture system, a schedule of local recouplings and total-charge measurements can therefore be understood as a field-theoretic analogue of alternating stabilizer checks.

Such a schedule can be encoded by a set of clusters
\[
  \mathcal I_t=\{I_{t,1},\ldots,I_{t,r_t}\}
\]
for each measurement time \(t\).  The clusters within a fixed \(\mathcal I_t\) should be compatible, so that their projectors generate a commutative algebra \(\cA_t\).  The algebras \(\cA_t\) and \(\cA_{t+1}\) need not commute.  The complete syndrome record is then a spacetime word
\[
  s=(s_1,s_2,\ldots,s_T),
  \qquad s_t\in\Spec(\cA_t),
\]
and decoding is inference over defect histories whose time-sliced footprints match this word.  Thus the measurement schedule converts a static fusion space into a spacetime decoding problem.

The difficult problem is the organization of fault tolerance around these state spaces.  One must specify how punctures are created, moved, measured, and possibly removed.  One must also specify how leakage out of the intended total charge sector is detected, how local footprints are repeatedly measured without revealing the logical state, and how the resulting classical data are decoded.  These are precisely the points at which the field-theoretic datum \(\mathfrak Q\) must be supplemented by \(\Gamma\), \(\cE\), and \(\mu\).

\subsection{A toy Ising-chain schedule}\label{subsec:toy-ising-chain}

For a concrete schedule, consider \(2m\) Ising \(\sigma\)-punctures arranged along a line or circle with total vacuum charge.  A first round measures the fusion outcomes of disjoint adjacent pairs
\[
  (1,2), (3,4),\ldots,(2m-1,2m).
\]
This is a commuting family of pairwise total-charge measurements.  In the fusion basis adapted to these pairings, the outcomes are diagonal labels in \(\{\one,\psi\}\), subject to the total-vacuum constraint.  This round is therefore a many-anyon analogue of a \(Z\)-type stabilizer layer.

A second round may measure shifted pairs
\[
  (2,3), (4,5),\ldots,(2m-2,2m-1),
\]
or larger overlapping clusters.  These measurements are not diagonal in the first fusion basis.  They are obtained by a sequence of \(F\)-moves, a local total-charge projection, and inverse \(F\)-moves.  In the four-anyon case this is exactly the calculation giving \(FZF^{-1}=X\).  In a longer chain, the shifted measurements produce local recoupling probes of the fusion path.  Thus alternating pairings give a categorical analogue of alternating check layers.

One should not confuse this toy schedule with a complete code.  Pair measurements of all adjacent anyons would generally reveal too much information if used naively on a small logical subspace.  A scalable design must choose a protected subspace, a gauge structure, or a redundancy constraint so that the measured pair and block outcomes reveal error data but not logical data.  Nevertheless, the toy schedule shows how a measurement pattern can be built from the primitive operations of the theory.

For such a schedule, an elementary local fault may have several effects:
\begin{enumerate}[label=(\roman*)]
\item it may change a pair-fusion outcome in the current round;
\item it may create a charge which is invisible until a later shifted round;
\item it may change the recoupling path between two fusion bases;
\item it may produce a braid or monodromy phase without changing a coarse charge label;
\item it may leak out of the total vacuum sector.
\end{enumerate}
The footprint record must decide which of these effects are actually measured.  The unmeasured effects remain inside the footprint fibre and must be controlled by distance, likelihood, or gauge redundancy.

This example also points to a useful graphical object: the \emph{recoupling graph}.  Its vertices are fusion trees for \(2m\) leaves and its edges are elementary \(F\)-moves.  A local measurement is diagonal at some vertex of this graph.  A measurement in a different channel is diagonal after moving to another vertex.  Error histories can therefore be represented by paths in an enlarged graph whose edges include recouplings, braids, local insertions, and measurement outcomes.  A decoder on a many-puncture conformal-block code is, in part, an inference problem on this graph.

The recoupling graph uses the categorical data essentially.  The \(F\)-symbols change basis between diagnostic contexts, the \(R\)-symbols govern braid phases and mixing, and quantum dimensions enter state-sum weights.  A stabilizer description retains the special case in which these transformations reduce to Clifford-compatible binary data; the footprint formalism retains the full categorical labels.

\subsection{Repeated measurement and spacetime defects}\label{subsec:spacetime-defects}

A static topological code becomes fault tolerant only after repeated noisy syndrome extraction is included.  The same is true here.  If footprints are measured in time, errors become spacetime defect histories rather than purely spatial defect segments.  A false measurement outcome, a missed detection, or a faulty recoupling move is itself part of the spacetime history that must be decoded.

In the surface-code setting, repeated syndrome extraction produces a three-dimensional decoding complex: two spatial dimensions plus time.  Error chains can end both in space and along temporal measurement defects.  The field-theoretic analogue is a labelled stratified spacetime.  Spatial defects carry anyon labels or categorical charges.  Time-like sheets carry changing boundary conditions or measurement choices.  Junctions carry fusion or condensation data.  The measured footprint is the restriction of this spacetime network to the observed measurement slices.

Repeated measurements lead to a spacetime footprint map:
\begin{equation}\label{eq:spacetime-footprint-map}
  \Foot_T:\{\text{spacetime defect histories}\}\longrightarrow
  \prod_{t=1}^T\Spec(\cA_t).
\end{equation}
Two histories in the same fibre of \(\Foot_T\) are observationally indistinguishable under the chosen measurement schedule.  A decoder should choose a recovery class by minimizing an effective cost or maximizing a posterior weight over that fibre.  Topological amplitudes, conformal-block norms, and physical noise probabilities can all contribute to this weight.

This formulation does not require an initial dichotomy between ``data errors'' and ``measurement errors'': both are parts of a spacetime defect network.  In a circuit implementation, they may have different microscopic probabilities.  In the field theory, they are different labelled local pieces of the same history.  This is especially natural for measurement-based computation, where the distinction between resource-state geometry and measurement schedule is itself part of the computational model.

\subsection{Decoder graphs, hypergraphs, and nonabelian fibres}\label{subsec:decoder-hypergraphs}

In an abelian stabilizer code, elementary faults can often be encoded by a linear boundary map.  A fault produces a syndrome vector, syndromes add mod~\(2\), and decoding becomes a minimum-weight problem over chains with a fixed boundary.  Nonabelian field-theoretic codes need not admit such a linearization.  Fusion outcomes may be noninvertible, multiplicity spaces may occur, and the order of local operations may matter.  Even so, one can often build a useful decoder graph or hypergraph by linearizing only the measured footprint data.

Let \(\mathcal F\) be a set of elementary fault types.  A measurement schedule determines the set \(\mathcal S\) of elementary detection events: changed charge labels, unexpected pair outcomes, leaked total charge, failed condensation, or time-like measurement inconsistencies.  Each fault \(f\in\mathcal F\) has an incidence pattern
\[
  \partial_{\mathrm{fp}} f\subseteq \mathcal S.
\]
In an abelian code this incidence pattern is a vector over \(\bZ_2\).  In a nonabelian code it may be a labelled hyperedge: the same fault can create several detection events with correlated charge labels.  A decoder graph is recovered only after forgetting enough label data to obtain pairwise binary events.

At the circuit level, this construction is close to the detector-error-model formalism used in contemporary stabilizer QEC: elementary fault mechanisms are mapped to sets of detector flips and logical frame changes, with graph or hypergraph structure retained when faults trigger correlated events \cite{Gidney2021Stim,DerksEtAl2024}.  The dictionary is
\[
\begin{array}{c|c}
\text{footprint language} & \text{detector-error-model language}\\ \hline
\text{elementary fault history} & \text{fault mechanism}\\
\text{detection event} & \text{detector flip}\\
\text{incidence hyperedge} & \text{correlated detector pattern}\\
\text{coarse footprint decoder} & \text{graph/hypergraph decoder}
\end{array}
\]
The extra structure sought here is the retention of nonabelian labels, fusion constraints, residual logical action, and geometry-dependent weights after this coarse detector pattern has been formed.

This motivates a two-stage decoder: a coarse abelian pass identifies candidate fault regions, after which categorical labels and amplitudes refine the inference inside those regions.  Such a strategy mirrors practical decoders for nonabelian or highly correlated noise: an initial combinatorial pass reduces the search space, while a local tensor-network or dynamic-programming pass handles the richer internal labels.

The footprint fibre is the natural mathematical object behind this two-stage strategy.  The coarse decoder identifies a large fibre over an abelianized footprint.  The refined decoder decomposes this fibre into categorical subfibres.  A field-theoretic likelihood then assigns weights to the remaining histories.  Symbolically, one has maps
\[
  \{\text{histories}\}
  \xrightarrow{\ \Foot\ }
  \{\text{categorical footprints}\}
  \xrightarrow{\ \mathrm{ab}\ }
  \{\text{abelianized syndromes}\}.
\]
A conventional stabilizer decoder operates at the right-hand end.  A fully categorical decoder operates at the middle term.  A geometry-sensitive decoder also uses additional weights on the left-hand term.

This hierarchy allows incremental development of the programme.  Early examples need not begin with a complete nonabelian decoder.  It may be enough, in early examples, to use field-theoretic data to improve a conventional decoder locally: adjust edge weights near defects, distinguish degeneracies that a binary syndrome would merge, or identify when a nominally correct recovery differs by a nontrivial anyonic loop.

\subsection{String-net and Turaev--Viro realizations}\label{subsec:string-net-tv}

Turaev--Viro codes and Levin--Wen string-net models provide a direct route to lattice realizations.  Given a unitary fusion category \(\cC\), the string-net Hilbert space is spanned by edge labelings subject to fusion constraints, and the ground space on a surface is a topological code.  For the Fibonacci category, the doubled Fibonacci theory is computationally powerful, and thresholds have been estimated numerically for the corresponding Turaev--Viro code \cite{KoenigKuperbergReichardt2010,SchotteZhuBurgelmanVerstraete2022}.

From the footprint point of view, violations of vertex or plaquette constraints are local measured sectors.  Equivalently, they are emergent anyon charges in \(Z(\cC)\), the Drinfeld center of the input category.  Decoding is inference over compatible string-net histories.  A fully developed state-sum decoder would assign weights to those histories by Turaev--Viro amplitudes, tube-algebra data, or conformal-block likelihood factors when a CFT realization is available.

The phrase \emph{state-sum decoder} should be understood cautiously.  It need not mean that one literally sums the full Turaev--Viro invariant for every possible error history.  Rather, it means that the local ingredients of the decoder should be inherited from the same state-sum data as the code itself: \(F\)-symbols, quantum dimensions, admissibility constraints, tube-algebra sectors, and boundary conditions.  A statistical decoder built from these data would be the nonabelian analogue of the random-plaquette gauge model associated with surface-code decoding.

In such a model, a defect history \(H\) would be assigned a weight of the schematic form
\begin{equation}\label{eq:statesum-weight}
  W(H)=\mu(H)\,\prod_{v} A_v(H)\prod_e A_e(H)\prod_f A_f(H),
\end{equation}
where \(\mu(H)\) is the physical noise probability and the remaining factors are local categorical or state-sum amplitudes.  Equation~\eqref{eq:statesum-weight} is a template whose ingredients depend on the model.  In some models the amplitudes may be signs, phases, quantum dimensions, or local Boltzmann weights.  In others, a tensor-network contraction may replace the explicit product.  The essential point is that the decoder should know the same local fusion data that define the code.

\subsection{Condensation, domain walls, and code deformation}\label{subsec:condensation}

Anyon condensation gives a parallel language for boundaries, domain walls, and code deformation.  In categorical terms, a boundary or wall is often described by a module category, a bimodule category, or a condensable algebra object; lattice and categorical models of gapped boundaries and domain walls provide a standard reference point \cite{KitaevKong2012}.  A topological charge which condenses at a boundary becomes locally invisible there; a charge which does not condense remains detectable.  In footprint language, the boundary changes which local charges are visible.

In the color-code setting, fault-tolerant logical operations and Floquet-style dynamical codes can be interpreted through condensation processes and spacetime domain walls \cite{KesselringEtAl2024}; the dynamically generated logical-qubit construction of Hastings--Haah is a canonical reference for the Floquet-code viewpoint \cite{HastingsHaah2021}.  In the present terminology, a time-like domain wall is a spacetime defect.  A measurement sequence is a controlled evolution of boundary conditions or condensable algebras.  The footprint is the pattern of charges that fail to condense, are transported across the wall, or are converted into other sectors.

In this language, code deformation changes the footprint map in a controlled way.  Suppose a wall \(W\) separates phases \(\cC_1\) and \(\cC_2\).  A charge \(a\in\cC_1\) approaching the wall may map to a sum of charges in \(\cC_2\), may be absorbed, or may be confined.  The relevant diagnostic question is how the local sector data transform across the wall:
\[
  a\quad \longmapsto\quad \sum_b N^W_{a,b}\,b.
\]
The integers or multiplicity spaces \(N^W_{a,b}\) encode a wall-crossing footprint.  A decoder across a deformed code must track these transformations in spacetime.

Condensation becomes essential once fault-tolerant operations are included.  Boundaries and code deformations are how many topological codes implement logical gates, prepare states, and change layouts.  A treatment of fault-tolerant operations must therefore include dynamical changes of the measured sector algebra.  Condensation supplies one categorical mechanism for doing so.

\subsection{Tensor-network and holographic realizations}\label{subsec:tensor-holographic}

Tensor-network codes provide another intermediate setting between abstract field theory and hardware-level circuits.  A tensor network assigns finite-dimensional tensors to vertices and contracts internal legs according to a graph.  In many topological examples, these tensors can be interpreted as discretized pair-of-pants decompositions, string-net tensors, or state-sum amplitudes.  In holographic examples, the network geometry controls which bulk degrees of freedom are protected against which boundary erasures, in the broader QEC interpretation of bulk reconstruction \cite{AlmheiriDongHarlow2015,PastawskiYoshidaHarlowPreskill2015}.

The footprint language has a natural interpretation in such networks.  A local bulk error can be pushed through tensors toward the boundary.  A boundary region sees the induced boundary action of the bulk operator, or the obstruction to pushing that operator away.  Thus the boundary signature of a bulk error is a footprint.  In stabilizer tensor networks this footprint may be an ordinary Pauli syndrome or a boundary stabilizer pattern.  In a categorical tensor network it may be a boundary charge sector, an annular idempotent, or a fusion-channel label.

The usual erasure-correction question then admits a field-theoretic refinement.  Besides asking whether a bulk operator can be reconstructed on a boundary region, one can ask which local footprints are visible there.  Let \(R\) be a boundary subsystem and \(R^{c}\) its complement.  A bulk defect history \(H\) may leave a footprint on \(R\), on \(R^c\), or only on their union.  The QEC question is whether these footprints distinguish correctable local errors without distinguishing logical bulk information.

Tensor networks are also useful computationally.  A categorical decoder may be expensive if formulated as a sum over labelled histories.  Tensor-network contraction gives an approximate or exact way to perform such sums when the geometry has bounded treewidth, hyperbolic hierarchy, or exploitable local structure.  This computational role has direct precedents in QEC, including tensor-network formulations of decoding and maximum-likelihood surface-code decoding \cite{FerrisPoulin2014,BravyiSucharaVargo2014}.  It is especially relevant for string-net and conformal-block models, where the same graphical calculus that defines the state space can also be used to contract likelihoods.

A distinction is needed here.  Holographic codes often emphasize erasure correction, complementary recovery, and bulk reconstruction, whereas the present paper concerns local noise, measured footprints, and syndrome-based recovery.  These questions overlap but are not identical.  A tensor-network realization of the footprint programme must specify the physical errors, the available boundary or bulk measurements, and the decoder.  Holographic geometry can organize these data, but it does not replace syndrome-admissibility.

\subsection{Measurement-based and hyperbolic routes}\label{subsec:mbqc-hyperbolic}

Measurement-based quantum computation supplies another route from local field-theoretic data to scalable architectures.  In the Raussendorf--Harrington--Goyal construction, topological fault tolerance is obtained from a three-dimensional cluster state, with defects and boundary conditions implementing protected operations \cite{RaussendorfHarringtonGoyal2007}.  The later programme of generating fault-tolerant cluster states from crystal structures expands this idea by using tilings and lattice geometry as design variables \cite{NewmanDeCastroBrown2020}.  In this language, the measurement pattern also defines a spacetime geometry.

The footprint formalism fits naturally with this point of view.  A cluster-state measurement pattern produces a classical record.  The record is interpreted by local parity constraints, correlation surfaces, and defect boundaries.  These are stabilizer shadows of a broader field-theoretic situation in which the observed data are local footprints of a spacetime history.  When the resource state has a topological or categorical origin, the same record can be lifted from Pauli parity language to charge-sector language.

Hyperbolic geometry is relevant because negative curvature can improve asymptotic rate or overhead relative to Euclidean constructions.  Canonical constructions include two-dimensional hyperbolic surface codes with constant-rate tradeoffs and threshold studies \cite{BreuckmannTerhal2016}, as well as homological codes from higher-dimensional arithmetic hyperbolic manifolds \cite{GuthLubotzky2014}.  Recent work develops and benchmarks CSS codes on hyperbolic lattices \cite{MahmoudAliRayan2026}, while hyperbolic cluster states have been proposed as constant-rate three-dimensional MBQC resources with threshold behaviour comparable to Euclidean cluster-state constructions \cite{MahmoudTournaireBachmannRayan2026}.  These results do not follow from the categorical footprint theorem.  They instead provide concrete geometric families in which footprint-based questions can be tested.

The field-theoretic question is what additional structure negative curvature should carry.  A hyperbolic lattice already has non-Euclidean combinatorics.  It may also support natural moduli, holonomy, Fuchsian group symmetries, and representation-theoretic data.  A footprint decoder on such a geometry could therefore combine several layers:
\[
\begin{array}{ccl}
\text{local check data} &\rightsquigarrow& \text{stabilizer or categorical footprints},\\
\text{hyperbolic combinatorics} &\rightsquigarrow& \text{growth, rate, and logical-network structure},\\
\text{holonomy or moduli} &\rightsquigarrow& \text{geometry-sensitive likelihood corrections},\\
\text{physical noise} &\rightsquigarrow& \text{hardware-specific weights}.
\end{array}
\]
There is no claim that every hyperbolic code should be recast in conformal or categorical terms.  Rather, hyperbolic geometry brings topology, representation theory, and decoding into contact, making it a plausible setting in which to test the footprint formalism.

\subsection{Logical operations as controlled motion in footprint space}\label{subsec:logical-motion}

A code architecture must support both memory and logical operations.  In a topological anyon model, logical gates arise from braiding, fusion measurement, magic-state injection, or code deformation.  In a surface-code architecture, logical operations arise from lattice surgery, twist defects, boundary motion, and measurement-based protocols.  In the footprint language, all of these are controlled motions through a space of measurement contexts.

Let \(\cA_0,\cA_1,\ldots,\cA_T\) be a sequence of syndrome-admissible algebras.  A fault-tolerant operation should transport the code space through this sequence while keeping harmful histories distinguishable or unlikely.  The ideal operation is the induced parallel transport, braid, wall-crossing functor, or recoupling transformation on the protected subspace.  Faults during the operation are spacetime defects whose footprints are measured by the changing algebras.

This separates a logical gate into two components:
\[
  \text{logical gate}
  =
  \text{controlled field-theoretic transport}
  +
  \text{fault-tolerant footprint monitoring}.
\]
The first component may be topological, as in braiding.  It may be conformal, as in monodromy of conformal blocks.  It may be categorical, as in a sequence of \(F\)- and \(R\)-moves.  It may be geometric, as in transport over a moduli space.  The second component is the QEC requirement: the monitoring must not reveal the logical state, and the remaining ambiguities must be correctable.

This framing is particularly useful for distinguishing mathematical universality from fault-tolerant universality.  A modular category may have braid group representations dense in a unitary group, as in the standard modular-functor universality results \cite{FreedmanLarsenWang2002}, or it may require supplemental measurements to become universal.  That is a statement about the ideal protected space.  A fault-tolerant architecture must also show that the corresponding spacetime histories can be monitored by syndrome-admissible footprints with a growing distance or suppressing likelihood.  The footprint formalism keeps these two requirements separate.

\subsection{From local tests to architecture principles}\label{subsec:architecture-principles}

The examples above suggest a practical workflow for designing field-theoretic QEC architectures.

The mathematical requirements can be organized as a three-level test.  At the first level one checks 
\emph{sector existence}: the proposed local measurement should arise from an actual decomposition of the state space, such as a fusion-tree edge, a tube-algebra idempotent, a boundary-condition summand, or a conformal-block factorization channel.  At the second level one checks 
\emph{sector compatibility}: the measurements used in one round should define a commuting algebra or a specified ordered instrument.  At the third level one checks 
\emph{sector harmlessness}: the information not measured inside each fibre should satisfy the fibrewise Knill--Laflamme scalar equations.  The first level is categorical, the second is operational, and the third is genuinely error-correcting.

A concrete architecture requires a field-theoretic input; a physical or combinatorial realization \(\Gamma\); a family of local footprint projectors; a syndrome-admissible commuting subalgebra for each measurement round; a noise model and likelihood rule on histories; and, finally, a distance, threshold, or finite-size scaling analysis.

These checks prevent algebraic size, geometric structure, or categorical richness from being mistaken for fault tolerance; robustness still requires the compatibility and scaling conditions above.

For each proposed architecture, four questions should be answered.
\begin{enumerate}[label=(\roman*)]
\item What is the local field-theoretic datum measured by the syndrome apparatus?
\item Which local histories have the same measured footprint?
\item Which differences between such histories act logically?
\item What geometry, topology, or physical noise model suppresses those logical differences?
\end{enumerate}
The questions do not depend on a particular hardware platform; they apply equally to anyons, surface-code lattices, cluster states, conformal blocks, oscillator codes, and geometric quantizations.

The next two sections change register.  Sections~2--8 contain the formal definitions, exact-correction results, worked Ising examples, and the conditional Peierls threshold criterion.  Sections~9--10 develop representation-theoretic and geometric extensions whose status ranges from concrete test constructions to longer-term proposals.  Unless a theorem or proposition is stated explicitly, these sections should be read as directions to be verified in specific code families, not as consequences of the preceding results.

\section{Representation-theoretic enlargements}\label{sec:repr-directions}

The categorical core was phrased in terms of unitary fusion categories, many of which arise from representation theory.  Representation-theoretic models also make projectors, braidings, traces, and sector decompositions explicit.  We consider several such directions, all based on the observation that a measured footprint may be realized by an idempotent or isotypic projection for an algebra action.

\subsection{Hopf-algebraic realizations of sector projectors}\label{subsec:hopf-sector-projectors}

Quantum groups and Hopf algebras give one standard algebraic source for the categories appearing above.  In Reshetikhin--Turaev theory, quantum-group representation categories provide braided and, after appropriate semisimplification or root-of-unity constructions, modular tensor categories from which three-dimensional TQFTs are built \cite{Drinfeld1985,Drinfeld1986,Jimbo1985,ReshetikhinTuraev1991,Turaev1994,ChariPressley1994,Kassel1995}.

Let \(H\) be a quasi-triangular Hopf algebra and let \(\operatorname{Rep}(H)\) be a suitable semisimple category of finite-dimensional representations.  The coproduct
\[
  \Delta:H\longrightarrow H\otimes H
\]
controls tensor products.  The antipode gives duality, and a universal \(R\)-matrix gives braiding operators.  A collection of local insertions labelled by modules \(V_1,\ldots,V_m\) therefore has a composite footprint obtained by decomposing
\[
  V_1\otimes\cdots\otimes V_m
\]
into simple or admissible summands.  In this setting the footprint projectors are representation-theoretic sector projectors.

In the semisimple case, central idempotents give a concrete model.  If
\[
  V_1\otimes\cdots\otimes V_m\cong \bigoplus_a W_a\otimes M_a,
\]
where \(W_a\) runs over simple modules and \(M_a\) is the multiplicity space, then the central idempotent \(e_a\) projects onto the \(a\)-isotypic component.  The coarse footprint records \(a\).  A refined footprint may record a measured multiplicity label or, more invariantly, a projector in \(M_a\).  These idempotents are algebraic sector projectors.  Their physical measurability is a separate condition on the realization \(\Gamma\) and on the measurement model.

Fusion multiplicities should therefore be retained.  In a multiplicity-free category, a total charge label may determine a local sector completely.  In a category with nontrivial fusion multiplicities, the same total charge may occur through several channels.  A footprint can then be coarse, recording only the simple summand \(W_a\), or fine, recording part of the multiplicity space \(M_a\).  Whether the fine data are measurable without disturbing logical information is an operational question.  Mathematically, however, the multiplicity space is part of the footprint fibre.

\subsection{Tube algebras and annular footprints}\label{subsec:tube-algebra}

String-net models naturally lead from local fusion projectors to annular algebras.  The Ocneanu tube algebra records labelled annuli and their composition; its simple modules classify bulk anyons in the Drinfeld center \(Z(\cC)\) \cite{Ocneanu1994Tube}.  This point of view is closely related to the theorem that the center of a suitable tensor category is modular \cite{Mueger2003}.  For footprint decoding, the annular picture is attractive because many local measurements are not point-like.  They surround a region and ask what total charge is enclosed.

A disk-like footprint answers: what charge exits a cluster?  An annular footprint answers: what charge threads an annulus?  In a lattice realization, this is the difference between a vertex violation and a loop or plaquette measurement enclosing a region.  In a modular category, the annular measurement is controlled by the action of closed labelled loops.  These loop operators are simultaneously diagonalized by topological charge sectors, and their eigenvalues are expressed through modular data.

The footprint algebra can therefore be enlarged beyond projectors attached to tensor-factor clusters.  One may also include annular operators generated by tube-algebra idempotents.  A measured sector can then record the bulk anyon type enclosed by a loop in addition to the fusion outcome of a small cluster.  In surface-code language this remains a syndrome measurement; in nonabelian language it measures a central sector of the tube algebra.

The annular language also clarifies the relation between stabilizer checks and categorical traces.  A stabilizer plaquette is a loop operator with eigenvalues \(\pm1\).  A string-net plaquette term is a weighted sum of loop insertions, with weights involving quantum dimensions.  Thus the usual binary check is recovered when the annular algebra has only abelian characters, while the general tube-algebra check records a richer charge sector.

\subsection{Yangians, affine Yangians, and commuting transfer data}\label{subsec:yangian-directions}

Yangians and affine Yangian-type algebras require more care.  Their inclusion would be justified only when the code family carries additional structure, such as a spectral parameter, an integrable transfer-matrix description, a quiver-variety realization, or a rational or trigonometric degeneration of conformal-block transport \cite{Molev2007,MaulikOkounkov2019,Tsymbaliuk2017}.  Mere formal analogy is not enough.

The most conservative role for such algebras is to organize commuting families of operators.  In integrable systems, transfer matrices often produce commuting observables indexed by a spectral parameter.  If a code family admits a compatible transfer-matrix construction, then some footprint algebras might be realized as specializations or degenerations of these commuting observables.  The measured sector would then be a joint eigenspace, and the decoder would use the residual ambiguity inside the joint spectrum.

A second possible role is geometric.  Yangian actions appear naturally in the quantum cohomology of quiver varieties and related symplectic resolutions \cite{MaulikOkounkov2019}.  If a family of codes is built from a moduli space whose geometry carries such an action, then the representation-theoretic decomposition may organize logical spaces, defect sectors, or wall-crossing operations.  The proposal is speculative but testable: one must identify a code state space, an algebra action on it, and syndrome-admissible sector projectors.

A third role is through rational conformal limits.  KZ-type connections, quantum-group monodromy, and Yang--Baxter structures are already linked in conformal field theory and integrable models.  In a decoding context, monodromy around collision divisors may generate recoupling and braid operations, while commuting Hamiltonians may control likelihood flows or asymptotics.  Here again, the key question is not whether Yangians are present in the surrounding mathematics.  It is whether their idempotents or joint spectra define measured footprints that satisfy the operational constraints of QEC.

\subsection{Nonsemisimple and modified-trace questions}\label{subsec:nonsemisimple}

The present paper uses a finite semisimple unitary fusion-category core because it gives a vacuum-sector scalar conclusion: a neutral contractible composite evaluates in \(\End_{\cC}(\one)\cong\bC\).  This conclusion is central to the categorical Knill--Laflamme argument.  Many interesting representation categories, however, are not semisimple.  Root-of-unity quantum groups before semisimplification, logarithmic conformal field theories, and categories with projective modules all suggest possible extensions beyond the present hypotheses.

In a nonsemisimple setting, the phrase ``acts as a scalar'' may need to be replaced by a statement involving negligible morphisms, modified traces, projective ideals, or a quotient category.  The footprint projectors themselves may no longer be honest orthogonal projections in a unitary Hilbert space.  They may instead be idempotents in a Karoubian envelope, projectors after semisimplification, or generalized eigenspace projectors for a nonnormal operator.

The difficulty lies in the QEC interpretation rather than in the footprint idea itself.  Exact correction requires an inner product, distinguishable measurement sectors, and recovery maps.  A nonsemisimple categorical model must therefore be paired with an analytic or physical realization that explains how generalized sector data become observable.  The mathematical direction is nevertheless promising, because nonsemisimple categories contain logarithmic and critical phenomena that are invisible in purely semisimple theories.

\section{Algebro-geometric and analytic directions}\label{sec:geometric-directions}

The conformal-block construction already introduces geometry through puncture positions, sewing parameters, and block norms.  This section asks how much further that dependence can be developed.  Besides topological sector labels, a decoder may in principle use information carried by meromorphic connections, spectral curves, Jacobians, and abelian varieties.

\subsection{Conformal blocks over moduli and sewing likelihoods}\label{subsec:moduli-sewing}

For conformal-block codes on marked curves, the state spaces form vector bundles over moduli.  In genus zero, the relevant open moduli space is \(M_{0,n}\), and its stable compactification \(\overline M_{0,n}\) records collisions and degenerations of marked points.  The conformal-block bundle carries a projectively flat connection, such as the Knizhnik--Zamolodchikov connection in Wess--Zumino--Witten models \cite{KnizhnikZamolodchikov1984,TsuchiyaUenoYamada1989,BeauvilleLaszlo1994}.  Parallel transport along paths in the configuration space implements braid and monodromy operations.

Near a boundary divisor of \(\overline M_{0,n}\), sewing coordinates identify a degeneration in which the curve splits into two components joined by a node.  Factorization expresses the conformal-block space as a direct sum over intermediate charges:
\begin{equation}\label{eq:sewing-factorization-expanded}
  \cV_{\vec X}(C)
  \cong
  \bigoplus_{a\in\Irr(\cC)}
  \cV_{\vec X_I,a}(C_1)\otimes
  \cV_{a^*,\vec X_{I^c}}(C_2),
\end{equation}
up to the usual choices and normalizations.  This is the analytic counterpart of the categorical total-charge decomposition.  The intermediate label \(a\) is precisely the footprint sector associated with the degeneration.

A geometry-sensitive decoder can use this factorization in two ways.  First, it can interpret a measured local charge as selecting one summand in \eqref{eq:sewing-factorization-expanded}.  Second, it can use the asymptotic behavior of conformal blocks in the sewing parameter \(q\).  A typical expansion has the schematic form
\begin{equation}\label{eq:conformal-block-asymptotic}
  \mathcal F_a(q)=q^{h_a-h_I-h_{I^c}}\bigl(c_a+O(q)\bigr),
\end{equation}
where \(h_a\) is a conformal weight and the exponent depends on the chosen channel.  Thus the same footprint sector may be more or less likely depending on the geometry.  In a topological theory only the label \(a\) remains.  In a conformal theory the position in moduli changes the relative weights of the sectors.

This makes the distinction between footprint and likelihood especially clear.  The footprint is the sector label \(a\).  The likelihood is the weight attached to that sector, which may depend on \(q\), on Hermitian metrics, on normalization conventions, and on the physical noise model.  A conformal-block decoder therefore has an inference problem that varies over moduli.

\subsection{Higgs bundles and spectral curves}\label{subsec:higgs-spectral-expanded}

There is a natural algebro-geometric extension of the same formalism over \(\bP^1\).  The conformal-block side supplies the connection picture: conformal-block bundles over \(M_{0,n}\) carry projectively flat meromorphic connections.  In semiclassical regimes one may pass to a parallel Higgs-bundle-theoretic picture, where a meromorphic connection is replaced by a Higgs field and its spectral curve, as in nonabelian Hodge theory and the spectral correspondence \cite{Hitchin1987,BeauvilleNarasimhanRamanan1989,Simpson1992}.

Let \(L=\mathcal O_{\bP^1}(t)\), let \(V\) be a rank \(r\) vector bundle on \(\bP^1\), and let
\[
  \Phi:V\longrightarrow V\otimes L
\]
be an \(L\)-twisted Higgs field.  The spectral curve \(S_\Phi\subset\operatorname{Tot}(L)\) is defined by
\begin{equation}\label{eq:spectral-equation-expanded}
  \det(\eta-\pi^*\Phi)=0,
\end{equation}
where \(\pi:\operatorname{Tot}(L)\to\bP^1\) is the projection and \(\eta\) is the tautological section of \(\pi^*L\).  For generic \(\Phi\), this is an \(r\)-fold cover of \(\bP^1\).  The discriminant of \eqref{eq:spectral-equation-expanded} has degree \(r(r-1)t\).  By Riemann--Hurwitz, the genus of a smooth generic spectral curve is
\begin{equation}\label{eq:spectral-genus-expanded}
  g(S_\Phi)=1-r+\frac{1}{2}r(r-1)t.
\end{equation}

Equation \eqref{eq:spectral-genus-expanded} matters because the generic Hitchin fibre is described by line bundles on \(S_\Phi\).  In the simplest smooth case it is an open subset of \(\Pic(S_\Phi)\), often subject to determinant or norm constraints.  The Hitchin system thereby replaces a nonabelian problem on \(\bP^1\) by abelian data on a spectral cover.  For decoding, the question is whether this abelianization can turn nonabelian local information into computable geometric weights.

For a QEC interpretation, additional code data are required.  The proposed dictionary is as follows:
\[
\begin{array}{ccl}
\text{Higgs field }\Phi &\rightsquigarrow& \text{semiclassical field configuration},\\
\text{spectral curve }S_\Phi &\rightsquigarrow& \text{branched cover carrying abelianized data},\\
\text{Jacobian }\Jac(S_\Phi) &\rightsquigarrow& \text{space of line-bundle or phase data},\\
\text{discriminant locus} &\rightsquigarrow& \text{collision, singularity, or enhanced-error locus},\\
\text{Hitchin base} &\rightsquigarrow& \text{parameter space for geometry-sensitive likelihoods}.
\end{array}
\]
The footprint formalism enters through the discriminant and factorization data.  A local collision of eigenvalues of \(\Phi\), or a degeneration of \(S_\Phi\), leaves locally visible algebraic data.  These data are not an error history themselves.  They are the algebro-geometric footprint of such a history.

\subsection{Rank-$2$ test cases on \texorpdfstring{\(\bP^1\)}{P1}}\label{subsec:rank-2-test-cases}

The formula \eqref{eq:spectral-genus-expanded} gives immediate small laboratories.  For rank \(r=2\) and twist \(t=2\), a generic spectral curve has genus
\[
  g(S_\Phi)=1-2+\frac{1}{2}\cdot2\cdot1\cdot2=1.
\]
Thus the generic Hitchin fibre is elliptic in nature.  For rank \(r=2\) and twist \(t=3\), the genus is
\[
  g(S_\Phi)=1-2+\frac{1}{2}\cdot2\cdot1\cdot3=2.
\]
The passage from \(t=2\) to \(t=3\) therefore changes the abelianized fibre from elliptic-curve data to genus-two Jacobian data.  These are small enough to compute explicitly but already rich enough to exhibit singular fibres and nontrivial discriminants.

For a rank-$2$ field, the spectral equation may be written locally as
\begin{equation}\label{eq:rank-2-spectral}
  \eta^2-a_1(z)\eta+a_2(z)=0,
\end{equation}
where \(a_1\in H^0(\bP^1,L)\) and \(a_2\in H^0(\bP^1,L^2)\), after choosing a trivialization.  The branch divisor is controlled by the discriminant
\begin{equation}\label{eq:rank-2-discriminant}
  \Delta(z)=a_1(z)^2-4a_2(z).
\end{equation}
When \(L=\mathcal O(t)\), the discriminant is a section of \(\mathcal O(2t)\), and its zeros are the branch points of the double cover.  Collisions of zeros of \(\Delta\) produce singular spectral curves.  From the footprint perspective, the local pattern of such collisions is a candidate algebraic footprint.

One possible toy decoder would proceed as follows.  Choose a family of Higgs fields whose spectral curves remain in a controlled region of the Hitchin base.  Interpret small perturbations of the coefficients \(a_i\) as geometric error histories.  Measure only coarse local data of the discriminant, such as whether branch points have entered specified windows or whether a local degeneration type has occurred.  The footprint fibre then consists of all perturbations with the same measured discriminant data.  A recovery rule chooses a correction by minimizing a geometric cost on the Hitchin base or on the corresponding Jacobian.

At this stage the construction remains an algebraic model of the same inverse problem.  A code construction additionally requires a Hilbert space, an encoding, an inner product, and physical measurement operations.  Nevertheless, such rank-$2$ models are useful because every piece can be computed: spectral curves, discriminants, Jacobians, singular fibres, and monodromy around discriminant strata.

\subsection{Matrix-valued polynomial models}\label{subsec:matrix-polynomial-models}

The rank-$2$ discussion extends to matrix-valued polynomial models.  Suppose
\[
  V=\bigoplus_{i=1}^r\mathcal O(d_i),
  \qquad L=\mathcal O(t),
\]
and consider an \(L\)-twisted Higgs field \(\Phi:V\to V\otimes L\).  The entry \(\Phi_{ij}\) is a section of
\[
  \mathcal O(-d_i+d_j+t).
\]
Thus the choice of splitting type \((d_1,\ldots,d_r)\) and twist \(t\) determines an explicit space of polynomial matrices.  Stability, spectral curves, and the Hitchin map can then be studied by direct algebraic methods.

For the footprint programme, this gives a bridge between abstract moduli and computable finite-dimensional data.  A local perturbation of a matrix entry can change the spectral discriminant.  A gauge transformation can move the perturbation without changing the underlying Higgs bundle.  A singular spectral curve can indicate a collision of eigenvalues.  These are algebraic analogues of local errors, gauge redundancies, and measured footprints.

The quotient by bundle automorphisms is important.  Two polynomial matrices may represent the same Higgs-bundle point.  Thus a geometric decoder should not infer histories in the raw affine space of matrix entries unless the gauge redundancy is included.  The correct footprint fibre is a quotient object: histories with the same measured local algebraic data, modulo transformations that do not change the encoded geometric state.

This mirrors the stabilizer situation.  In a stabilizer code, two error operators differing by a stabilizer act the same on the code.  In a Higgs-bundle model, two perturbations differing by a gauge transformation may represent the same geometric state.  The analogy is structural: in both cases, decoding must be performed modulo an equivalence relation.

\subsection{Jacobians, abelian varieties, and GKP analogies}\label{subsec:jacobians-gkp}

The appearance of Jacobians suggests a bridge to continuous-variable and oscillator codes.  The original Gottesman--Kitaev--Preskill construction encodes finite-dimensional quantum information into an oscillator using phase-space lattice symmetries \cite{GottesmanKitaevPreskill2001}.  A recent algebro-geometric framework of Mayrand and Royer relates GKP-type codes to polarized complex abelian varieties: in their dictionary, symplectically integral lattices define polarized abelian varieties, finite-dimensional code spaces are spaces of theta functions, logical Pauli gates arise from the theta group, and concatenation with stabilizer codes corresponds to isogeny \cite{MayrandRoyer2026}.

This fits naturally with the geometric direction above.  A smooth spectral curve \(S\) has a Jacobian \(\Jac(S)\), a complex abelian variety.  Quantization of line bundles on abelian varieties produces theta-function spaces.  If a Hitchin fibre or compactified Jacobian is used as an auxiliary geometric space for a field-theoretic code, then the Mayrand--Royer dictionary suggests a way to connect the resulting abelian variety to oscillator-code ideas.  The resulting connection would be a common geometric language in which conformal-block and GKP constructions use abelian varieties and theta functions to organize quantum states.

The most concrete speculative path is the following.  Start from a conformal-block or fusion-space code whose semiclassical limit is controlled by a Hitchin system.  Move to a region of the Hitchin base where the spectral curve is smooth.  Abelianize the data on \(S\), producing a Jacobian or Prym variety.  Quantize a polarization on this abelian variety, obtaining a theta-function space.  Then ask whether local field-theoretic footprints correspond to short displacement errors, theta-group characters, or isogeny kernels in the associated abelian variety.  If so, the GKP and fusion-code languages would meet through spectral geometry.

Several distance-like quantities can be compared in this setting.  In GKP codes, failure is controlled by short nontrivial displacements relative to a lattice.  In topological codes, it is controlled by short nontrivial cycles or defect networks.  Jacobians bring lattices, polarizations, theta groups, and systolic invariants into the same discussion.  One concrete question is whether footprint distance admits a systolic or polarization-dependent interpretation in an appropriate geometric limit.

\subsection{Continuous-variable footprints}\label{subsec:cv-footprints}

The abelian-variety direction also points toward a version of the footprint idea for continuous-variable codes.  In a GKP-type code, small phase-space displacements are correctable up to lattice equivalence.  The measured data record displacement modulo the stabilizer lattice, while a logical error corresponds to a displacement that crosses into a nontrivial coset.  Thus the syndrome is a phase-space footprint: a local displacement leaves a measured residue, but that residue does not determine the exact displacement history.

In the usual square-lattice GKP picture, the relevant geometry is flat phase space with a symplectic lattice.  In the Mayrand--Royer algebro-geometric picture, the same structure is organized by polarized complex abelian varieties and theta functions \cite{MayrandRoyer2026}.  This replacement is significant for the present programme because it puts continuous-variable QEC into a language closer to Jacobians, polarizations, and moduli.  The field-theoretic question becomes: when does a geometric quantization of an abelian variety carry a footprint measurement analogous to a stabilizer displacement syndrome?

One possible abstraction is the following.  Let \(A\) be a polarized abelian variety with a finite-dimensional theta-function space \(H^0(A,L)\).  Let \(K\) be a finite subgroup or finite quotient associated with the polarization, theta group, or an isogeny.  Small displacement errors are elements of a continuous group, while measured syndromes record their images modulo a lattice or finite subgroup.  The footprint fibre consists of displacements with the same measured residue.  The logical ambiguity is the residual action on \(H^0(A,L)\).

The comparison with the topological case is concrete.  A surface-code error chain is observed through its boundary and remains ambiguous up to homology.  A GKP displacement is observed modulo a lattice and remains ambiguous up to a finite logical quotient.  On a spectral-curve Jacobian, cycles give a lattice, theta functions give a state space, and polarizations define finite quotients.  Whether these structures can be connected by an actual QEC model is the substantive question behind the Jacobian direction.

The continuous-variable setting also emphasizes approximate correction.  Ideal GKP codewords are not normalizable, and physical realizations use finite-energy approximations.  Thus the exact scalar condition of the present paper would need to be relaxed to an approximate Knill--Laflamme condition.  This changes the role of likelihoods because tails of approximate wavefunctions contribute to failure probabilities.  A geometric footprint theory for oscillator codes would therefore have to combine the exact sector language of theta groups with analytic estimates coming from finite-energy states.

\subsection{Singular spectral curves and compactified fibres}\label{subsec:singular-spectral}

Any geometric extension must also treat singular loci.  Spectral curves become singular along the discriminant, conformal blocks degenerate at boundary divisors, and physical punctures may collide.  Error histories can drive the system toward these loci, so excluding them would remove precisely the regimes in which the geometric model may become most informative or most unstable.

When \(S_\Phi\) is singular, the ordinary Jacobian is replaced by a compactified Jacobian or a moduli space of torsion-free rank-one sheaves.  The fibre may acquire multiple components, singular strata, or vanishing cycles.  From a decoding perspective, these features may indicate enhanced ambiguity.  A singular point can mean that several sectors have become difficult to distinguish, that a local approximation has broken down, or that additional massless or low-cost histories have appeared.

There is a limited but potentially useful analogy with footprint fibres.  At a discriminant point of the Hitchin base, the abelianized description ceases to be generic.  A footprint fibre, by contrast, consists of histories not distinguished by the chosen measurement.  In both cases, the available data fail to separate objects that are distinct before projection to the base.  This raises the question of whether singularity theory, compactified Jacobians, or wall-crossing can diagnose ill-conditioning in a geometric decoder.

For \(\bP^1\)-based twisted Higgs bundles, these questions are especially concrete.  If \(V=\bigoplus_i\mathcal O(d_i)\) and \(L=\mathcal O(t)\), then the entries of an \(L\)-twisted Higgs field have degrees
\[
  \deg(\Phi_{ij})=-d_i+d_j+t.
\]
The result is an explicit matrix-valued polynomial model.  In low rank and small twist, the discriminant, spectral curve, and singular fibres can be computed directly.  These cases provide concrete tests of the proposed correspondence between local algebraic degenerations and QEC footprints \cite{Rayan2013NYJM,RayanSundbo2018,Rayan2018SIGMA}.

\subsection{Meromorphic connections, Stokes data, and irregular footprints}\label{subsec:stokes-footprints}

The connection side of nonabelian Hodge theory also suggests a refinement.  Meromorphic connections with irregular singularities have monodromy data enriched by Stokes matrices.  Boalch's symplectic approach to moduli spaces of meromorphic connections gives a geometric framework for such generalized monodromy data \cite{Boalch2001}.  In a field-theoretic decoding problem, Stokes data can be viewed as an analytic footprint of an irregular singularity.

The analogy has concrete content.  A regular singular point contributes monodromy around a puncture.  An irregular singular point contributes additional sectorial data depending on asymptotic directions.  If an error history produces or interacts with an irregular defect, then the locally visible data may include not just a charge label but also a Stokes sector.  A decoder would then need to infer histories compatible with both monodromy and Stokes footprints.

Such a direction would move the present framework from rational conformal blocks toward wild character varieties and irregular connections.  It is significantly more analytic than the finite fusion-category core, and it would require a careful Hilbert-space interpretation.  Nevertheless, it fits the same pattern: local singular behavior leaves sector data, the measurement model decides which sector data are visible, and decoding is inference over histories with fixed measured data.

\subsection{Geometric likelihoods and optimization over moduli}\label{subsec:geometric-likelihoods}

The geometric directions above all point to a broader idea: likelihoods may themselves be moduli-dependent.  In ordinary decoding, one often assumes independent identically distributed local noise.  In a field-theoretic model, the probability or amplitude of an error history may depend on geometry.  Puncture separations, conformal cross-ratios, hyperbolic distances, spectral-curve degenerations, and abelian-variety systoles can all affect which histories dominate.

A geometry-sensitive decoder would have two tasks: decode for a fixed geometry and help select the geometry itself.  For example, one might optimize puncture positions to maximize the gap between likely correctable histories and likely logical histories.  One might choose a polarization on an abelian variety to improve a shortest-displacement invariant.  One might choose a hyperbolic lattice or tiling to increase rate while maintaining a useful footprint distance.  These are design problems over moduli spaces.

The categorical theorem and the geometric discussion play different roles.  The theorem gives the local exact-correction constraint: neutral composites in the relevant footprint fibres must act as scalars.  Geometry can then supply additional data for ranking non-neutral or globally ambiguous histories.  A successful architecture would use the categorical condition to control local logical leakage and geometric information to suppress or distinguish global ambiguities.

\subsection{A consolidated geometric dictionary}\label{subsec:geometric-dictionary}

The geometric material can be summarized in a dictionary.  The entries are not definitions of a single theory; rather, they identify recurring roles played by geometric objects in possible field-theoretic decoders.

\begin{center}
\renewcommand{\arraystretch}{1.2}
\begin{tabular}{>{\raggedright\arraybackslash}p{0.28\linewidth}>{\raggedright\arraybackslash}p{0.64\linewidth}}
\toprule
Geometric object & Possible QEC role \\
\midrule
Marked curve or punctured sphere & Spatial datum carrying conformal-block state spaces and braid histories. \\
Conformal-block bundle & Family of code or ambient spaces varying over moduli. \\
Projectively flat connection & Ideal transport, braid representation, or monodromy operation. \\
Sewing parameter & Local coordinate controlling factorization and sector likelihoods. \\
Boundary divisor in \(\overline M_{0,n}\) & Collision or degeneration locus where a local footprint sector becomes visible. \\
Higgs field & Semiclassical field configuration or connection-side limit. \\
Spectral curve & Abelianized cover carrying eigenvalue and line-bundle data. \\
Discriminant & Local degeneration data; possible algebraic footprint of a perturbation. \\
Jacobian or Prym variety & Abelian variety organizing line-bundle, theta-function, or phase-space data. \\
Polarization & Choice of quantization and possible distance/systolic structure. \\
Compactified Jacobian & Replacement for the smooth fibre over singular spectral curves. \\
Stokes data & Irregular analytic footprint of meromorphic singularities. \\
\bottomrule
\end{tabular}
\end{center}

The table also indicates where the programme is most vulnerable.  Each row becomes a QEC statement only after a physical Hilbert space, an encoding, an allowed error family, and a measurement model are supplied.  The geometric object alone is not the code.  Its role is to organize sector data and likelihoods in a way that may be useful for decoding.

\section{Outlook and problems}\label{sec:outlook-problems}

The formalism developed in the paper leaves a concrete research programme.  The following problems are intended to make that programme concrete.  They are grouped around the themes of the paper: categorical footprints, scalable architectures, conformal likelihoods, representation theory, and geometry.

\begin{problem}[Classification of syndrome-admissible footprint algebras]\label{prob:classify-syndrome-algebras}
Let \(\cC\) be a unitary fusion category and let \(\cH_L\subset\Hom(\one,X_1\otimes\cdots\otimes X_n)\) be a chosen code subspace.  Classify the commutative subalgebras generated by compatible total-charge projectors which are syndrome-admissible for a specified family of errors.  In particular, determine when a local measurement algebra separates error representatives without revealing logical information.
\end{problem}

\begin{problem}[Multiplicity-sensitive footprints]\label{prob:multiplicity-sensitive}
Develop a version of the footprint formalism that records fusion multiplicity data in a controlled way.  Determine which projectors in multiplicity spaces can be measured without violating the no-logical-information condition, and identify examples where coarse total-charge footprints are insufficient but multiplicity-refined footprints give a correctable syndrome algebra.
\end{problem}

\begin{problem}[From the six-\(\sigma\) code to growing syndrome families]\label{prob:six-sigma-growing}
Starting from Proposition~\ref{prop:six-sigma-footprint-correction}, construct growing Ising-anyon families in which fixed-charge subspaces provide redundant logical encodings and local pair- or block-charge measurements form syndrome-admissible commuting rounds.  Determine which physically local error generators are resolved by those rounds, characterize same-footprint residual ambiguities, and prove a distance or threshold statement for the resulting family.
\end{problem}

\begin{problem}[State-sum decoders]\label{prob:state-sum-decoders}
For a Levin--Wen or Turaev--Viro code associated to a unitary fusion category \(\cC\), construct a decoder whose local weights are built from the same categorical data as the code: admissibility rules, \(F\)-symbols, quantum dimensions, tube-algebra sectors, and boundary labels.  Compare the resulting decoder with conventional minimum-weight or tensor-network decoders.
\end{problem}

\begin{problem}[Verifying the Peierls hypotheses]\label{prob:footprint-distance-threshold}
For concrete families of string-net, anyonic, or conformal-block codes, make Definition~\ref{def:footprint-distance} and the Peierls hypotheses of Definition~\ref{def:peierls-footprint-family} explicit.  Determine when local neutralizability below scale \(L\) follows from a topological, categorical, or geometric distance, and prove componentwise balance for natural decoders beyond the binary chain case of Proposition~\ref{prop:minimum-weight-balance}.
\end{problem}

\begin{problem}[Spacetime footprint decoding]\label{prob:spacetime-footprint-decoding}
Develop a spacetime decoding theory for time-dependent footprint algebras \(\cA_t\).  Treat data errors, measurement errors, recoupling faults, braid faults, and domain-wall faults as components of a single labelled spacetime defect history.  Identify conditions under which repeated noisy measurements yield a reliable recovery protocol.
\end{problem}

\begin{problem}[Condensation and dynamically changing syndrome algebras]\label{prob:condensation-syndrome}
Describe how anyon condensation, gapped boundaries, and domain walls transform footprint projectors.  Given a spacetime sequence of condensable algebras or module categories, determine the induced transformation of measured sectors and formulate a Knill--Laflamme-type condition for the resulting dynamical code deformation.
\end{problem}

\begin{problem}[Conformal-block likelihoods]\label{prob:conformal-likelihoods}
For conformal-block codes over \(M_{0,n}\), define likelihood models using Hermitian metrics, projectively flat connections, and sewing asymptotics.  Determine how the log-likelihood ratio between footprint sectors varies with cross-ratios and how this variation affects decoding.
\end{problem}

\begin{problem}[Hyperbolic field-theoretic architectures]\label{prob:hyperbolic-architectures}
Investigate whether hyperbolic surface-code and cluster-state constructions admit natural categorical or conformal enhancements.  In particular, determine whether hyperbolic holonomy, Fuchsian group representations, or moduli of hyperbolic structures can be used to refine footprint weights without compromising the operational simplicity of existing hyperbolic decoders.
\end{problem}

\begin{problem}[Yangian and transfer-matrix sector projectors]\label{prob:yangian-projectors}
Identify code families for which a Yangian, affine Yangian, quantum loop algebra, or related integrable structure acts on the relevant state spaces.  Determine whether commuting transfer matrices or Bethe subalgebras produce syndrome-admissible footprint algebras, and whether their spectra improve decoding or logical-gate design.
\end{problem}

\begin{problem}[Nonsemisimple scalar conditions]\label{prob:nonsemisimple-scalar}
Extend the categorical Knill--Laflamme argument beyond semisimple unitary fusion categories.  Determine the correct replacement for the scalar conclusion in nonsemisimple settings: for instance, whether composites should become scalar only after semisimplification, after quotienting negligible morphisms, or after evaluating modified traces on projective sectors.
\end{problem}

\begin{problem}[Higgs-bundle and spectral-curve decoders]\label{prob:higgs-decoder}
Construct explicit examples in which Higgs-bundle spectral data produce geometry-sensitive footprint weights.  In low-rank twisted Higgs-bundle models on \(\bP^1\), compute discriminants, singular spectral curves, compactified Jacobians, and candidate footprint fibres.  Determine whether these data can be connected to a concrete QEC recovery problem.
\end{problem}

\begin{problem}[Jacobians, theta functions, and GKP-type limits]\label{prob:jacobian-gkp}
Explore the relationship between spectral-curve Jacobians and GKP-type oscillator codes.  Determine whether theta-function spaces arising from Jacobians or Prym varieties can be interpreted as code spaces, whether theta-group actions realize logical Pauli operations, and whether footprint distance has a systolic interpretation on the underlying polarized abelian variety.
\end{problem}

\begin{problem}[Sharper and architecture-specific thresholds]\label{prob:field-theoretic-threshold}
Refine Theorem~\ref{thm:peierls-footprint-threshold} for specific architectures.  Optimize the constants for surface-code, string-net, Turaev--Viro, and measurement-based realizations; incorporate measurement noise and circuit-level faults; and compare minimum-weight, tensor-network, and state-sum decoders within the same footprint language.  A particularly important goal is to replace the abstract balance hypothesis by verifiable categorical or geometric criteria.
\end{problem}

These problems return to the distinction on which the paper is based.  Error correction is an inverse problem over histories compatible with measured local data, and the measurement model determines which part of the underlying field-theoretic footprint becomes classical syndrome information.  Topological stabilizer QEC is the most developed abelian instance of this pattern.  The open problem is to determine how far nonabelian, conformal, representation-theoretic, and algebro-geometric versions can be carried.

\medskip
\noindent\textbf{Acknowledgements.} The author is grateful to Mahmud Azam for questions and feedback during a reading of a draft version of the manuscript.  The author also acknowledges Maxence Mayrand for some useful comments.  The author is also grateful to the Abdus Salam International Centre for Theoretical Physics (ICTP) and the organizers of the workshop Strings \& Geometry 2025 for their hospitality and for a stimulating programme during which formative steps in this work were completed.  Finally, the author acknowledges the Natural Sciences and Engineering Research Council of Canada (NSERC) Discovery Grant program for partial support during the preparation of this manuscript.


\begin{thebibliography}{99}

\bibitem{KnillLaflamme1997}
E. Knill and R. Laflamme,
\emph{Theory of quantum error-correcting codes},
Phys. Rev. A \textbf{55} (1997), 900--911;
arXiv:quant-ph/9604034, doi:10.1103/PhysRevA.55.900.

\bibitem{Gottesman1997}
D. Gottesman,
\emph{Stabilizer codes and quantum error correction},
Ph.D. thesis, California Institute of Technology, 1997;
arXiv:quant-ph/9705052.

\bibitem{CalderbankShor1996}
A. R. Calderbank and P. W. Shor,
\emph{Good quantum error-correcting codes exist},
Phys. Rev. A \textbf{54} (1996), 1098--1105;
arXiv:quant-ph/9512032, doi:10.1103/PhysRevA.54.1098.

\bibitem{Steane1996}
A. M. Steane,
\emph{Error correcting codes in quantum theory},
Phys. Rev. Lett. \textbf{77} (1996), 793--797;
doi:10.1103/PhysRevLett.77.793.

\bibitem{KribsLaflammePoulin2005}
D. Kribs, R. Laflamme, and D. Poulin,
\emph{A unified and generalized approach to quantum error correction},
Phys. Rev. Lett. \textbf{94} (2005), 180501;
arXiv:quant-ph/0412076, doi:10.1103/PhysRevLett.94.180501.


\bibitem{BlumeKohoutNgPoulinViola2010}
R. Blume-Kohout, H. K. Ng, D. Poulin, and L. Viola,
\emph{Information preserving structures: A general framework for quantum zero-error information},
Phys. Rev. A \textbf{82} (2010), 062306;
arXiv:1006.1358, doi:10.1103/PhysRevA.82.062306.

\bibitem{BravyiLeemhuisTerhal2010}
S. Bravyi, B. Leemhuis, and B. M. Terhal,
\emph{Majorana fermion codes},
New J. Phys. \textbf{12} (2010), 083039;
arXiv:1004.3791, doi:10.1088/1367-2630/12/8/083039.

\bibitem{BondersonGurarieNayak2011}
P. Bonderson, V. Gurarie, and C. Nayak,
\emph{Plasma analogy and non-Abelian statistics for Ising-type quantum Hall states},
Phys. Rev. B \textbf{83} (2011), 075303;
arXiv:1008.5194, doi:10.1103/PhysRevB.83.075303.

\bibitem{Gottesman2014ConstantOverhead}
D. Gottesman,
\emph{Fault-tolerant quantum computation with constant overhead},
Quantum Inf. Comput. \textbf{14} (2014), 1338--1371;
arXiv:1310.2984, doi:10.26421/QIC14.15-16-5.

\bibitem{Gidney2021Stim}
C. Gidney,
\emph{Stim: a fast stabilizer circuit simulator},
Quantum \textbf{5} (2021), 497;
arXiv:2103.02202, doi:10.22331/q-2021-07-06-497.

\bibitem{DerksEtAl2024}
P.-J. H. S. Derks, A. Townsend-Teague, A. G. Burchards, and J. Eisert,
\emph{Designing fault-tolerant circuits using detector error models},
arXiv:2407.13826 (2024).

\bibitem{BreuckmannTerhal2016}
N. P. Breuckmann and B. M. Terhal,
\emph{Constructions and noise threshold of hyperbolic surface codes},
IEEE Trans. Inf. Theory \textbf{62} (2016), 3731--3744;
arXiv:1506.04029, doi:10.1109/TIT.2016.2555700.

\bibitem{GuthLubotzky2014}
L. Guth and A. Lubotzky,
\emph{Quantum error-correcting codes and 4-dimensional arithmetic hyperbolic manifolds},
J. Math. Phys. \textbf{55} (2014), 082202;
arXiv:1310.5555, doi:10.1063/1.4891487.

\bibitem{EllisonEtAl2023}
T. D. Ellison, Y.-A. Chen, A. Dua, W. Shirley, N. Tantivasadakarn, and D. J. Williamson,
\emph{Pauli topological subsystem codes from Abelian anyon theories},
Quantum \textbf{7} (2023), 1137;
arXiv:2211.03798, doi:10.22331/q-2023-10-12-1137.

\bibitem{BondersonShtengelSlingerland2008}
P. Bonderson, K. Shtengel, and J. K. Slingerland,
\emph{Interferometry of non-Abelian anyons},
Ann. Phys. \textbf{323} (2008), 2709--2755;
arXiv:0707.4206, doi:10.1016/j.aop.2008.01.012.

\bibitem{HastingsHaah2021}
M. B. Hastings and J. Haah,
\emph{Dynamically generated logical qubits},
Quantum \textbf{5} (2021), 564;
arXiv:2107.02194, doi:10.22331/q-2021-10-19-564.

\bibitem{FerrisPoulin2014}
A. J. Ferris and D. Poulin,
\emph{Tensor networks and quantum error correction},
Phys. Rev. Lett. \textbf{113} (2014), 030501;
arXiv:1312.4578, doi:10.1103/PhysRevLett.113.030501.

\bibitem{BravyiSucharaVargo2014}
S. Bravyi, M. Suchara, and A. Vargo,
\emph{Efficient algorithms for maximum likelihood decoding in the surface code},
Phys. Rev. A \textbf{90} (2014), 032326;
arXiv:1405.4883, doi:10.1103/PhysRevA.90.032326.

\bibitem{BondersonFreedmanNayak2008}
P. Bonderson, M. Freedman, and C. Nayak,
\emph{Measurement-only topological quantum computation},
Phys. Rev. Lett. \textbf{101} (2008), 010501;
arXiv:0802.0279, doi:10.1103/PhysRevLett.101.010501.

\bibitem{NayakSimonSternFreedmanDasSarma2008}
C. Nayak, S. H. Simon, A. Stern, M. Freedman, and S. Das Sarma,
\emph{Non-Abelian anyons and topological quantum computation},
Rev. Mod. Phys. \textbf{80} (2008), 1083--1159;
arXiv:0707.1889, doi:10.1103/RevModPhys.80.1083.

\bibitem{BurtonBrellFlammia2017}
S. Burton, C. G. Brell, and S. T. Flammia,
\emph{Classical simulation of quantum error correction in a Fibonacci anyon code},
Phys. Rev. A \textbf{95} (2017), 022309;
arXiv:1506.03815, doi:10.1103/PhysRevA.95.022309.

\bibitem{WoottonHutter2016}
J. R. Wootton and A. Hutter,
\emph{Active error correction for Abelian and non-Abelian anyons},
Phys. Rev. A \textbf{93} (2016), 022318;
arXiv:1506.00524, doi:10.1103/PhysRevA.93.022318.

\bibitem{DauphinaisPoulin2017}
G. Dauphinais and D. Poulin,
\emph{Fault-tolerant quantum error correction for non-Abelian anyons},
Comm. Math. Phys. \textbf{355} (2017), 519--560;
arXiv:1607.02159, doi:10.1007/s00220-017-2923-9.

\bibitem{LyonsBrown2026}
A. Lyons and B. J. Brown,
\emph{Quantum computing with anyons is fault tolerant},
arXiv:2602.11258, 2026.

\bibitem{Kitaev2003}
A. Yu. Kitaev,
\emph{Fault-tolerant quantum computation by anyons},
Ann. Phys. \textbf{303} (2003), 2--30;
arXiv:quant-ph/9707021, doi:10.1016/S0003-4916(02)00018-0.

\bibitem{Kitaev2006}
A. Kitaev,
\emph{Anyons in an exactly solved model and beyond},
Ann. Phys. \textbf{321} (2006), 2--111;
arXiv:cond-mat/0506438, doi:10.1016/j.aop.2005.10.005.

\bibitem{DennisKitaevLandahlPreskill2002}
E. Dennis, A. Kitaev, A. Landahl, and J. Preskill,
\emph{Topological quantum memory},
J. Math. Phys. \textbf{43} (2002), 4452--4505;
arXiv:quant-ph/0110143, doi:10.1063/1.1499754.

\bibitem{AharonovBenOr2008}
D. Aharonov and M. Ben-Or,
\emph{Fault-tolerant quantum computation with constant error rate},
SIAM J. Comput. \textbf{38} (2008), 1207--1282;
arXiv:quant-ph/9906129, doi:10.1137/S0097539799359385.

\bibitem{BravyiTerhal2009}
S. Bravyi and B. Terhal,
\emph{A no-go theorem for a two-dimensional self-correcting quantum memory based on stabilizer codes},
New J. Phys. \textbf{11} (2009), 043029;
arXiv:0810.1983, doi:10.1088/1367-2630/11/4/043029.

\bibitem{LevinWen2005}
M. A. Levin and X.-G. Wen,
\emph{String-net condensation: A physical mechanism for topological phases},
Phys. Rev. B \textbf{71} (2005), 045110;
arXiv:cond-mat/0404617, doi:10.1103/PhysRevB.71.045110.

\bibitem{KoenigKuperbergReichardt2010}
R. Koenig, G. Kuperberg, and B. W. Reichardt,
\emph{Quantum computation with Turaev--Viro codes},
Ann. Phys. \textbf{325} (2010), 2707--2749;
arXiv:1002.2816, doi:10.1016/j.aop.2010.08.001.

\bibitem{SchotteZhuBurgelmanVerstraete2022}
A. Schotte, G. Zhu, L. Burgelman, and F. Verstraete,
\emph{Quantum error correction thresholds for the universal Fibonacci Turaev--Viro code},
Phys. Rev. X \textbf{12} (2022), 021012;
arXiv:2012.04610, doi:10.1103/PhysRevX.12.021012.

\bibitem{KesselringEtAl2024}
M. S. Kesselring, J. C. Magdalena de la Fuente, F. Thomsen, J. Eisert, S. D. Bartlett, and B. J. Brown,
\emph{Anyon condensation and the color code},
PRX Quantum \textbf{5} (2024), 010342;
arXiv:2212.00042, doi:10.1103/PRXQuantum.5.010342.

\bibitem{PastawskiYoshidaHarlowPreskill2015}
F. Pastawski, B. Yoshida, D. Harlow, and J. Preskill,
\emph{Holographic quantum error-correcting codes: Toy models for the bulk/boundary correspondence},
JHEP \textbf{2015} (2015), 149;
arXiv:1503.06237, doi:10.1007/JHEP06(2015)149.

\bibitem{Backens2014}
M. Backens,
\emph{The ZX-calculus is complete for stabilizer quantum mechanics},
New J. Phys. \textbf{16} (2014), 093021;
arXiv:1307.7025, doi:10.1088/1367-2630/16/9/093021.

\bibitem{CoeckeKissinger2017}
B. Coecke and A. Kissinger,
\emph{Picturing Quantum Processes: A First Course in Quantum Theory and Diagrammatic Reasoning},
Cambridge University Press, 2017.

\bibitem{deBeaudrapHorsman2020}
N. de Beaudrap and D. Horsman,
\emph{The ZX calculus is a language for surface code lattice surgery},
Quantum \textbf{4} (2020), 218;
arXiv:1704.08670, doi:10.22331/q-2020-01-09-218.

\bibitem{ChancellorEtAl2016}
N. Chancellor, A. Kissinger, J. Roffe, S. Zohren, and D. Horsman,
\emph{Graphical structures for design and verification of quantum error correction},
Quantum Sci. Technol. \textbf{8} (2023), 045028;
arXiv:1611.08012, doi:10.1088/2058-9565/acf157.

\bibitem{BombinLitinskiNickersonPastawskiRoberts2023}
H. Bombin, D. Litinski, N. Nickerson, F. Pastawski, and S. Roberts,
\emph{Unifying flavors of fault tolerance with the ZX calculus},
Quantum \textbf{8} (2024), 1379;
arXiv:2303.08829, doi:10.22331/q-2024-06-18-1379.


\bibitem{GottesmanKitaevPreskill2001}
D. Gottesman, A. Kitaev, and J. Preskill,
\emph{Encoding a qubit in an oscillator},
Phys. Rev. A \textbf{64} (2001), 012310;
arXiv:quant-ph/0008040, doi:10.1103/PhysRevA.64.012310.

\bibitem{RaussendorfHarringtonGoyal2007}
R. Raussendorf, J. Harrington, and K. Goyal,
\emph{Topological fault-tolerance in cluster state quantum computation},
New J. Phys. \textbf{9} (2007), 199;
arXiv:quant-ph/0703143, doi:10.1088/1367-2630/9/6/199.

\bibitem{NewmanDeCastroBrown2020}
M. Newman, L. A. de Castro, and K. R. Brown,
\emph{Generating fault-tolerant cluster states from crystal structures},
Quantum \textbf{4} (2020), 295;
arXiv:1909.11817, doi:10.22331/q-2020-07-13-295.

\bibitem{MooreRead1991}
G. Moore and N. Read,
\emph{Nonabelions in the fractional quantum Hall effect},
Nucl. Phys. B \textbf{360} (1991), 362--396;
doi:10.1016/0550-3213(91)90407-O.

\bibitem{ReadMoore1992}
N. Read and G. Moore,
\emph{Fractional quantum Hall effect and nonabelian statistics},
Prog. Theor. Phys. Suppl. \textbf{107} (1992), 157--166;
arXiv:hep-th/9202001, doi:10.1143/PTPS.107.157.

\bibitem{KnizhnikZamolodchikov1984}
V. G. Knizhnik and A. B. Zamolodchikov,
\emph{Current algebra and Wess--Zumino model in two dimensions},
Nucl. Phys. B \textbf{247} (1984), 83--103;
doi:10.1016/0550-3213(84)90374-2.

\bibitem{TsuchiyaUenoYamada1989}
A. Tsuchiya, K. Ueno, and Y. Yamada,
\emph{Conformal field theory on universal family of stable curves with gauge symmetries},
Adv. Stud. Pure Math. \textbf{19} (1989), 459--566;
doi:10.2969/aspm/01910459.

\bibitem{BeauvilleLaszlo1994}
A. Beauville and Y. Laszlo,
\emph{Conformal blocks and generalized theta functions},
Comm. Math. Phys. \textbf{164} (1994), 385--419;
arXiv:alg-geom/9309003.

\bibitem{TuraevViro1992}
V. G. Turaev and O. Y. Viro,
\emph{State sum invariants of 3-manifolds and quantum 6j-symbols},
Topology \textbf{31} (1992), 865--902;
doi:10.1016/0040-9383(92)90015-A.

\bibitem{ReshetikhinTuraev1991}
N. Reshetikhin and V. G. Turaev,
\emph{Invariants of 3-manifolds via link polynomials and quantum groups},
Invent. Math. \textbf{103} (1991), 547--597;
doi:10.1007/BF01239527.

\bibitem{Turaev1994}
V. G. Turaev,
\emph{Quantum Invariants of Knots and 3-Manifolds},
De Gruyter Studies in Mathematics, vol. 18, Walter de Gruyter, 1994.

\bibitem{BakalovKirillov2001}
B. Bakalov and A. Kirillov Jr.,
\emph{Lectures on Tensor Categories and Modular Functors},
University Lecture Series, vol. 21, American Mathematical Society, 2001.



\bibitem{Mueger2003}
M. M\"uger,
\emph{From subfactors to categories and topology II: The quantum double of tensor categories and subfactors},
J. Pure Appl. Algebra \textbf{180} (2003), 159--219;
arXiv:math/0111205, doi:10.1016/S0022-4049(02)00248-7.

\bibitem{Drinfeld1985}
V. G. Drinfeld,
\emph{Hopf algebras and the quantum Yang--Baxter equation},
Soviet Math. Dokl. \textbf{32} (1985), no. 1, 254--258.

\bibitem{Drinfeld1986}
V. G. Drinfeld,
\emph{Quantum groups},
in \emph{Proceedings of the International Congress of Mathematicians, Berkeley, 1986},
vol. 1, American Mathematical Society, Providence, RI, 1987, pp. 798--820.

\bibitem{Jimbo1985}
M. Jimbo,
\emph{A \(q\)-difference analogue of \(U(\mathfrak g)\) and the Yang--Baxter equation},
Lett. Math. Phys. \textbf{10} (1985), 63--69;
doi:10.1007/BF00704588.

\bibitem{ChariPressley1994}
V. Chari and A. Pressley,
\emph{A Guide to Quantum Groups},
Cambridge University Press, 1994.

\bibitem{Molev2007}
A. Molev,
\emph{Yangians and Classical Lie Algebras},
Mathematical Surveys and Monographs, vol. 143, American Mathematical Society, 2007.

\bibitem{MaulikOkounkov2019}
D. Maulik and A. Okounkov,
\emph{Quantum groups and quantum cohomology},
Ast\'erisque No. \textbf{408} (2019), ix+209 pp.;
arXiv:1211.1287, doi:10.24033/ast.1074.

\bibitem{Tsymbaliuk2017}
A. Tsymbaliuk,
\emph{The affine Yangian of \(\mathfrak{gl}_1\) revisited},
Adv. Math. \textbf{304} (2017), 583--645;
arXiv:1404.5240, doi:10.1016/j.aim.2016.08.041.

\bibitem{Kassel1995}
C. Kassel,
\emph{Quantum Groups},
Graduate Texts in Mathematics, vol. 155, Springer, 1995;
doi:10.1007/978-1-4612-0783-2.


\bibitem{Boalch2001}
P. Boalch,
\emph{Symplectic manifolds and isomonodromic deformations},
Adv. Math. \textbf{163} (2001), 137--205;
doi:10.1006/aima.2001.1998.

\bibitem{Hitchin1987}
N. J. Hitchin,
\emph{Stable bundles and integrable systems},
Duke Math. J. \textbf{54} (1987), 91--114;
doi:10.1215/S0012-7094-87-05408-1.

\bibitem{BeauvilleNarasimhanRamanan1989}
A. Beauville, M. S. Narasimhan, and S. Ramanan,
\emph{Spectral curves and the generalised theta divisor},
J. Reine Angew. Math. \textbf{398} (1989), 169--179;
doi:10.1515/crll.1989.398.169.

\bibitem{Simpson1992}
C. T. Simpson,
\emph{Higgs bundles and local systems},
Publ. Math. Inst. Hautes Etudes Sci. \textbf{75} (1992), 5--95;
doi:10.1007/BF02699491.

\bibitem{Rayan2013NYJM}
S. Rayan,
\emph{Co-Higgs bundles on \(\mathbb P^1\)},
New York J. Math. \textbf{19} (2013), 925--945;
arXiv:1010.2526.

\bibitem{RayanSundbo2018}
S. Rayan and E. Sundbo,
\emph{Twisted argyle quivers and Higgs bundles},
Bull. Sci. Math. \textbf{146} (2018), 1--32;
arXiv:1803.04531, doi:10.1016/j.bulsci.2018.03.003.

\bibitem{Rayan2018SIGMA}
S. Rayan,
\emph{Aspects of the topology and combinatorics of Higgs bundle moduli spaces},
SIGMA \textbf{14} (2018), Paper No. 129, 18 pp.;
arXiv:1809.05732, doi:10.3842/SIGMA.2018.129.

\bibitem{MayrandRoyer2026}
M. Mayrand and B. Royer,
\emph{Complex abelian varieties and quantum error correction: a mathematical framework for GKP codes},
arXiv:2605.28784, 2026.

\bibitem{AzamRayan2024}
M. Azam and S. Rayan,
\emph{TQFTs and quantum computing},
Bull. Sci. Math. \textbf{194} (2024), 103454;
arXiv:2210.03556, doi:10.1016/j.bulsci.2024.103454.

\bibitem{MahmoudAliRayan2026}
A. A. Mahmoud, K. M. Ali, and S. Rayan,
\emph{Systematic approach to hyperbolic quantum error correction codes},
Phys. Rev. A \textbf{113} (2026), 042426;
arXiv:2504.07800, doi:10.1103/95mp-w7kr.

\bibitem{MahmoudTournaireBachmannRayan2026}
A. A. Mahmoud, G. Tournaire, S. Bachmann, and S. Rayan,
\emph{Hyperbolic cluster states for fault-tolerant measurement-based quantum computing},
arXiv:2603.27004, 2026.



\bibitem{Ocneanu1994Tube}
A. Ocneanu,
\emph{Chirality for operator algebras},
in \emph{Subfactors (Kyuzeso, 1993)},
H. Araki, Y. Kawahigashi, and H. Kosaki (eds.),
World Scientific, 1994, pp.~39--63.

\bibitem{KitaevKong2012}
A. Kitaev and L. Kong,
\emph{Models for gapped boundaries and domain walls},
Comm. Math. Phys. \textbf{313} (2012), 351--373;
arXiv:1104.5047, doi:10.1007/s00220-012-1500-5.

\bibitem{FreedmanLarsenWang2002}
M. H. Freedman, M. J. Larsen, and Z. Wang,
\emph{A modular functor which is universal for quantum computation},
Comm. Math. Phys. \textbf{227} (2002), 605--622;
arXiv:quant-ph/0001108, doi:10.1007/s002200200645.

\bibitem{AlmheiriDongHarlow2015}
A. Almheiri, X. Dong, and D. Harlow,
\emph{Bulk locality and quantum error correction in AdS/CFT},
J. High Energy Phys. \textbf{04} (2015), 163;
arXiv:1411.7041, doi:10.1007/JHEP04(2015)163.

\end{thebibliography}
\end{document}